\documentclass[11pt]{article}
\usepackage[latin1]{inputenc}
\usepackage{graphicx}
\usepackage{epsfig}
\usepackage{color}
\usepackage{fullpage}
\usepackage{lscape}
\usepackage{verbatim}
\usepackage{amsthm, amssymb}
\usepackage{amsmath}
\usepackage{xspace}
\newtheorem{Lem}{Lemma}
\newtheorem{theorem}{Theorem}
\newtheorem{Cor}{Corollary}

\newtheorem{Inv}{Invariant}
\def\polylog{\operatorname{polylog}}

\newcommand{\DSSSP}{\ensuremath{\mathcal D_{\mathit{SSSP}}}\xspace}
\newcommand{\DES}{\ensuremath{\mathcal{ES}}\xspace}
\newcommand{\DSCC}{\ensuremath{\mathcal D_{\mathit{SCC}}}\xspace}

\title{Decremental SSSP in Weighted Digraphs: Faster and Against an Adaptive Adversary}
\author{Maximilian Probst Gutenberg
        \footnote{Department of Computer Science,
                  University of Copenhagen,
                  \texttt{probst@di.ku.dk}}
        \and Christian Wulff-Nilsen
        \footnote{Department of Computer Science,
                  University of Copenhagen,
                  \texttt{koolooz@di.ku.dk},
                  \texttt{http://www.diku.dk/$_{\widetilde{~}}$koolooz/}. This research is supported by the Starting Grant 7027-00050B from the Independent Research Fund Denmark under the Sapere Aude research career programme.}}

\date{}
\begin{document}

\maketitle
\begin{abstract}
Given a dynamic digraph $G = (V,E)$ undergoing edge deletions and given $s\in V$ and constant $\epsilon$ with $0 < \epsilon\leq 1$, we consider the problem of maintaining $(1+\epsilon)$-approximate shortest path distances from $s$ to all vertices in $G$ over the sequence of deletions. Even and Shiloach (J.~ACM'$81$) give a deterministic data structure for the exact version of the problem in unweighted graphs with total update time $O(mn)$. Henzinger et al. (STOC'$14$, ICALP'$15$) give a Monte Carlo data structure for the approximate version with an improved total update time bound of $ O(mn^{0.9 + o(1)}\log W)$ with better bounds for sufficiently dense and sufficiently sparse graphs; here $W$ is the ratio between the largest and smallest edge weight. A drawback of their data structure and in fact of all previous randomized data structures is that they only work against an oblivious adversary, meaning that the sequence of deletions needs to be fixed in advance. This severely limits its application as a black box inside algorithms. We present the following $(1+\epsilon)$-approximate data structures:
\begin{enumerate}
\item the first data structure is Las Vegas and works against an adaptive adversary; it has total expected update time $\tilde O(m^{2/3}n^{4/3})$\footnote{Here, $\tilde O$ suppresses logarithmic factors so that $\tilde O(f(n)) = O(f(n)\polylog(f(n)))$; $\tilde\Omega$ and $\tilde\Theta$ are defined similarly.}  for unweighted graphs and $\tilde O(m^{3/4}n^{5/4}\log W)$ for weighted graphs,
\item the second data structure is Las Vegas and assumes an oblivious adversary; it has total expected update time $\tilde O(\sqrt m n^{3/2})$ for unweighted graphs and $\tilde O(m^{2/3}n^{4/3}\log W)$ for weighted graphs,
\item the third data structure is Monte Carlo and is correct w.h.p.~against an oblivious adversary; it has total expected update time $\tilde O((mn)^{7/8}\log W) = \tilde O(mn^{3/4}\log W)$.
\end{enumerate}
Each of our data structures can report the length of a $(1+\epsilon)$-approximate shortest path from $s$ to any query vertex in constant time at any point during the sequence of updates; if the adversary is oblivious, a query can be extended to also report such a path in time proportional to its length. Our update times are faster than those of Henzinger et al.~for all graph densities. For instance, when $m = \Theta(n^2)$, our second result improves their bound from $\tilde O(n^{2 + 3/4 + o(1)}\log W)$ to $\tilde O(n^{2 + 1/2})$ in the unweighted setting and to $\tilde O(n^{2 + 2/3}\log W)$ in the weighted setting. When $m = \Theta(n)$, our third result gives an improvement from $\tilde O(n^{1+5/6+o(1)}\log W)$ to $\tilde O(n^{1+3/4}\log W)$. Furthermore, our first data structure is the first to improve on the $O(mn)$ bound of Even and Shiloach for all but the sparsest graphs while still working against an adaptive adversary and works even in weighted graphs; this answers an open problem by Henzinger et al.
\end{abstract}
\newpage

\section{Introduction}\label{sec:Intro}
Computing shortest paths is a classical algorithmic problems dating back to the $1950$'s. A classical algorithm like BFS (breadth-first search) solves the single-source variant in linear time for unweighted graphs and Dijkstra's algorithm solves it in near-linear time for graphs with non-negative edge weights.

Maintaining shortest paths in a dynamic graph has also received attention for decades. The classical result of Even and Shiloach~\cite{EvenS81} from $1981$ states that there is a deterministic data structure which maintains a BFS tree from a given source vertex $s$ under edge deletions in the underlying graph with total update time $O(mn)$ where $m$ resp.~$n$ is the number of edges resp.~vertices. The BFS tree is maintained explicitly so that at any point during the sequence of updates, the shortest path distance from $s$ to a query vertex can be answered in constant time and the corresponding path can be reported in time proportional to its length.

Even and Shiloach assumed the graph to be undirected and unweighted. Henzinger and King~\cite{HenzingerK95} and King~\cite{King99} generalized this result to directed graphs with integer weights. The total update is $O(mD)$ where $D$ is the largest finite distance from $s$ in any of the graphs obtained during the sequence of deletions.

Focusing on directed unweighted graphs, the $O(mn)$ bound stood until a breakthrough result in $2014$ by Henzinger et al.~\cite{HenzingerKN142}. They obtained a randomized Monte Carlo bound of $\tilde O(mn^{0.984})$ and in a subsequent paper~\cite{HenzingerKN15}, they improved the bound to $O(\min\{m^{7/6}n^{2/3 + o(1)}, m^{3/4}n^{5/4 + o(1)}\})$. This is $O(mn^{0.9 + o(1)})$ for all $m$, is $O(1 + 5/6 + o(1))$ for $m = \Theta(n)$, and is $O(n^{2+3/4 + o(1)})$ for $m = \Theta(n^2)$. At the cost of a factor of $\log W$ in the running time, they generalized the result to weighted graphs where $W$ is the ratio between the largest and smallest edge weight.

For undirected unweighted graphs, further improvements over~\cite{EvenS81} exist. Bernstein and Roditty~\cite{BernsteinR11} showed a total update time bound of $O(n^{2+O(1/\sqrt{\log n})})$. Henzinger et al.~\cite{HenzingerKN141} improved this to $O(n^{1.8 + o(1)} + m^{1+o(1)})$ and later to a near-linear bound of $O(m^{1+o(1)})$~\cite{HenzingerKN143}; their result extends to weighted graphs at the cost of a factor of $\log W$ in the total update time. These improvements are all randomized. Chechik and Bernstein obtained deterministic bounds of $\tilde O(n^2)$~\cite{BernsteinC16} and of $\tilde O(n^{5/4}\sqrt m)  = \tilde O(mn^{3/4})$~\cite{BernsteinC17} for unweighted undirected graphs.

All the improvements over~\cite{EvenS81} mentioned above (excluding the generalization by King) maintain $(1+\epsilon)$-approximate distances rather than exact distances. A result by Roditty and Zwick~\cite{RodittyZ11} suggests that this is necessary since breaking the $O(mn)$ bound while maintaining exact distances would lead to major breakthroughs for, e.g., Boolean matrix multiplication. Henzinger et al.~\cite{HenzingerKNS15} later showed that such a result would give a truly subcubic time algorithm for online Boolean matrix-vector multiplication which again would be a major breakthrough. This suggests that, in order to break the $O(mn)$ bound, approximation is necessary.

The more restricted problem of maintaining reachability from a given source vertex to all vertices of a graph undergoing edge deletions has also been studied. Henzinger et al.~\cite{HenzingerKN15} gave a bound of $\tilde O(\min\{m^{7/6}n^{2/3},m^{3/4}n^{5/4+o(1)},m^{2/3}n^{4/3+o(1)}+m^{3/7}n^{12/7+o(1)}\})$ which is faster than their SSSP result for dense graphs. Significant progress was made by Chechik et al.~\cite{ChechikHILP16}. They showed how to obtain a total expected update time bound of $\tilde O(m\sqrt n)$ with constant query time.

\subsection{Our results}
In this paper, we focus on directed graphs. A limitation of the data structure of Henzinger et al.~\cite{HenzingerKN142}, and in fact of every randomized data structure referred to above, is that it assumes an oblivious adversary which fixes the sequence of updates in advance. This is in contrast to an adaptive adversary which is allowed to perform updates based on answers to previous distance queries. Several papers apply the data structure of Even and Shiloach as a black box inside an algorithm which performs modifications to the underlying graph based on the distances reported by this structure. The oblivious adversary assumption means that the randomized data structures above cannot be plugged in instead as a black box since the algorithm acts as an adaptive adversary.

We improve on the result in~\cite{HenzingerKN142} in two ways. First, we present a data structure which is faster than both~\cite{EvenS81} and~\cite{HenzingerKN142} for dense graphs, which is \emph{Las Vegas} rather than Monte Carlo and which works against an \emph{adaptive adversary}:
\begin{theorem}\label{Thm:SSSPAdaptive}
Let $G = (V,E)$ be a graph undergoing edge deletions by an adaptive adversary, let $s\in V$, and let $0 < \epsilon \leq 1$ be given. Then there is a data structure with total expected update time $\tilde O(m^{2/3}n^{4/3}/\epsilon^{2/3} + n^2/\epsilon^2)$ for unweighted graphs and $\tilde O((m^{3/4}n^{5/4}/\epsilon^{3/4} + n^2/\epsilon^2)\log W)$ for weighted graphs where $W$ is the ratio between the largest and smallest edge weight. At any point, when given any query vertex $u\in V$, the data structure outputs in $O(1)$ time a value $\tilde d_G(s,u)$ such that $d_G(s,u)\leq\tilde d_G(s,u)\leq (1+\epsilon)d_G(s,u)$.
\end{theorem}
We emphasize that this is the first result that works against an adaptive adversary and breaks the $O(mn)$ bound of Even and Shiloach for all but the sparsest graphs. Furthermore, this bound is broken even in the weighted setting. This answers an open problem stated in the follow-up work~\cite{HenzingerKN15a} to~\cite{HenzingerKN142,HenzingerKN15}.

It is relevant to point out that when we refer to an adaptive adversary, we allow it to make updates based on answers to previous queries; however, we do not allow it to somehow measure the time spent on handling individual updates and make further updates based on this information. We believe this is a fairly minor restriction since our motivation for allowing an adaptive adversary is to be able to employ the data structure as a black box inside an algorithm instead of the data structure of Even and Shiloach; as both data structures only give a guarantee on the total update time, it seems reasonable to assume that the choices made by the algorithm is independent of the time spent in individual updates.

Next, we provide a Las Vegas structure which is even faster for dense graphs:
\begin{theorem}\label{Thm:SSSPDense}
Let $G = (V,E)$ be a graph undergoing edge deletions by an oblivious adversary, let $s\in V$, and let $0 < \epsilon \leq 1$ be given. Then there is a data structure with total expected update time $\tilde O(\sqrt mn^{3/2}/\epsilon^{3/2})$ for unweighted graphs and $\tilde O(m^{2/3}n^{4/3}\log W/\epsilon^{5/3})$ for weighted graphs where $W$ is the ratio between the largest and smallest edge weight. At any point, when given any query vertex $u\in V$, the data structure outputs in $O(1)$ time a value $\tilde d_G(s,u)$ such that $d_G(s,u)\leq\tilde d_G(s,u)\leq (1+\epsilon)d_G(s,u)$. The data structure can also report a $(1+\epsilon)$-approximate path from $s$ to $u$ in time proportional to its length.
\end{theorem}

Finally, we present a Monte Carlo structure which is faster than~\cite{HenzingerKN142} for sparse graphs:
\begin{theorem}\label{Thm:SSSPSparse}
Let $G = (V,E)$ be a graph undergoing edge deletions by an oblivious adversary, let $s\in V$, and let $0 < \epsilon \leq 1$ be given. Then there is a data structure with total expected update time $\tilde O((mn)^{7/8}\log W/\epsilon^{3/4})$ where $W$ is the ratio between the largest and smallest edge weight. At any point, when given any query vertex $u\in V$, the data structure outputs in $O(1)$ time a value $\tilde d_G(s,u)$ such that $d_G(s,u)\leq\tilde d_G(s,u)$ and such that w.h.p., $\tilde d_G(s,u)\leq (1+\epsilon)d_G(s,u)$. The data structure can also report a $(1+\epsilon)$-approximate path from $s$ to $u$ in time proportional to its length.
\end{theorem}

Together, our results improve on the running time of~\cite{HenzingerKN142} for all graph densities. For instance, when $m = \Theta(n^2)$, we improve their bound from $\tilde O(n^{2+3/4+o(1)}\log W)$ to $\tilde O(n^{2+1/2})$ for unweighted graphs and to $\tilde O(n^{2+2/3}\log W)$ for weighted graphs; our bound for unweighted graphs in fact matches (up to logarithmic factors) the bound for decremental reachability in~\cite{ChechikHILP16}. When $m = \Theta(n)$, we get an improvement from the $\tilde O(n^{1+5/6+o(1)}\log W)$ bound of~\cite{HenzingerKN142} to $\tilde O(n^{1+3/4}\log W)$.

\section{Overview of Techniques}\label{sec:HighLevel}
All three of our data structures fit within the same overall framework. In this section, we give a high-level overview of this framework and explain how to obtain our results within this framework without going into details. In order to avoid too many technical details in this section, some of the calculations below are not quite accurate.

\paragraph{Maintaining low-diameter SCCs:}
Let $G = (V,E)$ be the decremental graph and let $n$ denote the number of vertices of $G$. The main goal is to maintain the SSSP tree in a different graph with some properties that make it easier to maintain the tree, more specifically in a graph with properties similar to those in a directed acyclic graph (DAG). An obvious first attempt might be to simply maintain an SSSP tree in the DAG obtained from $G$ by contracting its strongly-connected components (SCCs) and removing self-loops. This obviously fails since the information about distances between vertices of the same SCC is lost, meaning that distances could be significantly underestimated by the data structure.

This leads to a more refined attempt: split the SCCs into smaller strongly-connected subgraphs each of which has small diameter and then contract these. Using separators obtained from sparse BFS layers, it is easy to show the existence of a set $S$ such that $|S| = \tilde O(n/d)$ and such that all SCCs of $G\setminus E(S)$ have diameter at most $d$, for any chosen parameter $d > 0$; here $E(S)$ denotes the set of edges of $G$ incident to $S$. Chechik et al.~\cite{ChechikHILP16} describe a data structure that maintains such a decomposition efficiently under edge deletions to $G$.

Unfortunately, this does not resolve our issue above of significantly underestimating distances. Suppose for instance that $G$ has $\Theta(n)$ SCCs each of $\Theta(1)$ vertex size and diameter $1$ and $S = \emptyset$. Contracting these SCCs might reduce the length of a shortest path by a factor of $2$ if every second edge is internal to an SCC of $G\setminus E(S) = G$. We aim for an approximation factor of only $1+\epsilon$.

This naturally leads to the third attempt: let the diameter threshold $d$ for an SCC be proportional to its vertex size. Indeed, our data structure maintains a set $S$ such that each SCC of $G\setminus E(S)$ of vertex size roughly $n/2^i$ has diameter at most $d/2^i$ for some suitable value $d$. It it fairly easy to see that with this property, any shortest path in $G$ visits no more than order $d$ edges of the SCCs of $G\setminus E(S)$. Since we only need to focus on long shortest paths, i.e., sufficiently longer than $d$ (as shorter paths can be efficiently maintained by the Even-Shiloach data structure), our data structure can thus safely work on the graph obtained from $G\setminus E(S)$ by contracting its SCCs.

\paragraph{A hierarchy of SCCs:}
To maintain this decomposition efficiently, our data structure maintains the SCCs in a hierarchy of $\lg n$ levels where the $i$th level is responsible for splitting SCCs of size roughly $n/2^i$ when their diameter threshold $d/2^i$ is exceeded.

An important property of this hierarchical structure is that on level $i$, the number of vertices added to $S$ is $O(n/(d/2^i)) = O(2^in/d)$ in total; this follows from our observations above and from the fact that the diameter threshold is $d/2^i$. Hence, the lower the level, the smaller the number of vertices added to $S$. We will explain later in this overview why this property is useful.

\paragraph{Topological ordering and artificial edge weights:}
For now, let us focus on the simpler problem of maintaining an approximate SSSP tree in $G\setminus E(S)$. Consider a topological ordering of the SCCs of this graph. Contracting these SCCs (and removing self-loops), all edges of the resulting multigraph $M$ are forward edges. Again, since we are only interested in long shortest paths, we can afford a worse approximation of the unit weight of edges that go forward by a lot in the topological ordering since the number of these edges must be small. More precisely, the number of forward edges skipping $k$ vertices of $V$ (i.e., skipping SCCs of total vertex size $k$) is at most $n/k$ so if we are approximating shortest paths of length roughly $D$, we can give such edges a weight up to $\max\{1,\Theta(\epsilon Dk/n)\}$.

\paragraph{A faster Even-Shiloach-type structure with weighted edges:}
We present an extension of the data structure of Even and Shiloach which works for multigraphs and, more importantly, only scans an edge of weight up to $w$ a total of $O(D/w)$ times; this is a factor of $w$ better than their structure which may scan an edge $D$ times. The way we ensure this is to allow the weight of such an edge $e = (u,v)$ to be ``flexible'': when the distance estimate associated with $u$ increases, the weight of $e$ shrinks accordingly so that the distance estimate associated with $v$ remains the same. Once the weight of $e$ has shrunk from $w$ down to $0$ due to increases in the estimate for $u$, $e$ is scanned and the estimate for $v$ then increases by $w$, thus resetting the weight of $e$ back to $w$. It follows that $e$ is scanned only $O(D/w)$ times.

Returning to our topological ordering and multigraph $M$, consider any vertex $v$ of this graph. For $k = 1,\ldots,n$, we have at most one forward edge starting in $v$ and skipping $k$ vertices of $V$ and we can assign a weight of $\Theta(Dk/n)$ to this edge. Similarly, we have at most one forward edge ending in $v$ and skipping $k$ vertices of $V$ and as argued above, we can assign a weight of $\Theta(\epsilon Dk/n)$ to this edge. Using our improved Even-Shiloach-type data-structure, the cost of scanning the edges incident to $v$ over all updates is only $O(D\sum_{k = 1}^n1/\max\{1,\epsilon Dk/n\}) = O(n/\epsilon + n\log n/\epsilon) = O(n\log n/\epsilon)$. In comparison, Even-Shiloach would use up to $\Theta(nD)$ time if the vertex degree is $\Theta(n)$. Hence, we get considerable speed-up for denser graphs.

\paragraph{Dealing with back edges:}
So far, we have simplified our problem by disregarding $E(S)$, allowing us to assume that all edges of $M$ are forward edges. Removing this restriction, that is, including the edges of $E(S)$ in $M$, we may introduce back edges. This seemingly makes all our arguments above break down. However, this is where we use the property stated earlier, that the number of vertices added to $S$ on level $i$ is $O(2^in/d)$ in total. Each back edge $e\in E(S)$ incident to such a vertex can only skip $O(n/2^i)$ vertices of $V$; this follows by observing that $e$ must be fully contained in an SCC of $G\setminus E(S)$ (since $e$ is a back edge) and this SCC has vertex size $O(n/2^i)$. Hence, for any simple path $P$ in $M$, the total number of vertices of $V$ skipped by back edges of $P$ is $O(\sum_{i=0}^{\lg n} (2^in/d)(n/2^i)) = O(n^2\log n/d)$ and so the total number of vertices skipped by forward edges is $O(n + n^2\log n/d) = O(n^2\log n/d)$ as well. Above when we ignored back edges, we could set the weight of an edge skipping $k$ vertices equal to $\max\{1,\Theta(Dk/n)\}$. Now, we instead use a weight of $\lceil k/\tau\rceil$ for a suitable $\tau > 0$; hence, edges skipping at most $\tau$ vertices are given (their correct) weight $1$, edges skipping more than $\tau$ and at most $2\tau$ vertices are given weight $2$, and so on. This way, we get an additive error in the approximation to any shortest path in $M$ of $O(n^2\log n/(d\tau))$; to get an approximation factor of $(1+\epsilon)$ for shortest paths of length at least $D$, we pick $\tau$ such that this additive error is no more than $\epsilon D$.

Answering a query for the approximate distance from $s$ to a vertex $v$ in $V$ is now done by reporting the weight of a shortest path in $M$ from the vertex containing $s$ to the vertex containing $v$ and adding $2d$ to this result; as argued earlier, adding $2d$ ensures that the output value does not underestimate the shortest path distance $d_G(s,v)$.

The above describes the overall framework and Theorem~\ref{Thm:SSSPDense} follows easily from it. Note that this theorem is a sparsification result in that it only gives an improvement over Even-Shiloach for denser graphs. This makes sense given our description above where only high-degree vertices are guaranteed to have a large number of incident edges of high weight.

\paragraph{Dealing with an adaptive adversary:}
Theorem~\ref{Thm:SSSPAdaptive} is likewise a sparsification result. The added challenge here is that it needs to handle an adaptive adversary. In the description above, we referred to the approach of Chechik et al.~\cite{ChechikHILP16} to maintain low-diameter SCCs which can be used more or less directly to get Theorem~\ref{Thm:SSSPDense}. This approach maintains, for each SCC $C$, a low-depth in-out-tree from a randomly chosen root vertex $r\in C$ in order to keep track of the diameter of $C$; more specifically, to keep the diameter bounded by $d$, the depth of both the in-tree and the out-tree are kept to be at most $d/2$. The approach is only efficient if each such root $r$ can be kept hidden from an adversary since otherwise, if the adversary keeps deleting the edges incident to the current root $r$, this forces the data structure to make several expensive rebuilds of the in-out-tree of $C$.

In~\cite{ChechikHILP16}, this is not an issue since only the SCCs of $G$ are revealed to the adversary and these reveal no information about the locations of the roots of in-out-trees. In the setting of Theorem~\ref{Thm:SSSPAdaptive}, the approximate distances reported may potentially reveal this information. We therefore need to modify the approach of~\cite{ChechikHILP16}.

The rough idea is the following. Suppose that whenever an SCC $C$ is split by separator vertices being added to $S$, new roots with corresponding in-out trees are computed for all SCCs that $C$ is partitioned into. This will ensure that our data structure works against an adaptive adversary since as soon as there is a risk of information about the root of $C$ being revealed, new roots are chosen. However, this approach is too slow since the partition of $C$ might be very unbalanced, requiring several rebuilds of in-out-trees. What we do instead is to maintain the in-out-tree of $C$ up to a distance threshold $d_1/2$ (more generally, distance threshold $d_1/2^{i+1}$ for SCCs at level $i$ of our hierarchical structure) but we delay the partition of $C$ until $d_2-d_1$ vertices have become unreachable in the in-out-tree, for some other threshold value $d_2 > d_1$. Up to this point, the diameter of $C$ is at most $d_2$; this follows since $C$ is strongly-connected so any shortest path in $C$ uses at most $2d_1/2 = d_1$ edges of the in-out-tree and at most $d_2 - d_1$ additional edges through the vertices not in the in-out-tree.

The advantage of this is that new roots and ES-trees are only computed once a chunk of $d_2 - d_1$ vertices are far from $r$ in $C$ in which case $C$ is partitioned into SCCs of more balanced size. This leads to faster update time and we show that this delayed partitioning ensures that an adaptive adversary cannot gain information about the roots of the in-out-trees. The disadvantage is that we get a worse trade-off between SCC diameters and the size of $S$, namely the SCC diameter can be up to $d_2$ while $|S| = \tilde O(n/d_1)$. This is why we get a worse time bound in Theorem~\ref{Thm:SSSPAdaptive} compared to Theorem~\ref{Thm:SSSPDense}.

In our description above, we assumed that $C$ can only break apart into smaller SCCs when vertices are added to $S$. However, $C$ may also break apart simply because edges are deleted from $G$ and hence from $G\setminus E(S)$. Fortunately, this case is easy to deal with since the partition of $C$ does not depend on the choice of random root, implying that no random bits are revealed to the adversary.

\paragraph{A speed-up for sparse graphs:}
The data structure of Theorem~\ref{Thm:SSSPSparse} gives a speed-up over Even-Shiloach also for sparse graphs. This does not fit directly into the sparsification framework sketched above so we need to modify it. Consider again the multigraph $M$. The idea is to randomly sample a subset of vertices of $M$ and maintain locally shortest paths between these sampled vertices. More specifically, for each sampled vertex $s$, we essentially keep an Even-Shiloach data structure with source $s$ for the subgraph $M(s)$ of $M$ induced by the vertices that are ``close'' to $s$ in the topological order, i.e., for every $v\in M(s)$, the number of vertices of $V$ between $s$ and $v$ in the topological ordering is at most some suitable value $\delta$ . For each sampled vertex $s'\neq s$ in $M(s)$, we add a super edge $(s,s')$ to our fast weighted version of Even-Shiloach described earlier and the weight of $(s,s')$ is the length of the shortest $s$-to-$s'$ path in $M(s)$.

In expectation, the total number of super-edges is small and their weights are high. From this, it follows that the total cost of scanning these over all updates is asymptotically less than $n^2$, improving on the bound of Even and Shiloach. However, not all shortest paths in $M$ consist of super-edges exclusively. Consider a shortest path $P$ and decompose it into maximal subpaths with no interior sampled vertices. For one such subpath $Q$, consider the interesting case where both its start point $s_1$ and endpoint $s_2$ are sampled (only the first and last subpath of $P$ does not have this property and their contribution to the distance approximation is negligible). If $s_1$ and $s_2$ are close together in the topological ordering of $M$, i.e., if $s_1\in M(s_2)$ or $s_2\in M(s_1)$ then there is a super-edge $(s_1,s_2)$ of weight roughly equal to $Q$. Otherwise, the average number of vertices of $V$ skipped by an edge of $Q$ in the topological ordering must be large; using similar arguments as before, we may assign a high weight to these edges of $Q$ which gives a speed-up using our weighted version of Even-Shiloach. Theorem~\ref{Thm:SSSPSparse} follows from these observations.

\paragraph{Decremental SSSP in weighted graphs:}
Above we considered only the unweighted setting. Extending to the case where $G$ is weighted is fairly straightforward since we have already introduced artificial edge weights above. The main difficulty is to generalize the data structure for maintaining low-diameter SCCs to the weighted setting. Recall that SCCs were split using sparse BFS layers that were added to $S$. For weighted graphs, this trick no longer works since BFS layers need not be separators.

To deal with this, we first give a fairly standard reduction, at the cost of a factor of $\log W$ in the update time, to the case where edge weights of $G$ are between $1$ and $n$ and where the largest distance allowed is order $n$. Edges of weight larger than some suitable threshold value $\rho\in (1,n)$ are easy to deal with in our weighted Even-Shiloach structure so we simply omit them in the contracted SCCs of $G\setminus E(S)$ and include them all in our multigraph $M$. The total cost of scanning these high-weight edges is only $O(mD/\rho)$.

Now, the data structure that maintains SCCs only needs to deal with edges of weight between $1$ and $\rho$. To separate an out-tree (equivalently, in-tree) with edge weights bounded by $\rho$, consider picking all vertices whose distances from the root in the out-tree belongs to some range of the form $[k\rho,(k+1)\rho]$. This set is a separator similar to a BFS layer in the unweighted setting: the reason is that there can be no edge $(u,v)$ where $u$ has distance less than $k\rho$ and $v$ has distance more than $(k+1)\rho$ from the root of the in-out-tree. In a sense, we consider layers of ``thickness'' $\rho$ rather than $1$. The downside of this is that we get a worse bound on $|S|$. However, we still get a speed-up for weighted graphs, as can be seen in our theorems.

This completes our high-level description. The rest of the paper is organized as follows. Section~\ref{sec:Prelim} gives some basic definitions and introduces notation that will be used throughout the paper. Section~\ref{sec:LowDiamDecomp} describes the algorithm that maintains a low-diameter decomposition which works against an adaptive adversary. Section~\ref{sec:Multigraph} presents a data structure for maintaining a multigraph under various updates. Then an Even-Shiloach-type structure is given in Section~\ref{sec:FastES} which maintains a shortest path tree of such a multigraph for a special edge weight function. These sections together describe the general framework for our data structures which are presented in Sections~\ref{sec:SSSPAdaptive},~\ref{sec:SSSPDenseObl} and~\ref{sec:SSSPSparse}, respectively. We only consider unweighted graphs in the main part of the paper; the extension to weighted graphs is described in the appendix. Finally, we make some concluding remarks in Section~\ref{sec:ConclRem}.

\section{Preliminaries}\label{sec:Prelim}
All graphs that we consider are directed. For a graph $G$ with edge weight function $w$, we denote its vertex set by $V(G)$ and edge set by $E(G)$. Let $V = V(G)$ and $E = E(G)$ in the following. An edge $e\in E$ is \emph{incident} to a vertex $v\in V$ if $v$ is one of the endpoints of $e$. The set of edges of $E$ incident to a vertex $v$ is denoted $E(v)$ and $|E(v)|$ is the \emph{degree} of $v$. The \emph{weight} $w(P)$ of a path $P$ in $G$ is the sum of weights of its edges and the \emph{length} $|P|$ of $P$ is its number of edges. For two vertices $u$ and $v$ in $V$, $d_G(u,v)$ denotes the shortest path distance in $G$ from $u$ to $v$ w.r.t.~edge-weight function $w$. Given an $r\in V$ and a $d\in\mathbb R_+$, we define $B_{\mathit{in}}(r,G,d) = \{v\in V\vert d_G(v,r)\leq d\}$ and $B_{\mathit{out}}(r,G,d) = \{v\in V\vert d_G(r,v)\leq d\}$.

A \emph{$q$-quality separator} of a graph $G = (V,E)$ is a set $S\subseteq V$ such that every SCC of $H\setminus E(S)$ contains at most $|V| - q|S|$ vertices.

We let $\lg n$ denote the base $2$-logarithm of $n$.

For a graph $G$, a $d\in\mathbb N$, and an $r\in V(G)$, an \emph{ES-structure} of $G$ for distance $d$ and root $r$ is a data structure $\mathcal E_r$ consisting of two instances of the data structure of Even and Shiloach (with the generalization to weighted directed graphs by~\cite{HenzingerK95,King99} when needed). Both structures are initialized for root $r$ and distance threshold $d$; one structure is initialized for $G$, the other for the graph $G^-$ obtained from $G$ by reversing the orientations of all edges. The former resp.~latter is referred to as the \emph{out-tree} resp.~\emph{in-tree} of $\mathcal E_r$. We require $\mathcal E_r$ to keep track of the number of vertices of $G$ unreachable from $r$ in the in-tree and the number of vertices unreachable from $r$ in the out-tree. It follows easily from the analysis in~\cite{HenzingerK95,King99} (see also~\cite{ChechikHILP16}) that $\mathcal E_r$ can be initialized and maintained over edge deletions using a total of $O(|E(G)| d)$ time such that a query for the number of vertices unreachable from the in-tree or the out-tree can be answered in $O(1)$ time.

\section{Maintaining a Low-Diameter Decomposition}\label{sec:LowDiamDecomp}
In this section, we consider a graph $G = (V,E)$ undergoing edge deletions and give a data structure that maintains a growing vertex set $S\subseteq V$ such that $G\setminus E(S)$ contains only SCCs of small diameter and such that $S$ is small. A similar result was shown by Chechik et al.~\cite{ChechikHILP16}. However, in their application, the sequence of edge deletions is independent of the random bits used since only the SCCs of $G$ are revealed to the adversary. Since our decremental SSSP structure will output approximate distances that may reveal the structure of SCCs of $G\setminus E(S)$, we have the added challenge of having to deal with an adaptive adversary. We show the following result.
\begin{theorem}\label{Thm:LowDiamDecomp}
Let $G = (V,E)$ be an unweighted graph undergoing edge deletions, let $m = |E|$ and $n = |V|$, and let integers $0 < d_1 < d_2\leq n$ be given with $d_2 - d_1\geq 2\lg n$. Then there is a Las Vegas data structure which maintains a pair $(S,\mathcal V)$ where $S\subseteq V$ is a growing set and where $\mathcal V$ is the family of vertex sets of the SCCs of $G\setminus E(S)$ such that at any point, all these SCCs have diameter at most $d_2$ and $|S| = \tilde{O}(n/d_1)$.

After the initialization step, the data structure outputs the initial pair $(S,\mathcal V)$. After each update, the data structure outputs the set $S'$ of new vertices of $S$ where $S'\subseteq V'$ for some $V'\in\mathcal V$. Additionally, it updates $\mathcal V$ by replacing at most one $V'\in\mathcal V$ by the vertex sets $W_1,\ldots,W_p$ of the new SCCs of $G\setminus E(S)$ where $|W_i|\leq\frac 1 2 |V'|$ for $i = 1,\ldots,p-1$. Pointers to $W_1,\ldots,W_p$ are returned.


The total expected time is $\tilde O(m\sqrt n + mn/d_1 + mnd_1/(d_2 - d_1))$ and the data structure works against an adaptive adversary.
\end{theorem}
The proof is in many ways similar to that in~\cite{ChechikHILP16} but with some important differences that enable our structure to deal with an adaptive adversary that at any point knows the SCCs of $G\setminus E(S)$. We have highlighted these differences in the overview section above. A detailed proof of the theorem can be found in Appendix~\ref{sec:LowDiamDecompProof}.

\subsection{A hierarchical decomposition}
We will not use Theorem~\ref{Thm:LowDiamDecomp} directly in our SSSP data structure but rather the following corollary which gives the hierarchical decomposition of SCCs that we referred to in our overview.
\begin{Cor}\label{Cor:LowDiamDecomp}
Let $G = (V,E)$ be an unweighted graph undergoing edge deletions, let $m = |E|$ and $n = |V|$, and let integers $0 < d_1 < d_2\leq n$ be given with $d_2 - d_1\geq 2\lg n$. Then there is a Las Vegas data structure which maintains pairwise disjoint growing subsets $S_0, S_1,\ldots, S_{\lceil\lg d_1\rceil}$ of $V$ and a family $\mathcal V$ of subsets of $V$ with the following properties. For $i = 0,\ldots,\lceil\lg d_1\rceil$, let $G_i = G\setminus(\cup_{j = 0}^i E(S_j))$. Then over all updates, $\mathcal V$ is the family of vertex sets of the SCCs of $G_{\lceil\lg d_1\rceil}$ and for $i = 0,\ldots,\lceil\lg d_1\rceil$,
\begin{enumerate}
\item each SCC of $G_i$ of vertex size at most $n/2^i$ has diameter at most $d_2/2^i$,
\item if $i > 0$, every vertex of $S_i$ belongs to an SCC of $G_{i-1}$ of vertex size at most $n/2^i$, and
\item $|S_i| = \tilde O(n2^i/d_1)$.
\end{enumerate}
Furthermore, the sum of diameters of all SCCs of $G_{\lceil\lg d_1\rceil}$ is at most $2d_2$.

After the initialization step, the data structure outputs the initial sets $S_0, S_1,\ldots, S_{\lceil\lg d_1\rceil}$ and pointers to the sets of $\mathcal V$. After each update, the data structure outputs the new vertices of $S_0,\ldots,S_{\lceil\lg d_1\rceil}$. Additionally, it updates $\mathcal V$ by replacing at most one $V'\in\mathcal V$ by the vertex sets $W_1,\ldots,W_p$ of the new SCCs of $G_{\lceil\lg d_1\rceil}$ where $|W_i|\leq\frac 1 2 |V'|$ for $i = 1,\ldots,p-1$. Pointers to both the old set $V'$ and to the new sets $W_1,\ldots,W_p$ are returned. 

The total expected time is $\tilde O(m\sqrt n + mn/d_1 + mnd_1/(d_2 - d_1))$ and the data structure works against an adaptive adversary.
\end{Cor}
\begin{proof}
We say that a subgraph of $G$ has level $i\in\mathbb N_0$ if it contains more than $n/2^{i+1}$ vertices and at most $n/2^i$ vertices.

The data structure will generate strongly connected subgraphs during its execution. At initialization, it first generates the SCCs of $G$ and sets $S_0 = S_1 = \ldots = S_{\lceil\lg d_1\rceil} = \emptyset$.

Whenever an SCC $C$ has been generated (either in one of the graphs $G_0,\ldots,G_{\lceil\lg d_1\rceil}$ or during the initialization step above), let $i$ be its level. If $i = \lceil\lg d_1\rceil$ then add all vertices of $C$ to $S_i$, thereby partitioning $C$ into single-vertex SCCs in $G_i$. Now assume that $i < \lceil\lg d_1\rceil$. If $C$ was not created due to a level $i$-SCC being partitioned, an instance of the data structure $\mathcal D_i(C)$ from Theorem~\ref{Thm:LowDiamDecomp} is initialized and maintained for $C$ with parameters $d_1/2^i$ and $d_2/2^i$ instead of $d_1$ and $d_2$, respectively. Whenever $\mathcal D_i(C)$ partitions a level $i$ SCC $C'$, the new separator vertices are added to $S_i$.

The data structure outputs new vertices of $S_0,\ldots,S_{\lceil\lg d_1\rceil}$ and pointers to $W_1,\ldots,W_p$ in the same way as the data structure of Theorem~\ref{Thm:LowDiamDecomp}.

Correctness of the three items follows from the correctness of Theorem~\ref{Thm:LowDiamDecomp}, from the fact that trivially, $|S_i| = \tilde O(n2^i/d_1)$ for $i = \lceil\lg d_1\rceil$, and from observing that whenever a data structure $\mathcal D_i(C)$ is initialized, $C$ is an SCC of $G_{i-1}$ of level $i$ and hence of size at most $n/2^i$.

To finish the correctness proof, we need to show that the sum of diameters of all SCCs of $G_{\lceil\lg d_1\rceil}$ is at most $2d_2$ after each edge deletion in $G$. Consider one such SCC $C$ and let $i$ be its level. If $i = \lceil\lg d_1\rceil$ then $C$ has diameter $0$ so assume that $i < \lceil\lg d_1\rceil$. Then there is a data structure that was previously initialized for some level $i$-SCC which ensures that $C$ has diameter at most $d_2/2^i$ in $G_i$ and thus in $G_{\lceil\lg d_1\rceil}$. Since $C$ has level $i$, we also have $|C| > n/2^{i+1}$. If we think of the diameter of $C$ as its cost and distribute this cost evenly among the vertices of $C$, each vertex of $C$ is assigned a cost of less than $(d_2/2^i)/(n/2^{i+1}) = 2d_2/n$. The sum of diameters of all SCCs of $G_{\lceil\lg d_1\rceil}$ equals the total cost assigned to all vertices of $V$ which is less than $n(2d_2/n) = 2d_2$, as desired.

For the time bound, keeping track of levels of SCCs can easily be done within the time spent on instances of data structures of Theorem~\ref{Thm:LowDiamDecomp}. Also, observe that the vertex sets of all instances of data structures of Theorem~\ref{Thm:LowDiamDecomp} initialized for level $i$-SCCs are pairwise vertex-disjoint. Since there are only $O(\log n)$ levels, the time bound follows from Theorem~\ref{Thm:LowDiamDecomp}.
\end{proof}

\section{A Multigraph Data Structure}\label{sec:Multigraph}
In this section, we present a data structure $\mathcal M$ for maintaining a  multigraph $M = (V_M, E_M)$ of an underlying decremental simple digraph $G = (V,E)$ where $M$ is obtained from $G$ by contracting subsets of pairwise disjoint subsets of $V$ and removing self-loops. Each edge $e\in E$ has a \emph{level} $\ell(e)\in\{0,\ldots,k\}$ for some given parameter $k\in\mathbb N_0$; this level may increase over time. If for two distinct vertices $u$ and $v$ in $V_M$ there are multiple edges in $E_M$ from $u$ to $v$, let $i$ be the minimum level of these. The \emph{representative} for this multi-edge is a single edge $(u,v)$ with the associated level $\ell(u,v) = i$. Hence, replacing the edges of $M$ by the representative edges yields a simple graph.

The data structure $\mathcal M$ supports the following operations:
\begin{description}
\item [\texttt{Init}$(G = (V,E), \{V_1,\ldots,V_{\ell}\},\{E_0,\ldots,E_k\},\{\Delta_0,\ldots,\Delta_k\})$:] initializes $M$ for the graph $G = (V,E)$ where $V_M = \{V_1,\ldots,V_{\ell}\}$ is a partition of $V$, $\{E_0,\ldots,E_k\}$ is a partition of $E$, and $\Delta_0,\ldots,\Delta_k$ belong to $\mathbb N_0$. For $i = 0,\ldots,k$, each edge of $E_i$ is assigned the level $i$.
\item [\texttt{Delete}$(e)$:] deletes edge $e$ from $E$ and updates $M$ accordingly.
\item [\texttt{Increase}$(e,i)$:] assuming $e\in E$ has level $\ell(e) < i$, updates $\ell(e)\leftarrow i$.
\item [\texttt{Split}$(V',\{W_1,\ldots,W_{p-1}\})$:] splits the subset $V'\in V_M$ of $V$ into pairwise disjoint subsets $W_1,\ldots,W_{p-1},W_p$ where $W_p = V'\setminus\cup_{i = 1}^{p-1} W_i$ and updates $M$ accordingly by replacing $V'$ in $V_M$ by $W_1,\ldots,W_p$; it is assumed that $|W_i|\leq\frac 1 2 |V'|$ for $i = 1,\ldots,p-1$. The new vertex $W_p$ is identified with $V'$.
\end{description}
For each type of operation, $\mathcal M$ outputs a pair of sets, $(E_{\mathit{old}},E_{\mathit{new}})$. Set $E_{\mathit{old}}$ resp.~$E_{\mathit{new}}$ consists of the representative edges that disappear resp.~appear due to the update. If an edge changes level, we assume that it appears in $E_{\mathit{old}}$ resp.~$E_{\mathit{new}}$ with the old resp.~new level.

In addition, $\mathcal M$ in addition provides constant-time access to:
\begin{itemize}
\item the vertex of $V_M$ containing a given query vertex of $V$,
\item $\ell(e)$ for a given query edge $e\in E$,
\item pointers to linked lists $E_{\mathit{in}}(V',i)$ and $E_{\mathit{out}}(V',i)$ for a given query pair $(V',i)$ where $V'\in V_M$ and $i\in\{0,\ldots,k\}$; these lists contain, respectively, the representative level $i$-edges that are ingoing to and outgoing from $V'$, and
\item pointers to linked lists $V_{\mathit{in}}^{\Delta_i}(i)$ and $V_{\mathit{out}}^{\Delta_i}(i)$ for $i = 0,\ldots,k$; these lists contain, respectively, the set of vertices $V'$ of $M$ such that more than $\Delta_i$ representative level $i$-edges are ingoing to and outgoing from $V'$.
\end{itemize}

When convenient, we will view each vertex of $M$ as the subset of $V$ that was contracted to form this vertex.

We can implement $\mathcal M$ with the following performance guarantees (a proof can be found in Appendix~\ref{sec:ImplementationM}).
\begin{Lem}\label{Lem:M}
Let $m$ resp.~$n$ be the initial number of edges resp.~vertices of $G$. Then $\mathcal M$ can be implemented to support any sequence of operations using a total of $O(km\log n + m\log^2n + n)$ deterministic time.
\end{Lem}

\section{Fast Approximate Edge-Weighted ES-Trees}\label{sec:FastES}
In this section, we consider an instance $\mathcal M$ of the multigraph structure of Section~\ref{sec:Multigraph}. We shall ignore the underlying simple graph $G$ and only focus on the multigraph $M$ maintained by $\mathcal M$ and the changes made to $M$. We will present a data structure $\mathcal ES$ associated with $M$ which maintains an exact SSSP in $M$ from a given source $s\in V(M)$ up to a given threshold distance $D\in\mathbb N_0$ for a particular edge weight function $w:E_M\rightarrow\mathbb N_0$ specified below; this is the ``flexible'' edge weight function that we referred to in the overview. We assume that the vertex $s$ is never split by an update to $\mathcal M$.

The behavior and performance of the data structure will depend on two functions, $w_{M},W_{M}:\{0,\ldots,k\}\rightarrow \mathbb N$ where $w_{M} \leq W_{M}$ and where $k$ is as in Section~\ref{sec:Multigraph}. We leave these two functions unspecified in this section since the choice of these will depend on the applications in later sections. For any real number $x$ and $y\in\mathbb N$, denote by $\lceil x\rceil_i$ the smallest integer of value at least $x + w_{M}(i)$ which is divisible by $1 + W_{M}(i) - w_{M}(i)$; here we abuse notation and omit the functions $w_{M}$ and $W_{M}$ in the notation $\lceil x\rceil_i$.

The edge weight function $w:E_M\rightarrow \mathbb N$ is recursively defined by
\[
w(u,v) = \lceil d_M(s,u)\rceil_i - d_M(s,u),
\]
where $i = \ell(u,v)$ and $d_M(s,u)$ is the shortest path distance function in $M$ w.r.t.~$w$.

\begin{Lem}\label{Lem:EdgeWeightFunction}
The edge weight function $w$ above exists and is unique for all edges $(u,v)\in E_M$ where $u$ (and hence $v$) is reachable from $s$ in $M$. Furthermore, for each $(u,v)\in E_M$, $w_{M}(i)\leq w(u,v)\leq W_{M}(i)$ where $i = \ell(u,v)$.
\end{Lem}
\begin{proof}
If $w$ exists, it clearly satisfies $d_M(s,s) = 0$. Now, consider running Dijkstra's algorithm on $M$ with the following modification: initially, all edge weights are unspecified. Whenever a vertex $u$ is extracted from the min-priority queue, the weight $w(u,v)$ of each outgoing edge $(u,v)$ in $M$ is set to $\lceil d(u)\rceil_i$ where $i = \ell(u,v)$ and where $d(u)$ is the distance estimate for $u$ when it is extracted.

For every edge $(u,v)\in E(M)$ where $u$ is reachable from $s$ in $G$, the weight $w(u,v)$ is specified at termination of Dijkstra's algorithm. The first part of the lemma now follows from the correctness proof of Dijkstra's algorithm.

To show the second part, let $i\in\{0,\ldots,k\}$ be given and let $f_i:\mathbb N_0\rightarrow\mathbb N$ be defined by $f_i(x) = \lceil x\rceil_i - x$. We will show that $w_{M}(i)\leq f_i\leq W_{M}(i)$. This is clear if $w_M(i) = W_M(i)$ since then $f_i = w_M(i)$ so assume that $w_M(i) < W_M(i)$ and consider an $x\in\mathbb N_0$ such that $x + w_{M}(i)$ is divisible by $1 + W_{M}(i) - w_{M}(i)$. Then $f_i(x) = f_i(x + 1 + W_{M}(i) - w_{M}(i)) = w_{M}(i)$. Since $w_{M}(i) < W_{M}(i)$, we have $1 + W_{M}(i) - w_{M}(i)\geq 2$ so $(x + 1) + w_{M}(i)$ is not divisible by $1 + W_{M}(i) - w_{M}(i)$. Hence, $f_i(x+1) = (x + w_{M}(i) + (1 + W_{M}(i) - w_{M}(i))) - (x+1) = W_{M}(i)$. Furthermore, $\lceil x + c\rceil_i = \lceil x + 1\rceil_i$ for $c\in\{2,3,\ldots,W_{M}(i) - w_{M}(i)\}$ so $f_i(x+c) = W_{M}(i) - c + 1\geq w_{M}(i) + 1$. Hence, $w_{M}(i)\leq f_i\leq W_{M}(i)$ and the second part of the lemma follows.
\end{proof}

\subsection{The data structure}\label{subsec:ES}
We now present our dynamic data structure $\mathcal ES$ which maintains an SSSP tree $T$ from $s$ in $M$ up to distance $D$ w.r.t.~weight function $w$. The total update time is smaller than that of the data structure of Even and Shiloach~\cite{EvenS81} under certain assumptions about the distribution of edge levels. Tree $T$ is maintained under the types of updates to $\mathcal M$ specified in Section~\ref{sec:Multigraph}. The intuition behind the speed-up is that for any given edge $(u,v)$, when $d_M(s,u)$ increases, $(u,v)$ only needs to be scanned if $d_M(s,u)+w_M(i)$ is divisible by $1+W_M(i) - w_M(i)$ where $i = \ell(u,v)$. Hence, if $1+W_M(i) - w_M(i)$ is large, this will give a significant improvement over the data structure in~\cite{EvenS81} where a vertex pays its degree every time its distance from $s$ increases.


$\mathcal ES$ is a modification of the data structure of Even and Shiloach. When $\mathcal M$ has been initialized, the initial SSSP tree $T$ is found using the Dijkstra variant in the proof of Lemma~\ref{Lem:EdgeWeightFunction} on $M$. Each vertex $v$ of $V(M)$ is initialized with and maintains the following information:
\begin{itemize}
\item a distance estimate $d(v)$ which is initially $d_T(s,v)$,
\item an edge $(p(v),v)$ if $v\neq s$ where $p(v)$ is the parent of $v$ in $T$,
\item the set $P(v)$ of all representative edges $(u,v)$ with $\lceil d(u)\rceil_i = d(v)$ where $\ell(u,v) = i$.
\end{itemize}
Given a vertex $v\in V_M$, \DES can be queried in constant time for both $d(v)$ and $p(v)$.



\DES has two procedures, \texttt{UpdateP} and \texttt{UpdateDistances}. The procedure \texttt{UpdateP} takes no input and is automatically called after every operation applied to $\mathcal M$ and updates the $P$-sets without changing the distance estimate function $d$; the only exception to the latter is that after a \texttt{Split}-operation, for each new vertex $v$ in $M$ \texttt{UpdateP} sets $d(v)$ equal to the distance estimate for the vertex that was split. The procedure \texttt{UpdateDistances} may be called at any point and updates SSSP tree $T$ and the information associated with each vertex.

Internally, \DES maintains an initially empty min-priority queue $Q$ containing all vertices of $M$ whose $P$-sets changed since the last execution of \texttt{UpdateDistances}; the key values of these vertices are their distance estimates. This queue is updated by \texttt{UpdateP} and is emptied by \texttt{UpdateDistances}.

\paragraph{The \texttt{UpdateP} procedure:}
The procedure \texttt{UpdateP} is executed immediately after an update to $\mathcal M$ and does the following. Let $(E_{\mathit{old}},E_{\mathit{new}})$ be the output from $\mathcal M$. If the operation is a \texttt{Split}$(v,\{w_1,\ldots,w_{p-1}\})$-operation, then first set $d(w_j) = d(v)$ for $j = 1,\ldots,p-1$ and insert $w_1,\ldots,w_{p-1}$ and $v$ into $Q$ with their distance estimates as key values (ignoring insertions of vertices already present in $Q$). The remaining description of \texttt{UpdateP} in the following is shared among all types of operations applied to $\mathcal M$.

The $P$-sets are updated as follows. For each $(v_1,v_2)\in E_{\mathit{old}}$, update $P(v_2)\leftarrow P(v_2)\setminus\{(v_1,v_2)\}$. For each $(v_1,v_2)\in E_{\mathit{new}}$, if $\lceil d(v_1)\rceil_i = d(v_2)$ where $i = \ell(v_1,v_2)$, update $P(v_2)\leftarrow P(v_2)\cup\{(v_1,v_2)\}$.

For all $v$ for which $P(v)$ changed due to the above updates, $v$ is inserted into $Q$ with key value $d(v)$.

\paragraph{The \texttt{UpdateDistances} procedure:}
Procedure \texttt{UpdateDistances} is iterative. While $Q$ is non-empty, it does the following. A vertex $v$ with minimum key value is extracted from $Q$. If $d(v) > D$, the procedure proceeds to the next iteration.

Now, assume that $d(v)\leq D$ and $P(v) \neq \emptyset$. Then \texttt{UpdateDistances} sets $(p(v),v)$ equal to an arbitrary edge of $P(v)$ and proceeds to the next iteration.

Finally, assume that $d(v)\leq D$ and $P(v) = \emptyset$. Then the update $d(v)\leftarrow d(v) + 1$ is made. For each $i\in\{0,\ldots,k\}$ such that $d(v) - 1 + w_{M}(i)$ is divisible by $1 + W_{M}(i) - w_{M}(i)$, the following is done. For each $(v,w)\in E_{\mathit{out}}(v,i)$, if $\lceil d(v)\rceil_i\neq d(w)$ then $(v,w)$ is removed from $P(w)$ and $w$ is inserted into $Q$. For each $(w',v)\in E_{\mathit{in}}(v,i)$, if $\lceil d(w')\rceil_i = d(v)$ then $(w',v)$ is inserted into $P(v)$. Finally, $v$ is inserted into $Q$ and \texttt{UpdateDistances} proceeds to the next iteration.

\begin{Lem}\label{Lem:ES}
For $i = 0,\ldots,k$, let $\delta_i$ be an upper bound on the number of level $i$-representative edges incident to any vertex of $M$ during any execution of \texttt{UpdateDistances}. Then \DES is deterministic and after each call to \texttt{UpdateDistances}, it contains the current SSSP tree $T$ rooted at $s$ in $M$ up to distance $D$. The total time over all updates is
\[
  O((m+n)\log n + Dn(k + \log n) + Dn\sum_{i = 0}^k\delta_i/(1 + W_{M}(i) - w_{M}(i))).
\]

\end{Lem}
\begin{proof}
$\mathcal ES$ is clearly deterministic. For the correctness, we first show that all the $P$-sets are correctly maintained over all updates to $\mathcal M$. This is clear just prior to the first update so consider an update and assume that the $P$-sets are correct at the beginning of the update.

Regardless of the type of update, it is easy to see that the $P$-sets are correctly updated by procedure \texttt{UpdateP}. Now, consider an iteration of \texttt{UpdateDistances} in which a vertex $v$ is extracted from $Q$ and assume that at this point, all $P$-sets are correct. Just after the update $d(v)\leftarrow d(v) + 1$, an edge $(v,w)$ needs to leave $P(w)$ exactly when $\lceil d(v)-1\rceil_i = d(w)$, and $\lceil d(v)\rceil_i \neq d(w)$ where $i = \ell(v,w)$. This can only happen when $d(v) - 1 + w_{M}(i)$ is divisible by $1 + W_{M}(i) - w_{M}(i)$. Hence, \DES correctly updates $P(w)$ for each $w\neq v$ in the current iteration. A similar argument shows that $P(v)$ is correctly updated in the current iteration.



Next, we show the following: for each $d = 1,\ldots,D+1$, when all elements of $Q$ have key value at least $d$ then for all $v\in V_M$, $d(v) = d_M(s,v)$ if $d(v) < d$ and $d(v)\leq d_M(s,v)$ if $d(v)\geq d$.

The proof is by induction on the number of times procedures \texttt{UpdateP} and \texttt{UpdateDistances} have been executed so far. The claim holds initially since then $d(v) = d_M(s,v)$ for all $v\in V_M$. To show the induction step, consider first a single execution of \texttt{UpdateP} and assume that the claim holds at the beginning of this execution. \texttt{UpdateP} does not delete elements from $Q$ so the minimum key-value in $Q$ can only decrease during its execution. Furthermore, \texttt{UpdateP} does not change distance estimates. For new vertices $v$ obtained due to a \texttt{Split}-operation, \texttt{UpdateP} initializes $d(v)$ to a value of at most $d_M(s,v)$ and adds $v$ to $Q$. Thus, the claim also holds at the end of the execution of \texttt{UpdateP}.

It remains to show that the claim is maintained during a single execution of \texttt{UpdatesDistances}. This is done by induction on $d$. The claim is trivial for $d = 1$ so assume that $d > 1$ and that the claim holds for $d-1$. Consider a point in time when all elements of $Q$ have key value at least $d$. By the induction hypothesis, $d(u) = d_G(s,u)$ for all $u\in V$ with $d(u) < d-1$ and $d(u)\leq d_G(s,u)$ for all other $u$. Now, consider a $v\in V$ with $d(v) = d-1$. By the induction hypothesis, it suffices to show that $d(v) = d_G(s,v)$. Since $d(v)\leq d_G(s,v)$ we only need to show $d(v)\geq d_G(s,v)$.

We claim that $P(v)\neq\emptyset$. To see this, note that when \DES was initialized, all $P$-sets other than $P(s)$ were non-empty. Whenever a set $P(w)$ was changed in an execution of \texttt{UpdateP} or \texttt{UpdateDistances}, $w$ was inserted into $Q$ and when an element $w$ was extracted from $Q$ with $P(w) =  \emptyset$ and $d(w)\leq D$ then $d(w)$ was incremented and $w$ was inserted back into $Q$. Since currently, every element of $Q$ has key value at least $d$ and $d(v) = d-1\leq D$, we must have $P(v)\neq\emptyset$.

Let $(u,v)\in P(v)$ be given and let $i = \ell(u,v)$ be its level. We have $\lceil d(u)\rceil_i = d(v)$ so $d(u) < d(v) = d-1$. The induction hypothesis gives $d(u) = d_G(s,u)$ and we get $d(v) = \lceil d_G(s,u)\rceil_i = d_G(s,u) + w(u,v)\geq d_G(s,v)$. This completes the induction step.




At termination of any execution of \texttt{UpdateDistances}, $Q$ is empty. By the above, all elements of $Q$ trivially have key value at least $d = D+1$ at this point. Hence, $d(v) = d_G(s,v)$ for all $v\in V$. Since $(p(v),v)\in P(v)$ for each $v\in V\setminus s$, it follows that $T$ is a SSSP tree in $G$ with root $s$. This shows the correctness of \texttt{UpdateDistances} and thus of \DES.

It remains to show the time bound.  We implement $Q$ as a Fibonacci heap. Excluding the time spent in \texttt{UpdateDistances}, \DES spends time proportional to the total output size of $\mathcal M$ which is bounded by the total $O((m+n)\log n)$ time spent by $\mathcal M$.

To bound the total time spent in \texttt{UpdateDistances}, assume for the analysis that at all times, each vertex $v\in V$ has a distance estimate $d(v)$ equal to $d(v_M)$ where $v\in v_M\in V_M$. Consider an iteration of \texttt{UpdateDistances} in which a vertex $v_M$ is extracted from $Q$. The time spent in that iteration, including updates to $Q$ is $O(k + \log n + \sum_{i}\delta_i)$ where the sum is over those $i$ for which $d(v_M) - 1 + w_{M}(i)$ is divisible by $1 + W_{M}(i) - w_{M}(i)$. We charge this cost to an arbitrary vertex $v\in v_M$ of $V$; note that $d(v) - 1 + w_{M}(i) = d(v_M) - 1 + w_{M}(i)$ is divisible by $1 + W_{M}(i) - w_{M}(i)$. It follows that the total cost charged to vertices of $V$ over all updates is $O(Dn(k + \log n) + \sum_{i = 0}^k\delta_i/(1+W_M(i)-w_M(i)))$. This shows the desired.


\end{proof}

\section{Decremental SSSP Versus an Adaptive Adversary}\label{sec:SSSPAdaptive}
In this section, we present our first main result, Theorem~\ref{Thm:SSSPAdaptive}, in the unweighted setting. We generalize this result to the weighted setting in Appendix~\ref{sec:WeightedGraphs}.

\subsection{A reduction}\label{subsec:ReductionD}
In the following, let $D$ be a given power of $2$ between $1$ and $n$. We will present a data structure \DSSSP with total expected update time $\tilde O(m^{2/3}n^{4/3}/\epsilon^{2/3} + n^2/\epsilon^2)$ which can answer any intermediate query for $d_G(s,u)$ within an approximation factor of $(1+\epsilon)$, assuming $D\leq d_G(s,u) < 2D$; furthermore, if $d_G(s,u)\geq 2D$, the structure will output a value of at least $d_G(s,u)$ if queried with vertex $u$. The query time of the data structure is $O(1)$ and it works against an adaptive adversary. Furthermore, it outputs, after each update, the set of vertices whose approximate distances changed due to the update.

We now show that this suffices to show Theorem~\ref{Thm:SSSPAdaptive}. Consider the following decremental SSSP structure $\mathcal D$. This structure initializes and maintains each of the $O(\log n)$ structures above over the sequence of deletions. After each structure has been initialized, $\mathcal D$ queries, for all $u\in V$, the approximate $s$-to-$u$ distance of each structure and stores these $O(\log n)$ estimates in a min-priority queue $Q(u)$. The minimum key value in $Q(u)$ is a $(1+\epsilon)$-approximation of $d_G(s,u)$; this follows since none of the structures output a value below $d_G(s,u)$ and since either $D\leq d_G(s,u) < 2D$ for one of the $O(\log n)$ values of $D$ or $d_G(s,u) = \infty$.

At the end of each update, $\mathcal D$ queries each of the $O(\log n)$ structures for the approximate distances to those vertices $u$ whose estimate changed in that structure and the corresponding key in $Q(u)$ is updated. By maintaining this information after each update, $\mathcal D$ can thus answer any query in $O(1)$ time.

The total update time of $\mathcal D$ is $\tilde O(m^{2/3}n^{4/3}/\epsilon^{2/3} + n^2/\epsilon^2)$ since the total number of changes to estimates reported from the $O(\log n)$ structures cannot exceed their total update time. This shows the desired.

\subsection{Initialization}\label{subsec:Init}
We shall assume w.l.o.g.~that $G$ has no edges ingoing to $s$. Our data structure \DSSSP is initialized for graph $G = (V,E)$ as follows. First, an instance \DSCC of the data structure in Corollary~\ref{Cor:LowDiamDecomp} is initialized with parameters $0\leq d_1 < d_2$ to be fixed later. Let $S_0,\ldots,S_{\lceil\lg d_1\rceil}$ and (pointers to sets of) $\mathcal V$ be the output of \DSCC. \DSSSP computes an ordered list $\mathcal C = \{C_1,\ldots,C_{\ell}\}$ of the sets from $\mathcal V$ with properties that will be maintained throughout the sequence of updates. We therefore state these properties as an invariant:
\begin{Inv}\label{Inv:TopOrder}
During the sequence of updates,
\begin{enumerate}
\item $\mathcal C$ is a topological ordering of the graph obtained from $G_{\lceil\lg d_1\rceil}$ by contracting the vertex sets of $\mathcal V$, and
\item for $i = 0,\ldots,\lceil\lg d_1\rceil$ and for each SCC $C$ of $G_i$, the subsets of $\mathcal V$ contained in $C$ are consecutive in $\mathcal C$.
\end{enumerate}
\end{Inv}

We say that an edge of $G$ (or of multigraph $M$ below) is a \emph{forward edge} of $\mathcal C$ if the vertex set of $\mathcal C$ containing the start point of the edge does not appear later in $\mathcal C$ than the vertex set containing the endpoint of the edge. Invariant~\ref{Inv:TopOrder} ensures that all edges of $G_{\lceil\lg d_1\rceil}$ are forward edges of $\mathcal C$.

For each $C\in\mathcal C$, denote by $r(C)$ the sum of sizes of sets strictly preceding $C$ in $\mathcal C$; these values are computed and stored by \DSSSP. For any two $C,C'\in\mathcal C$ where $C$ precedes $C'$, we let $r(C,C') = r(C',C)$ denote the total size of sets strictly between $C$ and $C'$ in $\mathcal C$, i.e., $r(C,C') = r(C',C) = r(C') - r(C) - |C|$. For each $u\in V$, denote by $C(u)$ the set in $\mathcal C$ containing $u$.


Let $\tau > 0$ be some parameter to be fixed later. Define the function $\eta:(0,\infty)\mapsto\mathbb N_0$ by $\eta(x) = \lfloor\lg(x/\tau + 1)\rfloor$ and let $k = \eta(n)$. Note that $\eta([0,n]) = \{0,1,\ldots,k\}$.

\DSSSP initializes empty sets $E_0,\ldots,E_k$. Then it adds every edge $(u,v)\in E$ to $E_d$ where $d = \eta(r(C(u),C(v)))$. Below we will set an upper bound of $2^i$ on weights of edges of $M$ belonging to $E_i$, for $i = 0,\ldots,k$. Hence, each edge can have weight roughly up to the total size of sets of $\mathcal C$ strictly between its endpoints divided by $\tau$.

\DSSSP then sets the degree threshold $\Delta_i = 2^{i+2}\tau$ for $i = 0,\ldots,k$ which will be passed on as parameters when initializing multigraph structure $\mathcal M$ below. The intuition for choosing such a threshold is that when a vertex of $M$ is incident to more than $\Delta_i$ level $i$-representative edges, at least half of them should have their level increased due to the total size of SCCs between their endpoints being large.


Next, \DSSSP sets up an instance $\mathcal M$ of the multigraph structure from Section~\ref{sec:Multigraph} with the call \texttt{Init}$(G,\mathcal C,\{E_0,\ldots,E_k\}, \{\Delta_0,\ldots,\Delta_k)$. Denote by $M$ the multigraph that $\mathcal M$ maintains. Due to the assumption that $G$ has no ingoing edges to $s$, the vertex of $M$ containing $s$ thus represents the subset $\{s\}$ of $V$. We let this be the source of $M$ and for convenience, we shall refer to it as $s$ in the following. Note that the requirement from Section~\ref{sec:FastES} that the source of $M$ cannot be split is satisfied.

Finally, \DSSSP initializes an instance \DES of the data structure from Section~\ref{sec:FastES} associated with $\mathcal M$ with distance threshold $2D(1+\epsilon)$ and weight functions defined by $w_{M}(i) = 1$ and $W_{M}(i) = 2^i$, for $i = 0,\ldots,k$. The choice of distance threshold follows since we consider shortest path distances up to $2D$ and hence approximate shortest path distances up to $2D(1+\epsilon)$. This completes the description of the initialization step.

\subsection{Handling updates and queries}\label{subsec:UpdateQuery}
Now, consider the deletion of an edge $e = (u,v)$ from $E$. \DSSSP deletes $e$ from both \DSCC and $\mathcal M$.

If \DSCC does not output pointers to any new sets of $\mathcal V$, \DSSSP applies \texttt{UpdateDistances} to \DES and then finishes handling the deletion of $e$.

Now, consider the case where \DSCC outputs pointers to $W_1,\ldots,W_p$ and $V'$ where $W_1,\ldots,W_p$ replace $V'$ in $\mathcal V$. \DSSSP applies the operation \texttt{Split}$(V',\{W_1,\ldots,W_{p-1}\})$ to $\mathcal M$. Then \DSSSP computes a topological ordering of the multigraph obtained from $G_{\lceil\lg d_1\rceil}[W_1\cup\ldots\cup W_p]$ by contracting $W_1,\ldots,W_p$. Letting $W_{\pi(1)},\ldots,W_{\pi(p)}$ denote the corresponding ordering of $W_1,\ldots,W_p$, \DSSSP replaces $V'$ in $\mathcal C$ by the sublist $\langle W_{\pi(1)},\ldots,W_{\pi(p)}\rangle$. For each $W_j$, $r(W_j)$ and $|W_j|$ are computed and stored.



Next, as long as $\mathcal M$ contains a non-empty list of the form $V_{\mathit{in}}^{\Delta_i}(i)$ or $V_{\mathit{out}}^{\Delta_i}(i)$, \DSSSP picks such a list. If the list is of the form $V_{\mathit{in}}^{\Delta_i}(i)$ then let $V'$ be an arbitrary vertex of $V_{\mathit{in}}^{\Delta_i}(i)$; for each edge $(U,V')\in E_{\mathit{in}}(V',i)$, the value $j = \eta(r(U,V'))$ is computed and if $j > \ell(U,V')$, \DSSSP applies the operation \texttt{Increase}$((U,V'),j)$ to $\mathcal M$. Similarly, if the list is of the form $V_{\mathit{out}}^{\Delta_i}(i)$ then let $V'$ be an arbitrary vertex of $V_{\mathit{out}}^{\Delta_i}(i)$; for each edge $(V',U)\in E_{\mathit{out}}(V',i)$, the value $j = \eta(r(V',U))$ is computed and if $j > \ell(V',U)$, \DSSSP applies the operation \texttt{Increase}$((V',U),j)$ to $\mathcal M$.

Finally, \DSSSP applies \texttt{UpdateDistances} to \DES. This finishes the description of how our data structure handles an edge deletion.



To answer a query for the approximate distance from $s$ to a vertex $v$ in $G$, \DSSSP queries $\mathcal M$ to obtain the vertex $V'$ in $M$ containing $v$, queries \DES for $d(V')$, and returns $d(V') + 2d_2$.

\subsection{Correctness}\label{subsec:CorrectnessAdaptive}
It is easy to see from the description of \DSSSP that Invariant~\ref{Inv:TopOrder} is maintained.

We now introduce constraints on the parameters $\tau$, $d_1$, and $d_2$, ensuring that \DSSSP, when queried with a vertex $v$, outputs a value $\tilde d_G(s,v)\geq d_G(s,v)$ such that if $D\leq d_G(s,v) < 2D$ then $\tilde d_G(s,v)\leq (1+\epsilon)d_G(s,v)$. Later we optimize the values of the parameters within these constraints to minimize total update time.

In the following, we consider a query for a vertex $v$ in the current graph $G$ and let $\tilde d_G(s,v)$ be the estimate output by \DSSSP.

\paragraph{The lower bound on $\tilde d_G(s,v)$:}
We first show that $d_G(s,v)\leq\tilde d_G(s,v)$. Let $V'$ be the vertex of $M$ containing $v$.


We may assume that $\tilde d_G(s,v) < \infty$. Let $P_M = U_1,U_2,\ldots,U_p$ be a shortest path from $s$ to $V'$ in $M$; we may assume that $P_M$ consists of representative edges only. Form a path $P$ from $s$ to $v$ in $G$ from $P_M$ as follows. Let $u_1 = s$ and let $u_p' = v$. Pick for each edge $(U_i,U_{i+1})$ of $P_M$ an edge $(u_i',u_{i+1})$ of $G$ from $U_i$ to $U_{i+1}$. For $i = 1,\ldots,p$, let $P_i$ be a shortest path in $G[U_i]$ from $u_i$ to $u_i'$. Then $P$ is the concatenation $P_1\circ (u_1',u_2)\circ P_2\circ (u_2',u_3)\circ\cdots\circ (u_{p-1}',u_p)\circ P_p$. Since each edge of $M$ has weight at least $1$, we have $\tilde d_G(s,v)\geq |P_M| + 2d_2$. We also have $|P|\geq d_G(s,v)$ so it remains to show $|P_M| + 2d_2\geq |P|$. But this follows from Corollary~\ref{Cor:LowDiamDecomp} which implies $\sum_{i = 1}^p |P_i|\leq 2d_2$ and hence $|P| = |P_M| + \sum_{i = 1}^p |P_i|\leq |P_M| + 2d_2$.


\paragraph{The upper bound on $\tilde d_G(s,v)$:}
Next, assume that $D\leq d_G(s,v) < 2D$. We will put constraints on the parameters to ensure that $\tilde d_G(s,v)\leq (1+\epsilon)d_G(s,v)$. With $P_M$ defined as above, observe that $|P_M|\leq d_G(s,v)$. Since also $D\leq d_G(s,v)$, it suffices to ensure that $2d_2 + w(P_M) - |P_M| \leq \epsilon D$.

Let $E_1$ be the set of edges of $P_M$ incident to $\cup_{i = 0}^{\lceil\lg d_1\rceil} S_i$ and let $E_2$ be the remaining set of edges of $P_M$, i.e., those belonging to $G_{\lceil\lg d_1\rceil}$. Note that $w(P_M) - |P_M| = w(E_1) - |E_1| + w(E_2) - |E_2|$. In the following, we bound $w(E_1) - |E_1|$ and $w(E_2) - |E_2|$ separately.

For $i = 0,\ldots,\lceil\lg d_1\rceil$, it follows from Corollary~\ref{Cor:LowDiamDecomp} and from the fact that $P_M$ is simple that the number of edges of $E_1$ belonging to $E(S_i)\setminus\cup_{j < i}E(S_j)$ is at most $2|S_i| = \tilde O(n2^i/d_1)$ and by Corollary~\ref{Cor:LowDiamDecomp} and Invariant~\ref{Inv:TopOrder}, each of them belongs to level $\eta(n/2^i)$ and thus have weight $O(n/(2^i\tau))$ in $M$. Hence, $w(E_1) - |E_1| \leq w(E_1) = \tilde O(n^2/(d_1\tau))$.

To bound $w(E_2) - |E_2|$, define a potential function $\Phi:\mathcal C\rightarrow\mathbb R_0$ by $\Phi(C) = r(C)/\tau$. We will consider the changes to $\Phi$ as edges of $P_M$ are traversed in the order they appear along this path.

Consider one such edge $(C_1,C_2)$ of $P_M$. If $C_1$ belongs to some set $S_i$ and $C_2$ does not belong to $\cup_{j < i}S_j$ then by Corollary~\ref{Cor:LowDiamDecomp} and Invariant~\ref{Inv:TopOrder}, traversing $(C_1,C_2)$ reduces $\Phi$ by at most $n/(2^i\tau)$. Similarly, if $C_2$ belongs to some set $S_i$ and $C_1$ does not belong to $\cup_{j < i}S_j$ then traversing $(C_1,C_2)$ reduces $\Phi$ by at most $n/(2^i\tau)$.

The only remaining case is when $(C_1,C_2)\in E_2$. By Invariant~\ref{Inv:TopOrder}, traversing $(C_1,C_2)$ increases $\Phi$ by at least $r(C_1,C_2)/\tau$. Note that with $j = \eta(r(C_1,C_2))$, the weight of $(C_1,C_2)$ in $M$ is at most $2^j = 2^{\lfloor\lg(r(C_1,C_2)/\tau + 1)\rfloor}\leq 1 + r(C_1,C_2)/\tau$; this follows since $r(C_1,C_2)$ can only increase over time. Hence, an upper bound on the total increase in $\Phi$ over all edges $(C_1,C_2)$ of $P_M$ belonging to $G_{\lceil\lg d_1\rceil}$ will thus be an upper bound on $w(E_2) - |E_2|$.

Note that $\Phi$ can never exceed $n/\tau$ and by Corollary~\ref{Cor:LowDiamDecomp}, the total reduction in $\Phi$ over all edges of $P_M$ is bounded by $O(\sum_{i = 0}^{\lceil\lg d_1\rceil}|S_i|n/(2^i\tau)) = \tilde O(n^2/(d_1\tau))$. Hence, the total increase in $\Phi$ is bounded by $n/\tau + \tilde O(n^2/(d_1\tau)) = \tilde O(n^2/(d_1\tau))$.

We conclude that $w(P_M) - |P_M| = \tilde O(n^2/(d_1\tau))$. We can thus ensure that $\tilde d_G(s,v)\leq (1+\epsilon)d_G(s,v)$ with the following constraint:
\[
  d_2 + n^2/(d_1\tau) = \tilde O(\epsilon D),
\]
which should be interpreted as $d_2 + n^2/(d_1\tau)\leq \epsilon D/\log^c n$ for sufficiently large constant $c > 0$.



\subsection{Running time}\label{subsec:TimeAdaptive}
We now give the implementation details that will allow us to obtain an expression for the running time in the parameters introduced earlier. We will then minimize this expression under the constraint $d_2 + n^2/(\tau d_1) = \tilde O(\epsilon D)$ from the previous subsection.

The total time spent by \DSCC is $\tilde O(m\sqrt n + mnd_1/(d_2 - d_1))$ by Corollary~\ref{Cor:LowDiamDecomp} and the total time spent by $\mathcal M$ is $\tilde O(m)$ by Lemma~\ref{Lem:M}.

To bound the total time spent by \DES, we will bound $\delta_i$ in Lemma~\ref{Lem:ES} for $i = 0,\ldots,k$.

During initialization of \DSSSP, each edge $(u,v)\in E$ is added to $E_i$ where $i = \eta(r(C(u),C(v))) = \lfloor\lg(r(C(u),C(v))/\tau + 1)\rfloor$. Note that $r(C(u),C(v))/\tau\leq 2^{i+1} - 1$ and hence $r(C(u),C(v))\leq (2^{i+1}-1)\tau$. Since there is at most one representative edge between every ordered pair of vertices of $M$, it follows that at the end of the initialization, each vertex of $M$ is incident to at most $2(2^{i+1} - 1)\tau$ representative edges of level $i$, for $i = 0,\ldots,k$. The procedure \texttt{UpdateDistances} of \DES is only applied by \DSSSP once all lists of the form $V_{\mathit{in}}^{\Delta_i}(i)$ or $V_{\mathit{out}}^{\Delta_i}(i)$ are empty. Hence, we can choose $\delta_i = \max\{2(2^{i+1}-1)\tau,2\Delta_i\} = 2^{i+3}\tau$. It follows from Lemma~\ref{Lem:ES} that \DES takes total time $O((m+n)\log n + Dn\log n + Dn\sum_{i = 0}^k 8\tau) = \tilde O(m + Dn\tau)$.



We now bound the time for the work done during the initialization step (Section~\ref{subsec:Init}) and during an update (Section~\ref{subsec:UpdateQuery}), excluding the time spent by \DSCC, \DES, and $\mathcal M$.

It is straightforward to implement the initialization step in $O(m)$ time. For updates, consider first the time spent on updating $\mathcal C$ after a \texttt{Split}$(V',\{W_1,\ldots,W_{p-1}\})$-operation to $\mathcal M$. This can be done in time proportional to the size of the induced graph $G_{\lceil d_1\rceil}[W_1\cup\ldots\cup W_p]$ where $W_p = V'\setminus\cup_{i = 0}^{p-1}W_i$. Observe that each edge of this graph is incident to a set $W_i$ with $1\leq i\leq p-1$ and that $|W_i|\leq\frac 1 2 |V'|$. Hence, the total size of these induced subgraphs over all updates is $O(m\log n)$ which is thus a bound on the total time spent on updates to $\mathcal C$.

The remaining amount of time spent by \DSSSP is dominated by the total number of edges $(U,V')$ visited for $V'$ in a set $V_{\mathit{in}}^{\Delta_i}(i)$ plus the total number of edges $(V',U)$ visited for $V'$ in a set $V_{\mathit{out}}^{\Delta_i}(i)$. We only bound the former due to symmetry. Let $V'$ be a vertex of a set $V_{\mathit{in}}^{\Delta_i}(i)$ inspected by \DSSSP. By definition, $V'$ has more than $\Delta_i = 2^{i+2}\tau$ ingoing representative edges of level $i$. It follows that for more than $2^{i+1}\tau$ of these edges $(U,V')$, $r(U,V') > 2^{i+1}\tau$ and so $\eta(r(U,V'))\geq i+1$. Hence, at least half of these edges will have their level increased by at least $1$ with a call to \texttt{Increase}. Since an edge level can increase at most $k = O(\log n)$ times and since the total number of distinct representative edges over all updates is $\tilde O(m)$, the total number of edges $(U,V')$ visited by \DSSSP over all updates is $\tilde O(m)$.


Putting everything together, we have the constraint $d_2 + n^2/(d_1\tau) = \tilde O(\epsilon D)$. We also introduce the constraint $d_2\geq 2d_1$ so that $mnd_1/(d_2 - d_1) = O(mnd_1/d_2)$ and we have a total update time of
\[
  \tilde O(m\sqrt n + mn/d_1 + mnd_1/d_2 + Dn\tau).
\]

In Appendix~\ref{subsec:OptParSSSPAdaptive}, we show how to optimize the values of the various parameters under the given constraints to get the time bound of Theorem~\ref{Thm:SSSPAdaptive}.

\section{A Faster Data Structure for Dense Graphs}\label{sec:SSSPDenseObl}
In this section, we show our second main result, Theorem~\ref{Thm:SSSPDense}, in the unweighted setting. The generalization to weighted graphs can be found in Appendix~\ref{sec:WeightedGraphs}.

The result is obtained in much the same way as Theorem~\ref{Thm:SSSPAdaptive}. The main difference is that we use a faster version of the data structure of Theorem~\ref{Thm:LowDiamDecomp} which only works against an oblivious adversary. The result is described in the following theorem:
\begin{theorem}\label{Thm:LowDiamDecompObl}
Let $G = (V,E)$ be an unweighted graph undergoing edge deletions, let $m = |E|$ and $n = |V|$, and let integer $d$ with $0 < d \leq n$ be given. Then there is a Las Vegas data structure which maintains a pair $(S,\mathcal V)$ where $S\subseteq V$ is a growing set and where $\mathcal V$ is the family of vertex sets of the SCCs of $G\setminus E(S)$ such that at any point, all these SCCs have diameter at most $d$ and $|S| = \tilde{O}(n/d)$.

After the initialization step, the data structure outputs the initial pair $(S,\mathcal V)$. After each update, the data structure outputs the set $S'$ of new vertices of $S$ where $S'\subseteq V'$ for some $V'\in\mathcal V$. Additionally, it updates $\mathcal V$ by replacing at most one $V'\in\mathcal V$ by the vertex sets $W_1,\ldots,W_p$ of the new SCCs of $G\setminus E(S)$ where $|W_i|\leq\frac 1 2 |V'|$ for $i = 1,\ldots,p-1$. Pointers to $W_1,\ldots,W_p$ are returned.

The total expected time is $\tilde O(md + mn/d)$ assuming an oblivious adversary.
\end{theorem}

This theorem follows using techniques very similar to those in~\cite{ChechikHILP16}. We have therefore moved the proof to Appendix~\ref{sec:LowDiamDecompOblProof}.

We can now get a corollary similar to Corollary~\ref{Cor:LowDiamDecomp}.
\begin{Cor}\label{Cor:LowDiamDecompObl}
Let $G = (V,E)$ be an unweighted graph undergoing edge deletions, let $m = |E|$ and $n = |V|$, and let integer $d$ with $0 < d\leq n$ be given. Then there is a Las Vegas data structure which maintains pairwise disjoint growing subsets $S_0, S_1,\ldots, S_{\lceil\lg d\rceil}$ of $V$ and a family $\mathcal V$ of subsets of $V$ with the following properties. For $i = 0,\ldots,\lceil\lg d\rceil$, let $G_i = G\setminus(\cup_{j = 0}^i E(S_j))$. Then over all updates, $\mathcal V$ is the family of vertex sets of the SCCs of $G_{\lceil\lg d\rceil}$ and for $i = 0,\ldots,\lceil\lg d\rceil$,
\begin{enumerate}
\item each SCC of $G_i$ of vertex size at most $n/2^i$ has diameter at most $d/2^i$,
\item if $i > 0$, every vertex of $S_i$ belongs to an SCC of $G_{i-1}$ of vertex size at most $n/2^i$, and
\item $|S_i| = \tilde O(n2^i/d)$.
\end{enumerate}
Furthermore, the sum of diameters of all SCCs of $G_{\lceil\lg d\rceil}$ is at most $2d$.

After the initialization step, the data structure outputs the inital sets $S_0, S_1,\ldots, S_{\lceil\lg d\rceil}$ and pointers to the sets of $\mathcal V$. After each update, the data structure outputs the new vertices of $S_0,\ldots,S_{\lceil\lg d\rceil}$. Additionally, it updates $\mathcal V$ by replacing at most one $V'\in\mathcal V$ by the vertex sets $W_1,\ldots,W_p$ of the new SCCs of $G_{\lceil\lg d\rceil}$ where $|W_i|\leq\frac 1 2 |V'|$ for $i = 1,\ldots,p-1$. Pointers to both the old set $V'$ and to the new sets $W_1,\ldots,W_p$ are returned.

The total expected time is $\tilde O(md + mn/d)$ assuming an oblivious adversary.
\end{Cor}
\begin{proof}
The proof is identical to that of Corollary~\ref{Cor:LowDiamDecomp} except that Theorem~\ref{Thm:LowDiamDecompObl} is applied instead of Theorem~\ref{Thm:LowDiamDecomp} and $d$ is used instead of $d_1$ and $d_2$.
\end{proof}

\subsection{The SSSP structure}\label{subsec:SSSPDenseObl}
To obtain Theorem~\ref{Thm:SSSPDense}, we use a slight modification of \DSSSP which implements \DSCC as the data structure of Corollary~\ref{Cor:LowDiamDecompObl} rather than the data structure of Corollary~\ref{Cor:LowDiamDecomp}. Correctness follows from the same arguments as before with the constraint $d_2 + n^2/(\tau d_1) = \tilde O(\epsilon D)$ replaced by $d + n^2/(\tau d) = \tilde O(\epsilon D)$ to ensure the desired $(1+\epsilon)$-approximation factor. The time bound for distances $d_G(s,u)$ with $D\leq d_G(s,u) < 2D$ becomes
\[
\tilde O(md + mn/d + Dn\tau).
\]

In Appendix~\ref{subsec:OptParSSSPDense}, we optimize the values of the parameters to get the desired update time bound of Theorem~\ref{Thm:SSSPDense}.

\subsection{Reporting paths}\label{subsec:PathReport}
Consider a query vertex $v$. We now extend our data structure to be able to report a $(1+\epsilon)$-approximate path from $s$ to $v$ in $G$ in time proportional to its length, thereby showing the second part of Theorem~\ref{Thm:SSSPDense}.

From the min-priority queue for $u$ described in Section~\ref{subsec:ReductionD}, a structure \DSSSP can be obtained in $O(1)$ time such that the value $\tilde d_G(s,v)$ output by \DSSSP satisfies $d_G(s,v)\leq \tilde d_G(s,v)\leq (1+\epsilon)d_G(s,v)$. Let $D$ be the value associated with \DSSSP. If $D < n^{3/2}/(\sqrt m\epsilon^{3/2})$, \DSSSP is the normal structure of Even and Shiloach which allows for the desired path to be reported within the desired time bound by traversing parent pointers from $v$ back to $s$.

Now, assume that $D \geq n^{3/2}/(\sqrt m\epsilon^{3/2})$. We assume in the following that \DSSSP has access to the in- and out-trees maintained by \DSCC together with parent pointers in the two trees. We also use $d' = \frac 1 2 d$ instead of $d$ in Corollary~\ref{Cor:LowDiamDecompObl}. This does not affect the total asymptotic update time of \DSSSP.

The structure \DES maintains an SSSP-tree with parent pointers in multigraph $M$ so \DSSSP obtains from \DES the edges of the path $P_M = \langle C_1 = s, C_2,\ldots,C_k\rangle$ where $C_k$ is the vertex of $M$ containing $v$. The edges $e_1 = (v_1' = s,v_2), e_2 = (v_2',v_3),\ldots, e_{k-1} = (v_{k-1}',v_k)$ of $E$ corresponding to the edges $(C_1=s,C_2),(C_2,C_3),\ldots,(C_{k-1},C_k)$ of $P_M$ are obtained. For $i = 1,\ldots,k$, let $r_i$ be the root of the in-out-tree maintained by \DSCC in $C_i$. Also, let $v_1 = s$ and $v_k' = v$. For $i = 1,\ldots,k$, \DSSSP obtains the path $P_i$ from $v_i'$ to $r_i$ in the in-tree of $C_i$ and the path $P_i'$ from $r_i$ to $v_i$ in the out-tree of $C_i$. Finally, \DSSSP outputs the path $P = P_1\circ P_1'\circ e_1\circ P_2\circ P_2'\circ e_2\circ\cdots\circ e_{k-1}\circ P_k\circ P_k'$ as the answer to the path query.

We now show correctness and running time. It is clear from the description that the running time is $O(|P|)$. For correctness, recall that $\tilde d_G(s,v) = d(C_k) + 2d$ where $d(C_k)$ is the distance from $s$ to $C_k$ in \DES. We have that the total weight of edges $e_1,\ldots,e_{k-1}$ is $k-1\leq d(C_k)$. It remains to show that $\sum_{i = 1}^k|P_i|\leq 2d' = d$ and $\sum_{i = 1}^k|P_i'|\leq 2d' = d$ since then $|P|\leq d(C_k) + 2d = \tilde d_G(s,v)$. By symmetry, we only show the former inequality. But this follows from Corollary~\ref{Cor:LowDiamDecompObl} which implies that the sum of diameters of all SCCs of $G_{\lceil\lg d\rceil}$ is at most $2d'$.

\section{A Faster Data Structure for Sparse Graphs}\label{sec:SSSPSparse}
In this section, we show our third and final main result, Theorem~\ref{Thm:SSSPSparse}. As before, we focus only on unweighted graphs here and generalize to weighted graphs in Appendix~\ref{sec:WeightedGraphs}.

To simplify calculations, we will show an approximation factor of $(1+c_1\epsilon)^{c_2}$ for constants $c_1,c_2\geq 1$ instead of $(1+\epsilon)$; this suffices since we can always pick another $\epsilon' = \Theta(\epsilon)$ to ensure a factor of $(1+\epsilon)$.


\subsection{A modified multigraph structure}\label{subsec:ModM}
We need a slightly different multigraph structure $\mathcal M$ than the one in Section~\ref{sec:Multigraph}. In this subsection, we describe the changes made. Firstly, we no longer need parameters $\Delta_0,\ldots,\Delta_k$ and we shall ignore them in the call to \texttt{Init}, i.e., they are given arbitrary values in this call. Secondly, we need $\mathcal M$ to handle super-edges representing paths in $G$. With $M$ denoting the multigraph represented by $\mathcal M$, we introduce the following two additional operations for $\mathcal M$:
\begin{description}
\item [\texttt{SInsert}$(u,v,i)$:] inserts a super-edge $e$ from vertex $u\in V(M)$ to vertex $v\in V(M)$ and assigns $\ell(e) = i$; returns a pointer to $e$,
\item [\texttt{SDelete}$(e)$:] deletes super-edge $e$ from $M$.
\end{description}
Furthermore, $\mathcal M$ also supports the operation \texttt{Increase}$(e,i)$ for super-edges $e$.

In addition to the data in Section~\ref{sec:Multigraph}, $\mathcal M$ supports constant-time access to $\ell(e)$ also for super-edges $e$ and the $E_{\mathit{in}}$- and $E_{\mathit{out}}$-linked lists are extended to also include super-edges. A \texttt{Split}$(V',\{W_1,\ldots,W_{p-1}\})$-operation will not affect super-edges incident to $V'$; these all become incident to the new vertex $W_p$ that is identified with $V'$.

It is easy to see that with these modifications, $\mathcal M$ can be implemented to support all existing operations within the time time bounds of the following lemma and can support each \texttt{SInsert}- and \texttt{SDelete}-operation in constant time.
\begin{Lem}\label{Lem:ESSuper}
For $i = 0,\ldots,k$, let $\delta_i$ be an upper bound on the expected number of level $i$-super-edges incident to any vertex of $M$ during any execution of \texttt{UpdateDistances}. Let $x_i$ count the expected number of times \DES increases the distance estimate of any vertex of $M$ which at the moment its estimate increases is incident to at least one level $i$-super-edge. Let $m_i$ be the initial number of level $i$-edges of $E$ and assume that these edges never change level. Then \DES correctly maintains SSSP tree $T$ in $M$ up to distance $D$ over all updates using a total expected time of
\[
  O((m+n)\log n + (Dn(k + \log n) + \sum_{i = 0}^k(x_i\delta_i + m_iD))/(1 + W_{M}(i) - w_{M}(i))).
\]
\end{Lem}
\begin{proof}
Only the time bound analysis for \texttt{UpdateDistances} differs from the proof of Lemma~\ref{Lem:ES}. Assume for the analysis that at all times, each vertex $v\in V$ has a distance estimate $d(v)$ equal to $d(v_M)$ where $v\in v_M\in V_M$. Consider an iteration of \texttt{UpdateDistances} in which a vertex $v_M$ is extracted from $Q$. For $i = 0,\ldots,k$, let $\delta_i(v_M)$ be the number of level $i$-edges of $E$ incident to vertices of $v_M$ and let $\delta_i'(v_M)$ be the number of level $i$-super-edges incident to $v_M$. The time spent in that iteration, including updates to $Q$ is $O(k + \log n + \sum_{i}(\delta_i(v_M) + \delta_i'(v_M)))$ where the sum is over those $i$ for which $d(v_M) - 1 + w_{M}(i)$ is divisible by $1 + W_{M}(i) - w_{M}(i)$. We charge the $O(k + \log n)$ cost to an arbitrary vertex $v\in v_M$ of $V$; note that $d(v) - 1 + w_{M}(i) = d(v_M) - 1 + w_{M}(i)$ is divisible by $1 + W_{M}(i) - w_{M}(i)$. It follows that the total cost charged to vertices of $V$ over all updates is $O(Dn(k + \log n)/(1+W_M(i) - w_M(i)))$. Furthermore, $\sum_i\delta_i(v_M)$ summed over all extractions from $Q$ is $O(\sum_{i = 0}^k m_iD/(1 + W_M(i) - w_M(i)))$. Finally, $\sum_i\delta_i'(v_M)$ summed over all extractions from $Q$ is $O(\sum_{i = 0}^kx_i\delta_i/(1+W_M(i) - w_M(i)))$. This shows the desired.
\end{proof}

\subsection{Initialization}
As in previous sections, we let \DSSSP refer to the data structure approximating distances $d_G(u,v)$ with $D\leq d_G(u,v) < 2D$. We start by describing its initialization step.

First, an augmented graph $G_+$ is formed from $G$ as follows. Sample each vertex of $V$ independently with a certain probability $p = (c\ln n)/((1+\epsilon)D')$ where $c > 0$ is a constant and $D'\in\mathbb N$ is a value to be specified later. Let $V_+$ be $V$ together with a new vertex $s_v$ for each sampled vertex $v$. Let $E_+$ be $E$ together with edges $(s_v,v)$ for each sampled vertex $v$. Then $G_+ = (V_+,E_+)$.

\DSSSP initializes an instance \DSCC of the data structure of Corollary~\ref{Cor:LowDiamDecompObl} for $G_+$ with a parameter $d$ to be specified later. Note that each vertex of $V_+\setminus V$ forms a singleton set in the collection $\mathcal V$ output by \DSCC. \DSSSP computes $\mathcal C$ as in Section~\ref{subsec:Init}; since we still want the definition of $r(C)$ and $r(C,C')$ from Section~\ref{subsec:Init} to be w.r.t.~vertices of $V$, \DSSSP keeps all singleton sets $\{s_v\}$ at the end of the list $\mathcal C$.

With $k = \lceil\log_{1+\epsilon} n\rceil$, empty sets $E_0,\ldots,E_k$ are initialized and all edges of $E$ are added to $E_{\log_{1+\epsilon} \rho}$, for some value $\rho\leq n$ which is a power of $(1+\epsilon)$; this value will be specified later. All edges of $E_+\setminus E$ are added to $E_0$. Then two multigraph structures, $\mathcal M$ and $\mathcal M_+$, are initialized, both with the call \texttt{Init}$(G_+,\mathcal C, \{E_0,\ldots,E_k\})$. Let $M$ resp.~$M_+$ denote the multigraph maintained by $\mathcal M$ resp.~$\mathcal M_+$. At all times, these two multigraphs will be identical except that $M_+$ may in addition contain super-edges.

Each vertex $s_v$ is contained in a vertex of multigraph $M$ containing no other vertex of $V$; we shall refer to each such vertex of $M$ as a \emph{super-source}. If $z$ is the number of super-sources, consider a random permutation $\pi$ of $\{1,2,\ldots,z\}$. For some ordering of the super-sources not depending on $\pi$, assign the $i$th super-source the random \emph{priority} $\pi(i)$.

In the following, we say that a vertex of $M$ is sampled if it contains at least one of the vertices of $V$ sampled above. We denote by $N(s')$ the unique neighbor of super-source $s'$ in $\mathcal M$. The same notation is used for $\mathcal M_+$ and $M_+$.

Let $\delta\in\mathbb R_+$ be a value to be specified later. For each super-source $s'$, denote by $M(s')$ the following subgraph of the current multigraph $M$. $M(s')$ consists of the vertices $v$ where either $v = s'$ or where $r(N(s'),v)\leq\delta$. The edges of $M(s')$ are those of $M$ having both endpoints in $V(M(s'))$ such that at least one of the endpoints $v_M$ satisfies $r(N(s'),v_M) + |v_M|\leq\delta$. Associate unit-weight functions $w_{M(s')}$ and $W_{M(s')}$ with $M(s')$, i.e., $w_{M(s')}(i) = W_{M(s')}(i) = 1$ for $i = 0,\ldots,k$.

Let $D'\in\mathbb N$ be a value specified later. For each sampled vertex $v$ of $\mathcal C$, \DSSSP picks a super-source $s'$ of highest priority from the set of super-sources with $v$ as neighbor in $M_+$. Then \DSSSP initializes an instance \DES$(s')$ of a slightly modified version (see below) of the data structure of Section~\ref{sec:FastES} with source vertex $s'$, distance threshold $D'(1+\epsilon) + 1$, and for graph $M(s')$. We shall refer to super-sources $s'$ having an associated instance \DES$(s')$ as \emph{active}; all other super-sources are \emph{inactive}.

The modifications of \DES$(s')$ compared to that of Section~\ref{sec:FastES} are
\begin{enumerate}
\item \texttt{UpdateP} is applied just after an update to $\mathcal M$ iff that update changes $M(s')$, and
\item whenever a vertex $v$ of \DES$(s')$ is extracted from $Q$ in \texttt{UpdateDistances}, if either $d(v) > D'(1+\epsilon) + 1$ or $r(N(s'),v) > \delta$ (i.e., if $v$ is no longer in $M(s')$), the procedure sets $d(v) = \infty$ and continues to the next iteration.
\end{enumerate}

An instance \DES of the data structure of Section~\ref{sec:FastES} is initialized for $M_+$ with distance threshold $2D(1+\epsilon)$, and weight functions $w_{M_+}$ and $W_{M_+}$ defined by $w_{M_+}(i) = \lfloor(1+\epsilon)^i\rfloor$ and $W_{M_+}(i) = \lfloor(1+\epsilon)^{i+1}\rfloor$ for $i = 0,\ldots,k$.

Finally, \DSSSP adds super-edges to $\mathcal M_+$ as follows. Each structure \DES$(s')$ is queried for the distance $d_{M(s')}(N(s'),v) = d_{M(s')}(s',v) - 1$ from $N(s')$ to each sampled vertex $v\neq s'$ reachable from $s'$ in $M(s')$. For each such pair $(N(s'),v)$, if $D'(1+\epsilon)\leq d_{M(s')}(N(s'),v)\leq 2D'(1+\epsilon)$, the operation \texttt{SInsert}$(N(s'),v,i)$ is applied to $\mathcal M_+$ with $i = \lceil\log_{1+\epsilon}d_{M(s')}(N(s'),v)\rceil$. This completes the description of the initialization step for \DSSSP.

\subsection{Handling updates and queries}
Now, consider an update in which an edge $e$ is deleted from $G$. \DSSSP first deletes $e$ from \DSCC, $\mathcal M$, and $\mathcal M_+$.

Assume first that \DSCC does not output pointers to any new sets of the family $\mathcal V$ that it maintains. Then \DSSSP applies \texttt{UpdateDistances} to all structures \DES$(s')$ with non-empty queues. Next, for every super-edge $e' = (N(s'),v)$ of $\mathcal M_+$ for which $d_{M(s')}(N(s'),v)$ changed, the following is done. If $d_{M(s')}(N(s'),v) = \infty$, the operation \texttt{SDelete}$(e')$ is applied to $\mathcal M_+$. Otherwise, let $i = \lceil\log_{1+\epsilon}d_{M(s')}(N(s'),v)\rceil$; if $i > \ell(e')$, \DSSSP applies the operation \texttt{Increase}$(e',i)$ to $\mathcal M_+$.

Finally, \texttt{UpdateDistances} is applied to \DES.

\paragraph{Dealing with splits:}
Now, consider the case where \DSCC outputs pointers to $W_1,\ldots,W_p$ and $V'$ where $W_1,\ldots,W_p$ replace $V'$ in $\mathcal V$. First, \DSSSP applies \texttt{SDelete}$(e')$ to $\mathcal M_+$ for each super-edge $e'$ incident to $V'$. Next, \texttt{Split}$(V',\{W_1,\ldots,W_{p-1})$ is applied to $\mathcal M$ and to $\mathcal M_+$. \DSSSP updates $\mathcal C$ and computes for each $W_j$ values $r(W_j)$ and $|W_j|$ as in Section~\ref{subsec:UpdateQuery}.


Next, for each $W_j$ of $M$, if $W_j$ is sampled but has no ingoing edge from an active super-source, pick a super-source $s'$ of highest priority from the set of super-sources having edges ingoing to $W_j$ and initialize a new structure \DES$(s')$ as described in the initialization step above.

Now, for each active super-source $s'$ incident to one of the new vertices $W_j$ of $M$, \DES$(s')$ is queried for the distance $d_{M(s')}(N(s'),v) = d_{M(s')}(s',v) - 1$ from $N(s')$ to each sampled vertex $v\neq s'$ in $M(s')$. For each such pair $(N(s'),v)$, if $D'(1+\epsilon)\leq d_{M(s')}(N(s'),v)\leq 2D'(1+\epsilon)$, the operation \texttt{SInsert}$(N(s'),v,i)$ is applied to $\mathcal M_+$ with $i = \lceil\log_{1+\epsilon}d_{M(s')}(N(s'),v)\rceil$.

Next, \texttt{UpdateDistances} is applied to all structures \DES$(s')$ with non-empty queues and super-edges are deleted/have their levels increased accordingly, as described above. Finally, \texttt{UpdateDistances} is applied to \DES. This completes the description of how \DSSSP handles an update.

\paragraph{Answering queries:}
A query for an approximate distance from $s$ to a vertex $v$ in $G$ is answered by \DSSSP by querying $\mathcal M_+$ to obtain the vertex $V'$ in $M_+$ containing $v$, and then querying \DES for $d(V')$. \DSSSP returns $d(V') + 2d$.

\subsection{Correctness}\label{subsec:CorrectnessSparse}
We now show that when queried with a vertex $v\in V$, \DSSSP described above outputs an estimate $\tilde d_G(s,v)\geq d_G(s,v)$ such that if $d\leq d_G(s,v)\leq 2D$ then w.h.p., $\tilde d_G(s,v)\leq (1+2\epsilon)^2d_G(s,v)$.

Assume that $D\leq d_G(s,v) < 2D$. Let $P_{M}$ be a minimum-length path in $M$ from $s$ to the vertex $V'$ of $M$ containing $v$. Let $w$ be the edge weight function for $M_+$ and let $P_{M_+}$ be a shortest path from $s$ to $V'$ in $M_+$.

Consider any subpath of $P_M$ of length $x > D'(1+\epsilon)$. Since vertices of $V$ are sampled independently with probability $p = (c\ln n)/((1+\epsilon)D')$, the probability that this subpath contains no sampled interior vertices is at most $(1-p)^{x-1} < (1-(c\ln n)/((1+\epsilon)D'))^{D'(1+\epsilon)-1} < n^{-c/2}$. A union bound over all such subpaths shows that with probability greater than $1 - n^{2 - c/2}$, every subpath of $P_M$ of length greater than $D'(1+\epsilon)$ contains at least one sampled interior vertex. Assume that this event holds in the following.

By the above, we can partition $P_{M}$ into subpaths each of length between $D'(1+\epsilon)$ and $2D'(1+\epsilon)$ and each starting in either $s$ or in a sampled vertex and each ending in either $V'$ or a sampled vertex. Partition the set of these subpaths into two subsets, $\mathcal Q_1$ and $\mathcal Q_2$; $\mathcal Q_1$ consists of the subpaths $Q = v_1v_2\cdots v_k$ where $r(v_1,v_i)\leq\delta$ for $i = 2,\ldots,p$, and $\mathcal Q_2$ consists of the remaining subpaths.

Let $Q = v_1v_2\cdots v_p$ be a subpath in $\mathcal Q_1$. Then $Q$ is fully contained in $M(s')$ where $s'$ is the active super-source having an edge to $v_1$. Since $D'(1+\epsilon)\leq |Q|\leq 2D'(1+\epsilon)$, \DES$(s')$ maintains the distance $d_{M(s')}(v_1,v_p)$. Hence, $M_+$ contains a super-edge $e'$ from $v_1$ to $v_p$ with $\ell(e') = \lceil\log_{1+\epsilon} d_{M(s')}(v_1,v_p)\rceil$. Since $w_{M_+}(\ell(e'))\geq d_{M(s')}(v_1,V_p)$ and $W_{M_+}(\ell(e'))\leq (1+\epsilon)^2d_{M(s')}(v_1,v_p)$ and since $d_{M(s')}(v_1,v_p) = |Q|$, it follows from Lemma~\ref{Lem:EdgeWeightFunction} that $|Q|\leq w(e')\leq (1+\epsilon)^2|Q|$. 

Now, consider a subpath $Q = v_1v_2\cdots v_p$ in $\mathcal Q_2$. Then $r(v_1,v_i) > \delta$ for some $i\in\{2,\ldots,p\}$. In particular, $\sum_{j = 2}^p r(v_{j-1},v_j) > \delta$. Furthermore, all edges of $Q$ are present in $M_+$ and by Lemma~\ref{Lem:EdgeWeightFunction}, each such edge has weight in the range $[\lfloor\rho\rfloor,\lfloor(1+\epsilon)\rho\rfloor]$.

It follows from the above and from $\rho\geq 1$ that $w(P_{M_+})\geq |P_M|$ and hence that $\tilde d_G(s,v) = 2d + w(P_{M_+})\geq 2d + |P_M|\geq d_G(s,v)$ where the last inequality follows from Corollary~\ref{Cor:LowDiamDecompObl}. This shows the lower bound on $\tilde d_G(s,v)$.

We now add constraints to the parameters to ensure that w.h.p., the upper bound on $\tilde d_G(s,v)$ holds. It suffices to ensure that w.h.p., $2d + w(P_{M_+})\leq \epsilon D + (1+\epsilon)^2|P_M|$ since then the query algorithm returns w.h.p.~a value of $2d + w(P_{M_+})\leq \epsilon D + (1+\epsilon)^2|P_M|\leq (1+2\epsilon)^2d_G(s,v)$.

We use the same argument as in Section~\ref{subsec:CorrectnessAdaptive} but with $u\mapsto r(u)$ instead of $u\mapsto\Phi(u)$. It follows that the sum of $r(u,v)$ over all $(u,v)\in P_M$ is $\tilde O(n^2/d)$. Hence, $|\mathcal Q_2| = \tilde O(n^2/(d\delta))$ and since $1+\epsilon = O(1)$, we get
\[
  w(P_{M_+}) \leq \sum_{Q\in\mathcal Q_1}(1+\epsilon)^2|Q| + \sum_{Q\in\mathcal Q_2}(1+\epsilon)\rho|Q|
              = (1+\epsilon)^2|P_M| + \tilde O(n^2\rho D'/(d\delta)).
\]

We can thus ensure $2d + w(P_{M_+})\leq \epsilon D + (1+\epsilon)^2|P_M|$ by adding the constraint $d + n^2\rho D'/(d\delta) = \tilde O(\epsilon D)$.

\subsection{Running time}\label{subsec:TimeSSSPSparse}
We now analyze the running time of \DSSSP and express it as a function of the parameters introduced. We will then choose values for these parameters to minimize running time. We may assume that $m = \Omega(n)$.

By Corollary~\ref{Cor:LowDiamDecompObl}, the total time to maintain \DSCC is $\tilde O(md + mn/d)$ and as shown in Section~\ref{subsec:ModM}, the total time to maintain $\mathcal M$ is $O(km\log n + m\log^2n) = O(\frac 1 \epsilon m\log n + m\log^2n)$.

Next, we bound the total time spent by \DES. For $i = 0,\ldots,k$, we have $w_{M_+}(i) = (1+\epsilon)^i$ and $W_{M_+}(i) = (1+\epsilon)^{i+1}$. All edges that are not super-edges have level $\log_{1+\epsilon}\rho$ so with Lemma~\ref{Lem:ESSuper} applied to $M_+$, we have $m_i = 0$ for $i\neq\log_{1+\epsilon}\rho$ and $m_i = m$ for $i = \log_{1+\epsilon}\rho$.

To bound the $x_i$-values, observe that w.h.p., the number of sampled vertices is $O(n(\log n)/D')$. Furthermore, at any point, if a vertex $V'$ of $M_+$ is incident to at least one level $i$-super-edge then at least one vertex $v'\in V$ contained in $V'$ is sampled. It follows that $x_i = O(nD(\log n)/D')$.

Next, we bound the $\delta_i$-values. Since all super-edges have level at least $j = \lceil\log_{1+\epsilon}D'(1+\epsilon)\rceil$, we can set $\delta_i = 0$ for $i = 0,\ldots,j-1$. Now, let $i\in\{j,\ldots,k\}$ and consider a vertex $V'\in\mathcal C$ of $M_+$ at any given point in time. For any level $i$-super-edges incident to $V'$, the other endpoint $U'$ satisfies $r(U',V')\leq\delta$ and $U'$ is sampled. It follows that the expected number of level $i$-super-edges incident to $V'$ is at most $p\delta = \Theta((\log n)\delta/D')$. Hence, we can pick $\delta_i = \Theta((\log n)\delta/D')$. Also note that $1+W_M(i)-w_M(i) = \Omega(D')$.

Lemma~\ref{Lem:ESSuper} now implies that the total expected time to maintain \DES is
\begin{align*}
  & \tilde O\left(m + \frac{\frac 1 \epsilon Dn + \sum_{i = j}^kx_i\delta_i}{1 + W_M(i) - w_M(i)} + \frac{mD}{1 + W_M(\log_{1+\epsilon}\rho) - w_M(\log_{1+\epsilon}\rho)}\right)\\
  & = \tilde O\left(m + \frac{Dn}{\epsilon D'} + \frac{\delta Dn}{\epsilon (D')^3} + \frac{Dm}{\rho}\right).
\end{align*}

Next, we bound the total expected time to maintain structures \DES$(s')$ over all super-sources $s'$. Recall that $C(u)$ denotes the current set in $\mathcal C$ containing a vertex $u\in V$. At a given point in time, we say that a super-source $s'$ \emph{touches} a vertex $v\in V$ if $s'$ is active and $r(N(s'),C(v)) + |C(v)|\leq\delta$. We need the following lemma.
\begin{Lem}\label{Lem:Touch}
Let $v\in V$ be given. The expected number of super-sources that ever touch $v$ is $O(\log n + \delta(\log n)/D')$.
\end{Lem}
\begin{proof}
Note that when the last edge has been deleted from $E$, each vertex of $\mathcal C$ contains only a single vertex of $V$. For two vertices $u_1,u_2\in V$, denote by $r'(u_1,u_2)$ the value $r(\{u_1\},\{u_2\})$ at this point in time. Similarly, for any $u\in V$, let $r'(u)$ denote $r(\{u\})$ at this point in time. We say that a super-source $s_u$ is \emph{$\delta$-near} to $v$ if $r'(u,v)\leq\delta$. Otherwise, $s'$ is \emph{$\delta$-far} from $v$.

Since the sequence of updates to $E$ and $\mathcal C$ is independent of how vertices of $V$ are sampled by \DSSSP, it follows that the expected number of super-sources that are $\delta$-near to $v$ is $(2\delta+3)p = O(\delta(\log n)/D')$. This also bounds the expected number of $\delta$-near super-sources that ever touch $v$.

Let $S_-$ resp.~$S_+$ be the set of super-sources $s'$ that at some point touch $v$, are $\delta$-far from $v$, and where $N(s')$ precedes resp.~succeeds $\{v\}$ in $\mathcal C$ when $E = \emptyset$. It remains to show that $E[|S_-|] + E[|S_+|] = O(\log n)$. By symmetry, we only need to show that $E[|S_-|] = O(\log n)$.

Assume that $S_-\neq\emptyset$ and consider the first moment in time when a super-source $s_1\in S_-$ touches $v$. Then $s_1$ is active at this point and $r(N(s_1),C(v)) + |C(v)|\leq\delta$. Since $s_1$ is $\delta$-far from $v$, we also have $r(N(s_1),C(v)) + |N(s_1)| + |C(v)| > \delta$. It follows that no $C\in\mathcal C$ preceding $N(s_1)$ in $\mathcal C$ have ingoing edges from super-sources of $S_-$. Any $C\in\mathcal C$ succeeding $N(s_1)$ either succeeds $C(v)$ or satisfies $r(C,C(v)) + |C| + |C(v)|\leq\delta$; in the former case, $C$ has no ingoing edges from super-sources of $S_-$ (by definition of $S_-$) and in the latter case, every super-source having an edge to $C$ is $\delta$-near to $v$. We conclude that all super-sources of $S_-$ have edges ingoing to $N(s_1)$ and $s_1$ has the highest priority among these.

Consider the next point in time when another super-source $s_2\in S_-$ touches $v$. Just prior to this, $N(s_1)$ must have been split. Just after the split, $N(s_1)\neq N(s_2)$ since otherwise, $s_1$ would remain active. We claim that in fact, $N(s_1)$ must precede $N(s_2)$ in $\mathcal C$. To see this, observe that since $s_2$ touches $v$, we have $r(N(s_2),C(v)) + |C(v)|\leq\delta$ and since $s_1$ is $\delta$-far from $v$, we have $r(N(s_1),C(v)) + |N(s_1)| + |C(v)| > \delta$. This is only possible if $N(s_1)$ precedes $N(s_2)$ in $\mathcal C$.

Repeating this argument gives a sequence $s_1,s_2,\ldots,s_i$ of distinct super-sources from $S_-$; let $t_1,t_2,\ldots,t_i$ denote the sampled vertices of $V$ that they are incident to. For $j = 1,\ldots,i$, let $V_j$ denote the subset $N(s_j)$ of $V$ at the moment when $s_j$ starts touching $v$. Extending the above observations, we get that $V_1\supset V_2\supset\cdots\supset V_i$ and that for $j = 2,\ldots,i$, $r'(t_{j-1}) < r'(u)$ for all $u\in V_j$. Since $s_j$ has highest priority among the super-sources with edges ingoing to $V_j$, we can view $t_j$ as being chosen uniformly at random among the sampled vertices in $V_j$.

Denote by $T$ the set of all sampled vertices in $V$. It follows from the above that with probability greater than $1/2$, $|V_j\cap T|\leq\frac 1 2 |V_{j-1}\cap T|$, for $j = 2,\ldots,i$. Hence, $E[|S_-|] = E[i] = O(\log n)$, showing the desired.
\end{proof}
For any structure \DES$(s')$, consider an edge $e$ of $M(s')$ and let $V_1$ and $V_2$ denote its endpoints. By definition of $M(s')$, either $r(N(s'),V_1) + |V_1|\leq\delta$ or $r(N(s'),V_2) + |V_2|\leq\delta$. Then the edge in $E$ corresponding to $e$ has at least one endpoint that touches $s'$. It follows from Lemma~\ref{Lem:Touch} that each edge of $E$ and each vertex of $V$ is processed by no more than $O(\log n + \delta(\log n)/D')$ structures \DES$(s')$ in expectation. By Lemma~\ref{Lem:ESSuper}, the total expected time to maintain all structures \DES$(s')$ is
\[
  \tilde O\left(\frac{\delta m}{D'} + \frac{\delta n}{\epsilon} + \frac{\delta m}{\epsilon}\right)
= \tilde O\left(\frac{\delta m}{\epsilon}\right).
\]

Summing up, we have a total expected time bound for \DSSSP of
\[
  \tilde O\left(dm + \frac{mn}{d} + \frac{\delta m}{\epsilon} + \frac{Dn}{\epsilon D'} + \frac{\delta Dn}{\epsilon(D')^3} + \frac{Dm}{\rho}\right)
\]
under the constraint $d + n^2\rho D'/(d\delta) = \tilde O(\epsilon D)$ from Section~\ref{subsec:CorrectnessSparse}. In Appendix~\ref{subsec:OptParSSSPSparse}, we optimize the parameters under these constraints to get the update time bound of Theorem~\ref{Thm:SSSPSparse}.

Clearly, answering a distance query takes $O(1)$ time. Reporting an approximate path in time proportional to its length is done almost exactly as in Section~\ref{subsec:PathReport}. The only modification needed is that super edges need to be converted into corresponding paths in $G$ in time proportional to the lengths of these paths. However, this is straightforward as we maintain each such path in a structure \DES$(s')$ which maintains parent pointers for the tree that it maintains. The same analysis as in Section~\ref{subsec:PathReport} now shows the second part of Theorem~\ref{Thm:SSSPSparse}.

\section{Concluding Remarks}\label{sec:ConclRem}
We gave new data structures for the decremental approximate SSSP problem in both weighted and unweighted digraphs. Our time bounds are faster than those of Henzinger et al.~\cite{HenzingerKN142,HenzingerKN15a} for all graph densities. Furthermore, one of our data structures works against an adaptive adversary. This is the first improvement for such an adversary over the $O(mn)$ bound of Even and Shiloach~\cite{EvenS81} dating back to $1981$. It would be interesting to find concrete applications of this result. Given the large number of papers that make use of the data structure of~\cite{EvenS81}, we are optimistic that such applications exist.

It would also be interesting to try and improve our time bounds further. Can $\tilde O(m\sqrt n)$ be achieved, matching the bound for decremental reachability in~\cite{ChechikHILP16}? One of our data structures matches this bound when $m = \tilde O(n^2)$. Also, can our result for an adaptive adversary be extended to also report approximate shortest paths rather than only the approximate distances? Our structure is unable to do so since such a path may reveal the random bits. Finally, can we beat $O(mn)$ deterministically, even for unweighted digraphs with $(1+\epsilon)$-approximation? Such improvements are currently only known for undirected graphs.

\bibliographystyle{abbrv}
\bibliography{paper}

\begin{thebibliography}{10}

\bibitem{BernsteinC16}
A.~Bernstein and S.~Chechik.
\newblock Deterministic decremental single source shortest paths: beyond the
  {$o(mn)$} bound.
\newblock In {\em Proceedings of the 48th Annual {ACM} {SIGACT} Symposium on
  Theory of Computing, {STOC} 2016, Cambridge, MA, USA, June 18-21, 2016},
  pages 389--397, 2016.

\bibitem{BernsteinC17}
A.~Bernstein and S.~Chechik.
\newblock Deterministic partially dynamic single source shortest paths for
  sparse graphs.
\newblock In {\em Proceedings of the Twenty-Eighth Annual {ACM-SIAM} Symposium
  on Discrete Algorithms, {SODA} 2017, Barcelona, Spain, Hotel Porta Fira,
  January 16-19}, pages 453--469, 2017.

\bibitem{BernsteinR11}
A.~Bernstein and L.~Roditty.
\newblock Improved dynamic algorithms for maintaining approximate shortest
  paths under deletions.
\newblock In {\em Proceedings of the Twenty-Second Annual {ACM-SIAM} Symposium
  on Discrete Algorithms, {SODA} 2011, San Francisco, California, USA, January
  23-25, 2011}, pages 1355--1365, 2011.

\bibitem{ChechikHILP16}
S.~Chechik, T.~D. Hansen, G.~F. Italiano, J.~Lacki, and N.~Parotsidis.
\newblock Decremental single-source reachability and strongly connected
  components in {$\tilde O(m\sqrt n)$} total update time.
\newblock In {\em {IEEE} 57th Annual Symposium on Foundations of Computer
  Science, {FOCS} 2016, 9-11 October 2016, Hyatt Regency, New Brunswick, New
  Jersey, {USA}}, pages 315--324, 2016.

\bibitem{EvenS81}
S.~Even and Y.~Shiloach.
\newblock An on-line edge-deletion problem.
\newblock {\em J. {ACM}}, 28(1):1--4, 1981.

\bibitem{HenzingerKN143}
M.~Henzinger, S.~Krinninger, and D.~Nanongkai.
\newblock Decremental single-source shortest paths on undirected graphs in
  near-linear total update time.
\newblock In {\em 55th {IEEE} Annual Symposium on Foundations of Computer
  Science, {FOCS} 2014, Philadelphia, PA, USA, October 18-21, 2014}, pages
  146--155, 2014.

\bibitem{HenzingerKN142}
M.~Henzinger, S.~Krinninger, and D.~Nanongkai.
\newblock Sublinear-time decremental algorithms for single-source reachability
  and shortest paths on directed graphs.
\newblock In {\em Symposium on Theory of Computing, {STOC} 2014, New York, NY,
  USA, May 31 - June 03, 2014}, pages 674--683, 2014.

\bibitem{HenzingerKN141}
M.~Henzinger, S.~Krinninger, and D.~Nanongkai.
\newblock A subquadratic-time algorithm for decremental single-source shortest
  paths.
\newblock In {\em Proceedings of the Twenty-Fifth Annual {ACM-SIAM} Symposium
  on Discrete Algorithms, {SODA} 2014, Portland, Oregon, USA, January 5-7,
  2014}, pages 1053--1072, 2014.

\bibitem{HenzingerKN15}
M.~Henzinger, S.~Krinninger, and D.~Nanongkai.
\newblock Improved algorithms for decremental single-source reachability on
  directed graphs.
\newblock In {\em Automata, Languages, and Programming - 42nd International
  Colloquium, {ICALP} 2015, Kyoto, Japan, July 6-10, 2015, Proceedings, Part
  {I}}, pages 725--736, 2015.

\bibitem{HenzingerKN15a}
M.~Henzinger, S.~Krinninger, and D.~Nanongkai.
\newblock Sublinear-time decremental algorithms for single-source reachability
  and shortest paths on directed graphs.
\newblock {\em CoRR}, abs/1504.07959, 2015.

\bibitem{HenzingerKNS15}
M.~Henzinger, S.~Krinninger, D.~Nanongkai, and T.~Saranurak.
\newblock Unifying and strengthening hardness for dynamic problems via the
  online matrix-vector multiplication conjecture.
\newblock In {\em Proceedings of the Forty-Seventh Annual {ACM} on Symposium on
  Theory of Computing, {STOC} 2015, Portland, OR, USA, June 14-17, 2015}, pages
  21--30, 2015.

\bibitem{HenzingerK95}
M.~R. Henzinger and V.~King.
\newblock Fully dynamic biconnectivity and transitive closure.
\newblock In {\em 36th Annual Symposium on Foundations of Computer Science,
  Milwaukee, Wisconsin, USA, 23-25 October 1995}, pages 664--672, 1995.

\bibitem{King99}
V.~King.
\newblock Fully dynamic algorithms for maintaining all-pairs shortest paths and
  transitive closure in digraphs.
\newblock In {\em 40th Annual Symposium on Foundations of Computer Science,
  {FOCS} '99, 17-18 October, 1999, New York, NY, {USA}}, pages 81--91, 1999.

\bibitem{RodittyZ11}
L.~Roditty and U.~Zwick.
\newblock On dynamic shortest paths problems.
\newblock {\em Algorithmica}, 61(2):389--401, 2011.

\end{thebibliography}

\appendix

\section{Proof of Theorem~\ref{Thm:LowDiamDecomp}}\label{sec:LowDiamDecompProof}
In this section, we give a proof of Theorem~\ref{Thm:LowDiamDecomp} that was omitted in the main part of the paper.

\subsection{Sub-procedures}\label{subsec:SubProc}
Before describing the data structure of Theorem~\ref{Thm:LowDiamDecomp}, we need some sub-procedures which we present in the following.
\begin{Lem}\label{Lem:ThinLayer}
There is a deterministic algorithm \texttt{ThinLayer} which, on input $(H,r,d_1,d_2)$ where $r$ is a vertex in digraph $H = (V,E)$ and where $d_1$ and $d_2$ are integers with $d_2 - d_1\geq 2\lg(|(V|)$, outputs a subset $S$ of $V(H)$ for some $d$ with $d_1\leq d\leq d_2$ and $d_H(r,s) = d$ for all $s\in S$. If $|\{v\in V\vert d_H(r,v)\leq d_1\}|\geq k$ and $|\{v\in V\vert d_H(r,v)\geq d_2\}|\geq k$ for $k\in\mathbb N$ then $S$ is a $q$-quality separator $S$ in $H$ with $q = (d_2-d_1)/(2\lg|V(H)|)$ and each SCC of $H\setminus S$ has vertex size at most $|V| - k$. The running time is $O(1 + \sum_{v\in C}\deg_H(v)) = O(|E| + 1)$ where $C = \{v\in V\vert d_H(r,v)\leq d\}$.
\end{Lem}
\begin{proof}
For all $i\in N_0\cup\{\infty\}$, define the $i$th layer as $L_i = \{v\in V\vert d_H(r,v) = i\}$. Algorithm \texttt{ThinLayer}$(H,r,d_1,d_2)$ finds $S$ as follows. Grow a BFS tree $T$ from $r$ until a layer $L_d$ is found with $d_1\leq d\leq d_2$ such that both $|L_d|\leq\sum_{j < d}|L_j|/q$ and $|L_d|\leq\sum_{j > d}|L_j|/q$. Then output $S = L_d$.

To show correctness, suppose first for contradiction that there is no $d$ with $d_1\leq d\leq d_2$, $|L_d|\leq\sum_{j < d}|L_j|/q$, and $|L_d|\leq\sum_{j > d}|L_j|/q$. Then for $i = d_1,d_1+1\ldots,d_2$, either $\sum_{j\leq i}|L_j| > \sum_{j < i}|L_j|(1+1/q)$ or $\sum_{j\geq i}|L_j| > \sum_{j > i}|L_j|(1+1/q)$. This implies that either $\sum_{j\leq d_2}|L_j| > (1+1/q)^{(d_2 - d_1)/2}$ or $\sum_{j\geq d_2}|L_j| < |V(H)|/(1+1/q)^{(d_2 - d_1)/2}$. We have
\[
d_2 - d_1\geq 2\lg(|(V(H)|) \Rightarrow q\geq 1\Rightarrow (1+1/q)^q\geq 2
\Rightarrow (1+1/q)^{(d_2 - d_1)/2}\geq 2^{(d_2 - d_1)/(2q)} = |V(H)|.
\]
This is a contradiction since $\sum_{j\leq d_2}|L_j| < |V(H)|$ and $\sum_{j\geq d_2}|L_j|\geq k\geq 1$.

We have shown that $|L_d|\leq\sum_{j < d}|L_j|/q$ and $|L_d|\leq\sum_{j > d}|L_j|/q$ for some $d$ with $d_1\leq d\leq d_2$. Since there are no edges of $H$ from $\cup_{j < d} L_j$ to $\cup_{j > d} L_j$, $S = L_d$ is a $q$-quality separator of $H$. Furthermore, both $\cup_{j\leq d} L_j$ and $\cup_{j\geq d} L_j$ contain at least $k$ vertices so every SCC of $H\setminus S$ has size at most $|V(H)| - k$.

To show the running time, growing the BFS tree up to layer $d$ can clearly be done in $O(1 + \sum_{v\in C}\deg_H(v))$ time. Keeping track of the sums $\sum_{j < d'}|L_j|/q$ and $\sum_{j\leq d'}|L_j|/q$ for $d' = 0,\ldots,d$ can be done in additional $O(\sum_{j\leq d}|L_d|)$ time which is $O(1 + \sum_{v\in C}\deg_H(v))$ since each vertex of $\cup_{j\leq d}L_j\setminus\{r\}$ has at least one ingoing edge, namely from its parent in the BFS tree. Since $\sum_{j > d}|L_j|/q = n - \sum_{j\leq d}|L_j|/q$, the termination criterion can be tested within the desired time bound.

\end{proof}
\begin{Cor}\label{Cor:Separator}
There is a deterministic algorithm \texttt{Separator} which, on input $(H,d)$ where $H$ is a digraph of diameter at least $d\geq 2\lg(|V(H)|$, finds in $O(|V(H)| + |E(H)|)$ time a $q$-quality separator in $H$ with $q = d/(2\lg|V(H)|)$.
\end{Cor}
\begin{proof}
Set $d_1 = 0$, $d_2 = d$ and $k = 1$ in Lemma~\ref{Lem:ThinLayer}.
\end{proof}
\begin{Lem}\label{Lem:Partition}
There is a deterministic algorithm \texttt{Partition} which, on input $(H,d)$ where $H = (V,E)$ is a digraph and $d\in\mathbb N$, outputs a set $S\subseteq V$ such that each SCC of $H\setminus E(S)$ has diameter at most $d$. Letting $\mathcal C$ be the collection of these SCCs, $|S|\leq \frac{4\lg|V|}d \sum_{C\in\mathcal C}|C|(\lg|V| - \lg|C|)$ and the running time is $O((|V| + |E|)(1+\frac 1 d\sum_{C\in\mathcal C}|V(C)|(\log|V| - \log|V(C)|)))$.
\end{Lem}
\begin{proof}
\texttt{Partition}$(H,d)$ does the following. First, it computes the SCCs of $H$. For each SCC $C$, it picks an arbitrary vertex $r$ and computes distances in $C$ from $r$ to all vertices of $C$ and distances in $C$ from all vertices of $C$ to $r$. If all distances found are at most $d/2$, \texttt{Partition} finishes processing $C$. Otherwise, it calls \texttt{Separator}$(C,d/2)$ from Corollary~\ref{Cor:Separator}; let $S_C$ be the separator found. Then \texttt{Partition} finds the SCCs of $C\setminus S_C$ and recurses on $(C',d)$ for each such SCC $C'$. The output $S$ of \texttt{Partition}$(H,d)$ is the union of all separators found by \texttt{Separator} in the current and in all recursive calls.

We start by showing correctness. At termination, each SCC $C$ of $H\setminus E(S)$ has diameter at most $d$ since either $C$ consists of a single vertex of $S$ or the algorithm has certified that there is an $r\in V(C)$ such that $B_{\mathit{in}}(r,C,d/2) = B_{\mathit{out}}(r,C,d/2) = C$.

To bound $|S|$, consider the call made to \texttt{Separator}$(C,d/2)$, giving a separator $S_C$ in $H$. By Corollary~\ref{Cor:Separator}, each SCC of $C\setminus S_C$ contains at most $|V(C)| - d|S_C|/(4\lg|V(C)|)$ vertices. In particular, this holds for the at most one such SCC having size greater than $|V(C)/2$. Thus, there is a set $W$ of at least $d|S_C|/(4\lg|V(C)|)$ vertices of $C$ belonging to SCCs of $C\setminus S_C$ of size at most $|V(C)|/2$. We can pay for the separator size $|S_C|$ by letting each such vertex pay at most $4\lg|V(C)|/d$.

Over all recursive calls, a vertex $v\in C$ where $C\in\mathcal C$ is charged at most $\lg(|V|/|V(C)|)$ times. Hence the size of $S$ is at most $\sum_{C\in\mathcal C}4|C|(\lg|V|)(\lg(|V|/|V(C)|))/d$, as desired.

It remains to bound the running time. Consider any recursive call \texttt{Partition}$(H',d)$. Let $E_1$ be the set of edges of $E(H')$ with both endpoints in the same SCC of $H'$ and let $E_2 = E(H')\setminus E_1$. Excluding the time spent in recursive calls, \texttt{Partition}$(H',d)$ takes $O(|V(H')| + |E_1| + |E_2|)$ time. The sum of $|E_2|$ over all recursive calls \texttt{Partition}$(H',d)$ is $O(|E|)$. The sum of $O(|V(H')| + |E_1|)$ over all recursive calls \texttt{Partition}$(H',d)$ in which \texttt{Separator} is not called is $O(|V| + |E|)$. This is within the time bound of the lemma.

The time not yet accounted for is dominated by the total time spent in calls to \texttt{Separator}. By Corollary~\ref{Cor:Separator}, each call \texttt{Separator}$(C,d/2)$ takes $O(|V(C)| + |E(C)|)$ time. We use the same charging scheme as above but distribute a cost of $O(|V(C)| + |E(C)|)$ rather than $|S_C|$ among the vertices of $W$ where $W$ is defined as above. Since $C$ is strongly connected, we have $|S_C|\geq 1$ and thus $|W|\geq d/(4\lg|V(C)|)$ so each vertex of $W$ is charged a cost of no more than $O((|V(C)| + |E(C)|)(\log|V(C)|)/d) = O((|V| + |E|)(\log|V|)/d)$.

It follows that for each $C\in\mathcal C$ and each $v\in C$, $v$ is charged a total cost of $O((|V| + |E|)(\log|V|)(\log(|V|/|V(C)|))/d)$ over the entire execution of \texttt{Partition}$(H,d)$. The sum of this over all $C\in\mathcal C$ and all $v\in C$ is within the time bound of the lemma.

\end{proof}

We will use an extension \texttt{Partition+}$(H,d)$ of the algorithm \texttt{Partition}$(H,d)$ of Lemma~\ref{Lem:Partition} which works as follows. First, a call is made to \texttt{Partition}$(H,d)$, giving set $S$. For each SCC $C$ of $H\setminus S$, an ES-structure for distance $d/2$ is initialized for $C$ with a root chosen uniformly at random in $V(C)$. Then $S$ is returned together with the initialized ES-structures.

\subsection{The data structure}\label{subsec:LowDiamDecomp}
We now present the data structure of Theorem~\ref{Thm:LowDiamDecomp}. We shall ignore its output for now and only focus on the internal maintenance of $S$ and $\mathcal V$.

\paragraph{Initialization:}
The data structure initializes by calling \texttt{Partition+}$(G,d_1)$. Let $S$ be the set of separator vertices returned and let $G_S$ denote $G\setminus E(S)$. Then an instance $\mathcal D$ of the data structure of Chechik et al.~\cite{ChechikHILP16} is initialized for $G_S$. Set $S$ is growing during the sequence of edge deletions, and $\mathcal D$ will maintain the SCCs of $G_S$ at any point during this sequence. We implicitly assume in the following that immediately after the termination of each call to \texttt{Partition+} and to \texttt{ThinLayer}, the separator vertices found are added to $S$ and their incident edges are removed from $\mathcal D$.

\paragraph{Handling an update:}
Now, consider an update consisting of the deletion of an edge $e$. First, $e$ is deleted from $\mathcal D$ and from the at most one ES-structure $\mathcal E_r$ containing $e$. For the new SCC $C$ of $G_S$ containing $r$, all edges not in $C$ are removed from $\mathcal E_r$. Let $n_r$ denote the number of vertices in the graph $H_r$ maintained by $\mathcal E_r$. Note that $H_r$ may contain isolated vertices not belonging to $C$.

At this point, if at most $(d_2 - d_1)/2$ vertices of $H_r$ are missing from the in-tree and at most $(d_2 - d_1)/2$ vertices of $H_r$ are missing from the out-tree of $\mathcal E_r$, the data structure terminates.

Now, suppose the converse. For each new SCC $C$ of $G_S$ not containing the root of an ES-structure, \texttt{Partition+}$(C,d_1)$ is invoked. Next, let $H_r^-$ be the graph obtained from $H_r$ by reversing all edges. The data structure picks a vertex $r'$ uniformly at random from $H_r$ and computes a BFS-tree $T_{\mathit{in}}$ in $H_r^-$ and a BFS-tree $T_{\mathit{out}}$ in $H_r$, both from root $r'$ and up to distance $d_1/2$.

Assume first that at least $(d_2 - d_1)/2$ vertices of $H_r$ are missing from $T_{\mathit{in}}$ and that $|V(B_{\mathit{in}}(r',H_r,d_1/4))|\geq n_r/2$. A call is made to \texttt{ThinLayer}$(H_r^-,r,d_1/4 + 1,d_1/2)$. Let $L$ denote the returned set of vertices. The data structure destroys $\mathcal E_r$, calls \texttt{Partition+}$(C\setminus L,d_1)$, and terminates.

Exactly the same is done if at least $(d_2 - d_1)/2$ vertices of $H_r$ are missing from $T_{\mathit{out}}$ and if $|V(B_{\mathit{out}}(r',H_r,d_1/4))|\geq n_r/2$ but with the call \texttt{ThinLayer}$(H_r,r,d_1/4 + 1,d_1/2)$.

If at this point the data structure has not terminated, it destroys $\mathcal E_r$ and makes a call to \texttt{Partition+}$(C,d_1/8)$.

The above data structure internally maintains $S$ and the SCCs of $G\setminus E(S)$. To ensure the output form of Theorem~\ref{Thm:LowDiamDecomp}, we extend the data structure as follows. After initialization, it outputs a pointer to the set $S$ found as well as pointers to the sets of $\mathcal V$, i.e., the vertex sets of the SCCs of $G\setminus E(S)$. After an update, the data structure outputs the set $S'$ of new vertices of $S$. Let $V'\in\mathcal V$ be the subset containing $S'$. To output vertex sets $W_1,\ldots,W_p$, assume that $p\geq 2$ since otherwise, no update is needed to $\mathcal V$. Also, assume w.l.o.g.~that $|W_i|\leq\frac 1 2 |V'|$ for $i = 1,\ldots,p-1$. Then the data structure adds $W_1,\ldots,W_{p-1}$ to $\mathcal V$, updates $V'$ to $V'\setminus\cup_{i = 1}^{p-1}W_i$, and identifies $W_p$ with $V'$. Finally, it outputs pointers to $S'$ and to $W_1,\ldots,W_p$.

\subsection{Correctness}\label{subsec:CorrectnessLowDiamDecomp}
We now show that the data structure described above correctly maintains the information stated in Theorem~\ref{Thm:LowDiamDecomp}.

The correctness of \texttt{Partition} from Lemma~\ref{Lem:Partition} implies the correctness of the initialization step for the data structure of Theorem~\ref{Thm:LowDiamDecomp}. Vertices added to $S$ during the sequence of updates belong to separators found in calls to \texttt{Partition} and \texttt{ThinLayer}. Consider any one of the calls to \texttt{ThinLayer}$(H',r,d_1/4 + 1,d_1/2)$ where $H'$ is either of the form $H_r$ or $H_r^-$. This gives a $q$-quality separator $S'$ in $H'$ with $q = d_1/(8\lg n)$. Let $\mathcal C$ be the SCCs of $H'\setminus S'$. Using the same argument as in the proof of Lemma~\ref{Lem:Partition}, $|S'|\leq\frac{8\lg|V(H')|}d \sum_{C\in\mathcal C}|C|(\lg|V(H')| - \lg|C|)$.

Combining this with the size bound of Lemma~\ref{Lem:Partition} and a telescoping sums argument over all subgraphs of $G$ and $G^-$ that \texttt{Partition} and \texttt{ThinLayer} are applied to gives the desired bound $|S| = O(n(\log^2 n)/d)$ at any point during the sequence of updates.

Now, consider an update in which an edge $e$ is deleted. With the notation of Section~\ref{subsec:LowDiamDecomp}, assume first that $\mathcal E_r$ is not destroyed in this update. Since the total number of vertices of $H_r$ missing from either the in-tree or the out-tree of $\mathcal E_r$ is at most $2(d_2 - d_1)/2 = d_2 - d_1$, all new SCCs generated contain at most $d_2 - d_1$ vertices. Thus, these SCCs and all future SCCs generated from them must have diameter at most $d_2 - d_1 < d_2$.

Let $C$ be the new SCC containing $r$, consider any $u,v\in V(C)$, and let $P$ be a shortest $u$-to-$v$ path in $C$. Let $u'$ resp.~$v'$ be the first resp.~last vertex of $P$ belonging to the in- resp.~out-tree of $\mathcal E_r$. Then $|P[u,u']| + |P[v',v]|\leq 2(d_2 - d_1)/2 = d_2 - d_1$ and $|P[u',v']|\leq 2d_1/2 = d_1$ so $|P|\leq d_2$. It follows that $C$ has diameter at most $d_2$. Hence, the data structure maintains the invariant that each SCC has diameter at most $d_2$.

If $\mathcal E_r$ is destroyed in the current update then the invariant is clearly maintained due to the calls to \texttt{Partition}.

We have shown the size bound of $S$ and that after the initialization step and each update, each SCC of $G\setminus E(S)$ has diameter at most $d_2$. This shows the correctness of the data structure.

\subsection{Running time}\label{subsec:TimeLowDiamDecomp}
We now bound the total time for initialization and updates spent by our data structure. We may assume that $G$ initially is strongly connected since otherwise we can maintain a data structure separately for each SCC; this ensures that $m\geq n-1$.

As shown by Chechik et al.~\cite{ChechikHILP16}, the total time to maintain $\mathcal D$ is $\tilde O(m\sqrt n)$. The time bound of Lemma~\ref{Lem:Partition} together with a telescoping sums argument shows that the total time for all calls to \texttt{Partition} is $O(mn(\log n)/d_1)$.

Now, consider an update in which an edge $e$ is deleted. We use the same notation here as in Section~\ref{subsec:LowDiamDecomp}.

We have already accounted for the total time to maintain $\mathcal D$. If the deletion of $e$ from $\mathcal D$ causes $C$ to break apart then let $W$ be the union of vertex sets of the new SCCs not containing $r$. We can easily extend $\mathcal D$ to report these vertex sets in time proportional to their total size $|W|$; this follows since $\mathcal D$ explicitly maintains an identifier for each vertex $v$ denoting the SCC containing $v$. Since the set of edges not in $C$ to be removed from $\mathcal E_r$ are exactly those that are incident to $W$, identifying these edges can thus be done in time proportional to their number. The cost of this can be charged to the cost of deleting these edges from $\mathcal E_r$.

Excluding the time to find sets $S'$ and $W_1,\ldots,W_p$, we claim that the total time spent on outputting $S'$ and pointers to $W_1,\ldots,W_p$ over all updates is $O(n\log n)$. The total time to output sets $S'$ is proportional to their total size which is $\tilde O(n/d_1)$. Obtaining $W_1,\ldots,W_p$ and outputting pointers to these sets in a single update takes time $O(\sum_{i = 1}^{p-1}|W_i|)$ where $|W_i|\leq\frac 1 2 |V'|$. Distributing this time cost evenly over all vertices of $\cup_{i = 1}^{p-1}W_i$, each vertex pays only $O(\log n)$ over all updates. This shows the desired $O(n\log n)$ bound.

\paragraph{Bounding the time for maintaining ES-structures:}
The total time cost not yet accounted for is dominated by the total time spent on maintaining ES-structures which we bound in the following. We will make use of the following lemma.
\begin{Lem}\label{Lem:SCCSizeReduction}
Consider an update in which an ES-structure $\mathcal E_r$ is destroyed and let $n_r$ be the number of vertices of the graph maintained by $\mathcal E_r$. Then with probability greater than $1/8$, each new ES-structure created in the update is for an SCC with vertex size at most $\max\{\frac 3 4 n_r, n_r - (d_2 - d_1)/2\}$.
\end{Lem}
\begin{proof}
Consider the ES-tree $\mathcal E_r$ from the update in which it was created until the update in which it was destroyed. To simplify notation, label the updates so that $\mathcal E_r$ was created in update $1$. Note that the graph $H_r$ maintained by $\mathcal E_r$ contains $n_r$ vertices at all times (some of which may become isolated during the sequence of updates).

For $i = 1,\ldots$, let $R_i$ be the set of vertices $r'\in V(H_r)$ whose ES-structures would be destroyed at the end of update $i$ had $r'$ been chosen instead of $r$ in update $1$. Let $U_i = \cup_{j = 1}^i R_j$, let $r_i = |R_i|$, and let $u_i = |U_i|$. Let $t$ be the smallest value such that $u_t = n_r$. Note that $R_1,\ldots,R_t$ form a partition of $V(H_r)$. For $i = 1,\ldots,t$ and for each $r'\in V(H_r)\setminus U_{i-1}$, we have $\Pr(r = r' | r\notin U_{i-1}) = 1/(n_r - u_{i-1})$ at the beginning of update $i$; this follows since if $r\notin U_{i-1}$, the output of the data structure up until, and not including, update $i$ is independent of the choice of $r$ in $V(H_r)\setminus U_{i-1}$.

Let $t_r$ be the random variable denoting the update in which $\mathcal E_r$ was destroyed. For $i = 1,\ldots,t$,
\[
  \Pr(t_r\geq i) = \prod_{j = 1}^{i-1} \Pr(r\notin R_j | r\notin U_{j-1})
               = \prod_{j=1}^{i-1}\frac{n_r - u_j}{n_r - u_{j-1}}
               = 1 - u_{i-1}/n_r.\]

In particular, if we pick the unique index $i$ such that $u_i\geq n_r/2$ and $u_{i-1} < n_r/2$ then $\Pr(t_r\geq i) > 1/2$.

Pick a vertex $r'\in R_j$ for some $j\leq i$ and consider what would have happened had $r'$ been picked as the root in update $1$. More than $(d_2 - d_1)/2$ vertices must be missing from either the in-tree or the out-tree of $\mathcal E_{r'}$ just prior to this structure being destroyed in update $j$. Also, no vertices of $H_r$ have been added to $S$ during updates $1$ through $j-1$.

Now, consider again our situation with $r$ being picked in update $1$. It follows from the above that at the end of update $j$, there are more than $(d_2 - d_1)/2$ vertices of $H_r$ which cannot be reached in either the in-tree or the out-tree of $r'$ in $H_r$ up to distance $d_1/2$. Since updates consist of deletions only, this is also the case for $r'$ in any later update.

Assume in the following that $t_r\geq i$; as shown above, this event happens with probability greater than $1/2$. Then just prior to $\mathcal E_r$ being destroyed in update $t_r$, let $W$ be the set of vertices $r'\in V(H_r)$ for which there are more than $(d_2 - d_1)/2$ vertices in $H_r$ all unreachable in the in-tree or all unreachable in the out-tree of $r'$ in $H_r$ up to distance $d_1/2$. By the choice of $i$ and by the assumption that $t_r\geq i$, we have $|W|\geq n_r/2$.

We consider two possible cases at the end of update $t_r$, one of which must occur:
\begin{enumerate}
\item For ${}\geq n_r/4$ vertices $r'\in W$, $|V(B_{\mathit{in}}(r',H_r,d_1/4))| \geq n_r/2\land |V(B_{\mathit{out}}(r',H_r,d_1/4))|\geq n_r/2$,
\item For ${}\geq n_r/4$ vertices $r'\in W$, $|V(B_{\mathit{in}}(r',H_r,d_1/4))| < n_r/2\lor |V(B_{\mathit{out}}(r',H_r,d_1/4))| < n_r/2$.
\end{enumerate}
In the first case, a call to \texttt{ThinLayer} in update $t_r$ will be executed with probability at least $1/4$ (conditioned on the event $t_r\geq i$ assumed above). It follows from the description of our data structure and from Lemma~\ref{Lem:ThinLayer} with $k = \min\{n_r/2,(d_2 - d_1)/2\}$ that every new ES-structure created in an update in which \texttt{ThinLayer} is applied is for an SCC with vertex size at most $\max\{n_r/2,n_r - (d_2 - d_1)/2\}\leq\max\{\frac 3 4 n_r,n_r - (d_2-d_1)/2\}$.

Now, assume the second case. We will show that if a call to \texttt{Partition}$(C,d_1/8)$ is made in update $t_r$, each SCC of $C\setminus S$ contains less than $\frac 3 4 n_r$ vertices (if such a call is not made, a call is instead made to \texttt{ThinLayer}). Let $C$ denote the SCC in $H_r$ containing $r$ just prior to picking a random root in update $t_r$. We may assume that $|V(C)|\geq \frac 3 4 n_r$ since otherwise, every SCC has size less than $\frac 3 4 n_r$ vertices.

Since at least $n_r/4$ vertices of $W\subseteq V(H_r)$ has the second property above and since $V(C)\subseteq V(H_r)$ and $|V(C)|\geq\frac 3 4 n_r$, at least one vertex $r'\in V(C)$ exists having that property. Since $|V(B_{\mathit{in}}(r',H_r,d_1/4))| < n_r/2$ or $|V(B_{\mathit{out}}(r',H_r,d_1/4))| < n_r/2$ and since $|V(C)|\geq\frac 3 4 n_r$, $C$ has diameter greater than $d_1/4$. Hence, for any $u\in V(C)$, there is a $v\in V(C)$ such that either $d_C(u,v) > d_1/8$ or $d_C(v,u) > d_1/8$. It follows from the above that \texttt{Partition}$(C,d_1/8)$ calls \texttt{Separator}$(C,d_1/16)$ and recurses. The same argument shows that if in any recursive call an SCC $C'$ exists of size at least $\frac 3 4 n_r$ then \texttt{Separator}$(C',d_1/16)$ is called. Hence, when \texttt{Partition}$(C,d_1/8)$ terminates, each SCC of $C\setminus S$ contains less than $\frac 3 4 n_r$ vertices.
\end{proof}



We are now ready to bound the expected total time spent on maintaining ES-structures. Consider an update in which an ES-structure $\mathcal E_r$ is destroyed and let $n_r$ be the number of vertices of the graph maintained by $\mathcal E_r$. By Lemma~\ref{Lem:SCCSizeReduction}, with probability greater than $1/8$, each new ES-structure created in the update is for an SCC with vertex size at most $\max\{\frac 3 4 n_r, n_r - (d_2 - d_1)/2\}$. Hence, the total expected time spent on maintaining ES-structures that were created in an update that destroyed another ES-structure is $\tilde O(md_1\cdot n/(d_2 - d_1))$. This bounds the total time to maintain ES-structures since a new ES-structure can only be created during initialization or in an update in which another ES-structure is destroyed. We have now completed the proof of Theorem~\ref{Thm:LowDiamDecomp}.

\section{Implementation of the Multigraph Structure}\label{sec:ImplementationM}
In this section, we give the implementation details for multigraph structure $\mathcal M$ and show that this implementation has the performance stated in Lemma~\ref{Lem:M}. Vertices of $V$ are assigned unique indices in $\{0,\ldots,|V|-1\}$. For the inital graph $M$, each vertex is similarly assigned a unique id from $\{0,\ldots,|V_M| - 1\}$. A counter is then initialized to $|V_M|$ and whenever a new vertex appears in $M$ due to a \texttt{Split}$(V',\{W_1,\ldots,W_{p-1}\})$-operation, the counter is incremented and the new vertex is given the current counter value as its id; the vertex $W_p$ is given the same index as $V'$.

Given the above assignment, each ordered vertex pair of $M$ has an associated ordered index pair and lexicographically ordering these thus defines an ordering of all ordered vertex pairs of $M$. $\mathcal M$ keeps a balanced binary search tree $T$ for this ordering and contains all ordered vertex pairs $(V_1,V_2)$ for which there is at least one edge of $E$ from set $V_1$ to set $V_2$. The node of $T$ for each such pair $(V_1,V_2)$ is associated with a min-priority queue $Q(V_1,V_2)$ containing all edges of $E$ from $V_1$ to $V_2$ keyed by their levels. If $V_1\neq V_2$, the level of the minimum element of $Q(V_1,V_2)$ is thus the level of the representative edge $(V_1,V_2)$ and $\mathcal M$ stores this representative edge and its level with the node of $T$ representing $(V_1,V_2)$.

$\mathcal M$ maintains a pointer from each edge of $E$ to its entry in the queue containing it. For each pair $(V',i)$ where $V'\in V_M$ and $i\in\{0,\ldots,k\}$, $\mathcal M$ maintains pointers to $E_{\mathit{in}}(V',i)$ and $E_{\mathit{out}}(V',i)$. $\mathcal M$ maintains a pointer from each representative edge to its entry in the at most one $E_{\mathit{in}}$-list containing it and its entry in the at most one $E_{\mathit{out}}$-list containing it. Pointers are also kept from vertices of $V_M$ to their corresponding entries in $V_{\mathit{in}}$- and $V_{\mathit{out}}$-lists.

In addition, $\mathcal M$ maintains a mapping from indices of $V$ to indices of $V_M$ where the index of a vertex $v$ is mapped to the vertex of $V_M$ whose corresponding set contains $v$. This mapping also allows for mapping an edge $(v_1,v_2)$ of $E$ to the corresponding vertex pair $(V_1,V_2)$ in $M$ where $v_1\in V_1$ and $v_2\in V_2$.

Finally, $\mathcal M$ maintains the lengths of all $E_{\mathit{in}}$- and $E_{\mathit{out}}$-lists.

\paragraph{Implementing \texttt{Init}$(G = (V,E), \{V_1,\ldots,V_{\ell}\},\{E_0,\ldots,E_k\})$:}
This operation starts by initializing the indices of vertices of $V$ and $V_M$. It then initializes $T$ and the queues associated with nodes of $T$ and sets up all the $E_{\mathit{in}}$-, $E_{\mathit{out}}$-, $V_{\mathit{in}}^{\Delta_i}(i)$, and $V_{\mathit{out}}^{\Delta_i}(i)$-lists. Finally, pointers as described above are obtained and stored.

\paragraph{Implementing \texttt{Delete}$(e)$:}
This operation first deletes $e$ from $E$. It then identifies the corresponding edge $(V_1,V_2)$ in $E_M$ and removes $e$ from $Q(V_1,V_2)$; if $Q(V_1,V_2)$ becomes empty, the node of $T$ storing $(V_1,V_2)$ is deleted. If $V_1 = V_2$, no further updates are done so assume $V_1\neq V_2$.

If $e$ was the minimum element of $Q$, the representative edge $(V_1,V_2)$ is removed from $E_{\mathit{in}}(V_2,i)$ and from $E_{\mathit{out}}(V_1,i)$ where $i$ is its level. If this causes $|E_{\mathit{in}}(V_2,i)|\leq\Delta_i$ then $V_2$ is removed from $V_{\mathit{in}}^{\Delta_i}(i)$ and if $|E_{\mathit{out}}(V_1,i)|\leq\Delta_i$ then $V_1$ is removed from $V_{\mathit{out}}^{\Delta_i}(i)$. If $Q$ is non-empty, let $j$ be its new minimum key value. Then representative edge $(V_1,V_2)$ is given its new level $j$ and is inserted into lists $E_{\mathit{in}}(V_2,j)$, $E_{\mathit{out}}(V_1,j)$, $V_{\mathit{in}}^{\Delta_i}(j)$, and $V_{\mathit{out}}^{\Delta_i}(j)$.

\paragraph{Implementing \texttt{Increase}$(e,i)$:}
Let $j = \ell(e)$. First, $\ell(e)$ is updated to $i$ and the corresponding vertex pair $(V_1,V_2)$ in $M$ is identified. The key value of $e$ in $Q(V_1,V_2)$ is increased to $i$. If $V_1\neq V_2$ then, depending on whether this causes the minimum key in $Q(V_1,V_2)$ and hence the level of representative edge $(V_1,V_2)$ to change, sets $E_{\mathit{in}}(V_2,i)$, $E_{\mathit{in}}(V_2,j)$, $E_{\mathit{out}}(V_1,i)$, $E_{\mathit{out}}(V_1,j)$, $V_{\mathit{in}}^{\Delta_i}(i)$, $V_{\mathit{in}}^{\Delta_i}(j)$, $V_{\mathit{out}}^{\Delta_i}(i)$, and $V_{\mathit{out}}^{\Delta_i}(j)$ are updated in a manner similar to what is described above for \texttt{Delete}.


\paragraph{Implementing \texttt{Split}$(V',\{W_1,\ldots,W_{p-1}\})$:}
First, for each edge $e$ of $E$ incident to $\cup_{i = 1}^{p-1}W_i$, the following is done. Identify the vertex pair $(V_1,V_2)$ of $M$ corresponding to $e$ and remove $e$ from $Q(V_1,V_2)$. If this causes the minimum key value in $Q(V_1,V_2)$ to change or causes $Q(V_1,V_2)$ to become empty, updates similar to those for \texttt{Delete} and \texttt{Increase} are made.

Next, $W_1,\ldots,W_{p-1}$ are assigned new indices using the counter as described above and the mapping from vertices of $V$ belonging to $\cup_{i = 1}^{p-1} W_i$ to vertices of $V_M$ are updated accordingly. Next, for each edge $e$ of $E$ incident to $\cup_{i = 1}^{p-1}W_i$, identify the vertex pair $(V_1,V_2)$ of $M$ corresponding to $e$ and insert $e$ into $Q(V_1,V_2)$. Further updates as described above are done if this causes the minimum key value of $Q(V_1,V_2)$ to change.




\paragraph{Proving Lemma~\ref{Lem:M}:}
We now show that the implementation of $\mathcal M$ above satisfies Lemma~\ref{Lem:M}.

Correctness follows easily from the above description. For the running time, we first focus on the \texttt{Init}-operation. Using a red-black tree for $T$, the total number of nodes of $T$ is bounded by $m$ and can thus be set up in $O(m\log n)$ time, excluding the time to prepare the auxiliary data associated with each node of $T$. We use a binary heap implementation for queues $Q(V_1,V_2)$; the total time to build these is $O(m)$ since each edge is in exactly one queue. Once these have been initialized, the representative edges and their levels can be identified in $O(m)$ time and within $O(m+n)$ time, the remaining lists and pointers can be initialized as well.

Each \texttt{Delete}- and each \texttt{Increase}-operation can be executed in $O(\log n)$ time. This follows since it involves a constant number of queue updates, at most one deletion from $T$, and a constant number of pointer and linked list updates. The total number of \texttt{Delete}-operations is at most $m$ and the total number of \texttt{Increase}-operations is at most $km$. Hence, the total time for all \texttt{Delete}- and \texttt{Increase}-operations is $O(km\log n)$.

During a \texttt{Split}$(V',\{W_1,\ldots,W_{p-1}\})$-operation, a constant number of updates to queues, to $T$, and to pointers and linked lists are performed for each edge of $E$ incident to $\cup_{i = 1}^{p-1}W_i$. Since $|W_i|\leq\frac 1 2 |V'|$ for $i = 1,\ldots,p-1$, any single edge of $E$ is considered only $O(\log n)$ times in all \texttt{Split}-operations. Hence, the total time for all these operations is $O(m\log^2n)$. This completes the proof.


\section{Proof of Theorem~\ref{Thm:LowDiamDecompObl}}\label{sec:LowDiamDecompOblProof}
To show the theorem, we need the following lemma.
\begin{Lem}\label{Lem:FastSep}
There is a deterministic algorithm \texttt{FastSeparator} which takes as input $(H,u,d)$ where $H$ is a digraph, $u\in V(H)$, and $d\geq 2\lg(|V(H)|)$, such that $d_G(u,v)\geq d$ for at least one $v\in V(H)$. The output is a $q$-quality separator $S$ in $H$ with $q = d/(2\lg|V(H)|)$ and $d_H(u,s) = d'$ for all $s\in S$ and some $d'\leq d$. The running time is $O(1 + \sum_{v\in C}\deg_H(v))$ where $C = \{v\in V(H)\vert d_H(u,v)\leq d'\}$.
\end{Lem}
\begin{proof}
\texttt{FastSeparator}$(H,u,d)$ applies \texttt{ThinLayer}$(H,u,0,d)$. The result now follows from Lemma~\ref{Lem:ThinLayer} with $r = u$, $d_1 = 0$, $d_2 = d$, and $k = 1$. 
\end{proof}

We now present the data structure of Theorem~\ref{Thm:LowDiamDecompObl}. Since we are going to apply Lemma~\ref{Lem:FastSep} with parameter $d/2$, we assume that $d\geq 4\lg n$; this is w.l.o.g.~since if $d < 4\lg n$, we could use a trivial data structure which keeps $S = V$ at all times.

The data structure initializes by calling \texttt{Partition+}$(G,d)$. Let $S$ be the set of separator vertices returned and let $G_S$ denote $G\setminus S$. We implicitly assume in the following that immediately after the termination of each call to \texttt{Partition+} and to \texttt{FastSeparator}, the separator vertices found are added to $S$. Also, we implicitly assume for each ES-structure $\mathcal E_r$ that as soon as a vertex becomes unreachable from $r$ in either $H_r$ or in $H_r^-$, all its incident edges are removed from $\mathcal E_r$.

Now, consider an update consisting of the deletion of an edge $e$. If there is an ES-structure $\mathcal E_r$ containing $e$, an iterative procedure is applied which maintains a queue $Q$. Let $H_r$ and $H_r^-$ be the graphs maintained by $\mathcal E_r$. At any time, $Q$ consists of the set of vertices $v$ where $\max\{d_{H_r}(r,v), d_{H_r^-}(r,v)\} > d/2$.

The iterative procedure executes as follows as long as $Q\neq\emptyset$. Extract an arbitrary vertex $v$ from $Q$. Assume first that $d_{H_r}(r,v) > d/2$. Then \texttt{FastSeparator}$(H_r^-,v,d/4)$ is applied, giving a (possibly empty) set of vertices $S_v$. Each edge incident to $S_v$ is removed from $\mathcal E_r$ and $Q$ is updated accordingly. The procedure then continues to the next iteration. The other case where $d_{H_r^-}(r,v) > d/2$ (and $d_{H_r}(r,v)\leq d/2$) is handled in the same way except that the call is made to \texttt{FastSeparator}$(H_r,v,d/4)$. This completes the description of the iterative procedure. If at termination of this procedure, more than half of the vertices $v\in V(H_r) = V(H_r^-)$ satisfy $\min\{d_{H_r}(r,v),d_{H_r^-}(r,v)\}\geq d/2$, $\mathcal E_r$ is destroyed and \texttt{Partition+}$(C_r,d/4)$ is invoked where $C_r$ is the SCC of $G_S$ containing $r$.

Next, for each new SCC $C$ of $G_S$ not containing an ES-structure, \texttt{Partition+}$(C,d/4)$ is invoked. This completes the description of the data structure of Theorem~\ref{Thm:LowDiamDecompObl}. It provides its output in the same manner as in Section~\ref{sec:LowDiamDecomp}.






\subsection{Correctness}
We first show the invariant that at initialization and after each edge deletion, each SCC $C$ of $G_S$ has an associated ES-structure $\mathcal E_r$ such that $V(C)$ is exactly the set of vertices reachable from $r$ in both $H_r$ and $H_r^-$ and all these vertices are within distance $d/2$ from $r$ in both graphs.

Lemma~\ref{Lem:Partition} and the description of \texttt{Partition+} implies that the invariant holds at initialization so consider an update in which an edge $e$ is deleted and assume that the invariant holds at the beginning of this update. If $e$ is not in any SCC of $G_S$ then the invariant clearly holds at the end of the update. Otherwise, there is an ES-structure $\mathcal E_r$ containing $e$. Whenever a vertex leaves the SCC of $G_S$ containing $r$, it clearly also becomes unreachable from $r$ in either $H_r$ or $H_r^-$. The calls to \texttt{FastSeparator} ensure that whenever a vertex $v$ becomes unreachable from $r$ in either $H_r$ or $H_r^-$, then either there is no path from $r$ to $v$ in $G_S$ or there is no path from $v$ to $r$ in $G_S$. Hence, $v$ is no longer in the same SCC of $G_S$ as $r$.

We have shown that at termination of the iterative procedure, the SCC $C_r$ containing $r$ contains exactly the vertices of $V(H_r) = V(H_r^-)$ that are reachable from $r$ in both $H_r$ and $H_r^-$. Since $Q = \emptyset$ at this point, all these vertices are within distance $d/2$ from $r$ in both $H_r$ and $H_r^-$. This shows the invariant for $C_r$ and if \texttt{Partition+}$(C_r,d/4)$ is applied, the invariant clearly also holds for every SCC generated in this call. Simililarly, for every new SCC $C$ of $G_S$ not contained in $C_r$, \texttt{Partition+}$(C,d/4)$ is applied. Hence, the invariant holds for all SCCs of $G_S$ when the data structure finishes processing the deletion of $e$.

By the invariant, it follows that at initialization and after each update, each SCC of $G_S$ has diameter at most $2d/2 = d$. The bound on $|S|$ follows using arguments similar to those in Section~\ref{subsec:CorrectnessLowDiamDecomp}. This shows the correctness of the data structure of Theorem~\ref{Thm:LowDiamDecompObl}.

\subsection{Running time}
To bound the running time, consider some point during the sequence of updates in which a new SCC $C$ is generated and thus an ES-structure $\mathcal E_r$ is initialized for $C$ where $r$ is chosen uniformly at random from $V(C)$. Note that until $\mathcal E_r$ is destroyed, $V(H_r) = V(H_r^-) = V(C)$ since vertices are never deleted from $\mathcal E_r$, only edges. Hence, prior to the update in which $\mathcal E_r$ is destroyed, each new SCC of $G_S$ generated after an edge deletion in $\mathcal E_r$ must have size at most $|V(C)|/2$ since at least half the vertices of $V(H_r) = V(C)$ satisfy $\min\{d_{H_r}(r,v),d_{H_r^-}(r,v)\} < d/2$ and these vertices induce a strongly connected subgraph of $G_S$.

Now consider such an update in which an ES-structure $\mathcal E_r$ is destroyed. Let $C$ be the SCC for which $\mathcal E_r$ was created in an earlier update. Let $C'$ be $C$ intersected with the current edge set $E$ and define for each $r'\in V(C)$, $K_{r'}$ to be the subgraph of $G_S$ induced by the set of vertices $v$ with $\min\{d_{H_{r'}}(r',v),d_{H_{r'}^-}(r',v)\} < d/2$. We have $|V(K_r)| < |V(C)|/2$. Since for each $s\in S$, either $d_{H_r}(r,s) \geq d/2$ or $d_{H_r^-}(r,s) \geq d/2$, $K_r$ is also an induced subgraph of $G$.

Order the vertices $r'\in V(C)$ by the update in which $\mathcal E_{r'}$ would have been destroyed, had $r'$ been picked instead of $r$ when $C$ was generated. Since $r$ was chosen uniformly at random from $V(C)$, it is among the last half of vertices w.r.t.~this order with probability at least $1/2$; assume this event in the following.

When $\mathcal E_r$ is destroyed then for at least half the vertices $r'\in V(C)$, $|V(K_{r'})| < |V(C)|/2$. We claim that this implies that any subgraph $K$ of $C$ with $|V(K)|\geq |V(C)|/2$ has diameter greater than $d/4$. This follows since $H$ contains at least one vertex $r'$ with $|V(K_{r'})| < |V(C)|/2$. But then $V(K_{r'})$ is a strict subset of $V(H)$ so $H$ must have diameter greater than $d/4$, as desired.

It follows from this that in the update in which $\mathcal E_r$ is destroyed, every SCC which is a subgraph of $C$ has size at most $|V(C)|/2$ due to calls to \texttt{Partition+} with parameter $d/4$.

The total cost of maintaining an ES-structure $\mathcal E_r$ can be paid for by charging each vertex of $V(H_r)$ a cost of $d$ times its degree in the initial graph $G$. It follows from the above that in expectation, each vertex is charged this amount at most $\lg n$ times. Hence, the total expected time to maintain ES-structures is $O(md\log n)$. Similarly, considering the collection of all graphs that \texttt{Partition+} is applied to, each vertex belongs to at most $\lg n$ of these graphs in expectation. By Lemma~\ref{Lem:Partition}, the total expected time for all calls to \texttt{Partition+} is $\tilde O(mn/d)$. By Lemma~\ref{Lem:FastSep}, the cost of a call to \texttt{FastSeparator}$(H_r,v,d/4)$ runs in time proportional to the total degree of vertices that become unreachable from $r$ in either $H_r$ or $H_r^-$. Hence, the total time for all calls to \texttt{FastSeparator} is dominated by the time spent on maintaining ES-structures. This completes the proof of Theorem~\ref{Thm:LowDiamDecompObl}.

\section{Optimizing parameters}
In this section, we optimize parameters for our data structures to get the desired time bounds.

\subsection{Optimizing parameters for Theorem~\ref{Thm:SSSPAdaptive}}\label{subsec:OptParSSSPAdaptive}
In Section~\ref{subsec:TimeAdaptive}, we obtained the constraints $d_2 + n^2/(d_1\tau) = \tilde O(\epsilon D)$ and $d_2\geq 2d_1$ and a total update time of
\[
  \tilde O(m\sqrt n + mn/d_1 + mnd_1/d_2 + Dn\tau).
\]
We minimize this time bound by maximizing $d_2$ and minimizing $\tau$, i.e., we we set $d_2 = \tilde\Theta(\epsilon D)$ and $\tau = \tilde\Theta(n^2/(\epsilon Dd_1))$, thereby satisfying the constraint $d_2 + n^2/(d_1\tau) = \tilde O(\epsilon D)$. This gives a time bound of
\[
  \tilde O(m\sqrt n + mn/d_1 + mnd_1/(\epsilon D) + n^3/(\epsilon d_1))
= \tilde O(m\sqrt n + mnd_1/(\epsilon D) + n^3/(\epsilon d_1))
\]
We minimize this bound by setting $d_1 = n\sqrt D/\sqrt m$. The constraint $d_2\geq 2d_1$ is ensured by requiring $n\sqrt D/\sqrt m = \tilde O(\epsilon D)$, i.e., $D = \tilde\Omega(n^2/(\epsilon^2m))$. Assuming this in the following, the time bound simplifies to
\[
  \tilde O(m\sqrt n + \sqrt m n^2/(\epsilon\sqrt D)).
\]
The $O(mD)$ algorithm of Even and Shiloach~\cite{EvenS81} is no slower than this when $D = \tilde O(\sqrt n)$ and when $D = \tilde O((n^2/(\epsilon\sqrt m))^{2/3}) = \tilde O(n^{4/3}/(\epsilon^{2/3}m^{1/3}))$, i.e., when $D = \tilde O(n^{4/3}/(\epsilon^{2/3}m^{1/3}))$ (since $\epsilon\leq 1$). For these values of $D$, their algorithm runs in time $\tilde O(m^{2/3}n^{4/3}/\epsilon^{2/3})$. When $D = \tilde\Omega(n^{4/3}/(\epsilon^{2/3}m^{1/3}))$, our algorithm runs within the same time bound. To ensure the constraint $D = \tilde\Omega(n^2/(\epsilon^2m))$, we run their algorithm when $D = \tilde O(n^2/(\epsilon^2 m))$, which takes time $\tilde O(n^2/\epsilon^2)$. This shows Theorem~\ref{Thm:SSSPAdaptive}.

\subsection{Optimizing parameters for Theorem~\ref{Thm:SSSPDense}}\label{subsec:OptParSSSPDense}
We have a time bound of $\tilde O(md + mn/d + Dn\tau)$ and the constraint $d + n^2/(\tau d)$. We set $\tau = \tilde\Theta(n^2/(\epsilon d D))$, giving a time bound of
\[
\tilde O(md + mn/d + n^3/(\epsilon d)) = \tilde O(md + n^3/(\epsilon d))
\]
under the constraint $d = \tilde O(\epsilon D)$. If $D\geq n^{3/2}/(\sqrt m\epsilon^{3/2})$, we pick $d = \tilde\Theta(n^{3/2}/(\sqrt m\sqrt\epsilon))$ while satisfying the constraint and this gives a time bound of $\tilde O(\sqrt m n^{3/2}/\epsilon^{3/2})$. If $D < n^{3/2}/(\sqrt m\epsilon^{3/2})$, we apply the data structure of Even and Shiloach~\cite{EvenS81}, giving a time bound of $O(mD) = O(\sqrt m n^{3/2}/\epsilon^{3/2})$. This shows the total update time bound of Theorem~\ref{Thm:SSSPDense}.

\subsection{Optimizing parameters for Theorem~\ref{Thm:SSSPSparse}}\label{subsec:OptParSSSPSparse}
Summing up, we have a total expected time bound for \DSSSP of
\[
  \tilde O\left(dm + \frac{mn}{d} + \frac{\delta m}{\epsilon} + \frac{Dn}{\epsilon D'} + \frac{\delta Dn}{\epsilon(D')^3} + \frac{Dm}{\rho}\right).
\]
We now minimize this bound under the constraint $d + n^2\rho D'/(d\delta) = \tilde O(\epsilon D)$ from Section~\ref{subsec:CorrectnessSparse}.

We minimize the time bound by minimizing $\delta$ so we set $\delta = \tilde \Theta(\rho D' n^2/(\epsilon dD))$ such that the constraint is satisfied and simplifies to $d = \tilde O(\epsilon D)$. The time bound becomes
\[
  \tilde O\left(dm + \frac{mn}{d} + \frac{\rho D' mn^2}{\epsilon^2 dD} + \frac{Dn}{\epsilon D'} + \frac{\rho n^3}{\epsilon^2(D')^2d} + \frac{Dm}{\rho}\right).
\]
We pick $D'$ to balance the third and fifth terms, giving $D' = \tilde\Theta((Dn/m)^{1/3})$. However, we need $D'\geq 1$ so this assumes that $m = \tilde O(Dn)$. Below, we consider the case when this fails. The running time simplifies to
\[
  \tilde O\left(dm + \frac{mn}{d} + \frac{\rho m^{2/3}n^{7/3}}{\epsilon^2dD^{2/3}} + \frac{m^{1/3}}{\epsilon} + \frac{Dm}{\rho}\right).
\]
We pick $d$ to balance the first and third terms, giving $d = \tilde\Theta((\sqrt\rho n^{7/6})/(\epsilon D^{1/3}m^{1/6}))$. We now get a time bound of
\[
  \tilde O\left(\frac{\sqrt\rho m^{5/6}n^{7/6}}{\epsilon D^{1/3}} + \frac{\epsilon D^{1/3}m^{7/6}}{\sqrt\rho n^{1/6}} + \frac{m^{1/3}}{\epsilon} + \frac{Dm}{\rho}\right).
\]
We pick $\rho$ to balance the first and last terms, giving $\rho = \tilde\Theta((\epsilon^{2/3}D^{8/9}m^{1/9})/n^{7/9})$. The time bound simplifies to
\[
  \tilde O\left(\frac{D^{1/9}m^{8/9}n^{7/9}}{\epsilon^{2/3}} + \frac{\epsilon^{2/3}m^{10/9}n^{2/9}}{D^{1/9}} + \frac{m^{1/3}}{\epsilon}\right)
= \tilde O\left(\frac{D^{1/9}m^{8/9}n^{7/9}}{\epsilon^{2/3}} + \frac{m^{1/3}}{\epsilon}\right).
\]

The second term is smaller than the bound of Theorem~\ref{Thm:SSSPSparse} when $1/\epsilon = \tilde O(m^{13/6}n^{7/2})$. When $1/\epsilon = \tilde\Omega(m^{13/6}n^{7/2})$, the time bound of Theorem~\ref{Thm:SSSPSparse} becomes $\tilde O(m^{5/2}n^{7/2})$ which can be achieved with the $O(mD) = O(mn)$-time data structure of Even and Shiloach.

It follows that we only need to focus on the first term above. This term is better than the $O(mD)$ bound of Even and Shiloach when $D = \tilde\Omega(n^{7/8}/(\epsilon^{3/4}m^{1/8}))$, giving a time bound of $\tilde O((mn)^{7/8}/\epsilon^{3/4})$. The same time bound is obtained by applying the algorithm of Even and Shiloach when $D = \tilde O(n^{7/8}/(\epsilon^{3/4}m^{1/8}))$.

We now verify that parameters are set in their required ranges. Assuming $D = \tilde\Omega(n^{7/8}/(\epsilon^{3/4}m^{1/8}))$ (since we apply the algorithm of Even and Shiloach otherwise), we get that $\rho = \Omega(1)$ which satisfies the requirement that $\rho\geq 1$. Since $D \leq n$, we get $\rho = \tilde O(n)$ which satisfies the requirement $\rho\leq n$.

Corollary~\ref{Cor:LowDiamDecompObl} requires $1\leq d\leq n$. Plugging in the expression for $\rho$ in the expression for $d$ gives $d = \tilde\Theta(D^{1/9}n^{7/9}/(\epsilon^{2/3}m^{1/9}))$. Since we may assume that $D = \tilde\Omega(n^{7/8}/(\epsilon^{3/4}m^{1/8}))$, we get $d = \tilde\Omega(n^{7/8}/(\epsilon^{3/4}m^{1/8}))$ which satisfies the requirement that $d\geq 1$. Since $D\leq n$, we get $d = \tilde O(n^{8/9}/(\epsilon^{2/3}m^{1/9}))$ which satisfies $d\leq n$ when $1/\epsilon = \tilde O((mn)^{1/6})$. Note that when $1/\epsilon = \tilde\Omega((mn)^{1/6})$, the time bound $\tilde O((mn)^{7/8}/\epsilon^{3/4})$ becomes $\tilde O(mn)$ which can be obtained with the algorithm of Even and Shiloach.

Since $D' = \tilde\Theta((Dn/m)^{1/3})$ and since we may assume that $D = \tilde\Omega(n^{7/8}/(\epsilon^{3/4}m^{1/8}))$, we can ensure the requirement $D'\geq 1$ by requiring that $Dn/m = \tilde\Omega(1)$ which is ensured by requiring that $n^{15/8}/m^{9/8} = \tilde\Omega(\epsilon^{3/4})$, i.e.~that $1/\epsilon = \tilde\Omega(m^{3/2}/n^{5/2})$. Since $\epsilon \leq 1$, this requirement is satisfied when $m = \tilde O(n^{5/3})$. When $m = \tilde\Omega(n^{5/3})$, we can apply the data structure of Theorem~\ref{Thm:SSSPDense} which gives a time bound of $\tilde O(\sqrt m n^{3/2}/(\epsilon^{3/2}))$. This is better than $\tilde O((mn)^{7/8}/\epsilon^{3/4})$ when $m = \tilde\Omega(n^{5/3}/\epsilon^2) = \tilde\Omega(n^{5/3})$.

The final requirement that $\delta\in\mathbb R_+$ is clearly satisfied. This shows the update time bound of Theorem~\ref{Thm:SSSPSparse}.

\section{Generalization to weighted graphs}\label{sec:WeightedGraphs}
In this section, we extend our data structures from the main part of the paper to the case where $G$ is a weighted graph where the ratio between the smallest and largest edge weight is at most some given value $W\geq 1$. We only consider distance queries; extending to path queries follows the same approach as in Section~\ref{subsec:PathReport} and the end of Section~\ref{subsec:TimeSSSPSparse}.

\subsection{Low-diameter decomposition in weighted graphs}\label{subsec:LowDiamDecompW}
We start with Theorem~\ref{Thm:SSSPAdaptive}. The first step is to generalize the results from Section~\ref{subsec:SubProc} to weighted graphs:
\begin{Lem}\label{Lem:WThinLayer}
There is a deterministic algorithm \texttt{WThinLayer} which, on input $(H,r,d_1,d_2)$ where $r$ is a vertex in digraph $H = (V,E)$ with edge weights of at least $1$ and less than $\omega\in\mathbb N$ and where $d_1$ and $d_2$ are integers divisible by $\omega$ with $d_2 - d_1\geq 2\omega\lg(|(V|)$, outputs a subset $S$ of $V(H)$ for some $d$ divisible by $\omega$ with $d_1 < d < d_2$ and $d - \omega < d_H(r,s) \leq d$ for all $s\in S$. If $|\{v\in V\vert d_H(r,v)\leq d_1\}|\geq k$ and $|\{v\in V\vert d_H(r,v)\geq d_2\}|\geq k$ for $k\in\mathbb N$ then $S$ is a $q$-quality separator $S$ in $H$ with $q = (d_2-d_1)/(2\omega\lg|V(H)|)$ and each SCC of $H\setminus S$ has vertex size at most $|V| - k$. The running time is $O(|C|\log|V(H)| + \sum_{v\in C}\deg_H(v)) = O(|V|\log|V(H)| + |E|)$ where $C = \{v\in V\vert d_H(r,v)\leq d\}$.
\end{Lem}
\begin{proof}
For all $i\geq\omega$ divisible by $\omega$, define the $i$th layer as $L_i = \{v\in V\vert i - \omega < d_H(r,v) \leq i\}$. Algorithm \texttt{WThinLayer}$(H,r,d_1,d_2)$ finds $S$ as follows. Grow a shortest path tree $T$ from $r$ with Dijkstra's algorithm in $H$ until a layer $L_d$ is found with $d_1 < d \leq d_2$ such that both $|L_d|\leq\sum_{j < d}|L_j|/q$ and $|L_d|\leq\sum_{j > d}|L_j|/q$. Then output $S = L_d$.

To show correctness, suppose first for contradiction that there is no $d$ with $d_1\leq d\leq d_2$, $|L_d|\leq\sum_{j < d}|L_j|/q$, and $|L_d|\leq\sum_{j > d}|L_j|/q$. Then for $i = d_1,d_1+\omega,d_1+2\omega,\ldots,d_2$, either $\sum_{j\leq i}|L_j| > \sum_{j < i}|L_j|(1+1/q)$ or $\sum_{j\geq i}|L_j| > \sum_{j > i}|L_j|(1+1/q)$. This implies that either $\sum_{j\leq d_2}|L_j| > (1+1/q)^{(d_2 - d_1)/(2\omega)}$ or $\sum_{j\geq d_2}|L_j| < |V(H)|/(1+1/q)^{(d_2 - d_1)/(2\omega)}$. We have
\[
d_2 - d_1\geq 2\omega\lg(|(V(H)|) \Rightarrow q\geq 1\Rightarrow (1+1/q)^q\geq 2
\Rightarrow (1+1/q)^{(d_2 - d_1)/(2\omega)}\geq 2^{(d_2 - d_1)/(2\omega q)} = |V(H)|.
\]
This is a contradiction since $\sum_{j\leq d_2}|L_j| < |V(H)|$ and $\sum_{j\geq d_2}|L_j|\geq k\geq 1$.

We have shown that $|L_d|\leq\sum_{j < d}|L_j|/q$ and $|L_d|\leq\sum_{j > d}|L_j|/q$ for some $d$ with $d_1\leq d\leq d_2$. Since every edge of $H$ has weight less than $\omega$, there are no edges of $H$ from $\cup_{j < d} L_j$ to $\cup_{j > d} L_j$ so $S = L_d$ is a $q$-quality separator of $H$. Furthermore, both $\cup_{j\leq d} L_j$ and $\cup_{j\geq d} L_j$ contain at least $k$ vertices so every SCC of $H\setminus S$ has size at most $|V(H)| - k$.

To show the running time, growing the shortest path tree up to layer $d$ with Dijkstra's algorithm can be done in $O(|C|\log|V(H)| + \sum_{v\in C}\deg_H(v))$ time. Keeping track of the sums $\sum_{j < d'}|L_j|/q$ and $\sum_{j\leq d'}|L_j|/q$ for $0\leq d'\leq d$ can be done in additional $O(\sum_{j\leq d}|L_d|)$ time which is $O(1 + \sum_{v\in C}\deg_H(v))$ since each vertex of $\cup_{j\leq d}L_j\setminus\{r\}$ has at least one ingoing edge, namely from its parent in the shortest path tree. Since $\sum_{j > d}|L_j|/q = n - \sum_{j\leq d}|L_j|/q$, the termination criterion can be tested within the desired time bound.
\end{proof}
\begin{Cor}\label{Cor:WSeparator}
There is a deterministic algorithm \texttt{WSeparator} which takes as input a pair $(H,d)$ where $H$ is a digraph in which every edge has weight at least $1$ and less than $\omega\in\mathbb N$ and where $H$ has diameter at least $d\geq 2\omega\lg(|V(H)|$ where $d$ is divisible by $\omega$. In $O(|V(H)|\log|V(H)| + |E(H)|)$ time, the algorithm outputs a $q$-quality separator in $H$ with $q = d/(2\omega\lg|V(H)|)$.
\end{Cor}
\begin{proof}
Set $d_1 = 0$, $d_2 = d$ and $k = 1$ in Lemma~\ref{Lem:WThinLayer}.
\end{proof}
\begin{Lem}\label{Lem:WPartition}
There is a deterministic algorithm \texttt{WPartition} which, on input $(H,d)$ where $H = (V,E)$ is a digraph with edge weights of at least $1$ and less than $\omega\in\mathbb N$ and where $d\in\mathbb N$ is divisible by $\omega$, outputs a set $S\subseteq V$ such that each SCC of $H\setminus E(S)$ has diameter at most $d$. Letting $\mathcal C$ be the collection of these SCCs, $|S|\leq \frac{4\omega\lg|V|}d \sum_{C\in\mathcal C}|C|(\lg|V| - \lg|C|)$ and the running time is $O((|V|\log|V| + |E|)(1+\frac \omega d\sum_{C\in\mathcal C}|V(C)|(\log|V| - \log|V(C)|)))$.
\end{Lem}
\begin{proof}
\texttt{WPartition}$(H,d)$ does the following. First, it computes the SCCs of $H$. For each SCC $C$, it picks an arbitrary vertex $r$ and computes distances in $C$ from $r$ to all vertices of $C$ and distances in $C$ from all vertices of $C$ to $r$ using Dijkstra's algorithm. If all distances found are at most $d/2$, \texttt{Partition} finishes processing $C$. Otherwise, it calls \texttt{WSeparator}$(C,d/2)$ from Corollary~\ref{Cor:WSeparator}; let $S_C$ be the separator found. Then \texttt{Partition} finds the SCCs of $C\setminus S_C$ and recurses on $(C',d)$ for each such SCC $C'$. The output $S$ of \texttt{WPartition}$(H,d)$ is the union of all separators found by \texttt{WSeparator} in the current and in all recursive calls.

We start by showing correctness. At termination, each SCC $C$ of $H\setminus E(S)$ has diameter at most $d$ since either $C$ consists of a single vertex of $S$ or the algorithm has certified that there is an $r\in V(C)$ such that $B_{\mathit{in}}(r,C,d/2) = B_{\mathit{out}}(r,C,d/2) = C$.

To bound $|S|$, consider the call made to \texttt{WSeparator}$(C,d/2)$, giving a separator $S_C$ in $H$. By Corollary~\ref{Cor:WSeparator}, each SCC of $C\setminus S_C$ contains at most $|V(C)| - d|S_C|/(4\omega\lg|V(C)|)$ vertices. In particular, this holds for the at most one such SCC having size greater than $|V(C)/2$. Thus, there is a set $W$ of at least $d|S_C|/(4\omega\lg|V(C)|)$ vertices of $C$ belonging to SCCs of $C\setminus S_C$ of size at most $|V(C)|/2$. We can pay for the separator size $|S_C|$ by letting each such vertex pay at most $4\omega\lg|V(C)|/d$.

Over all recursive calls, a vertex $v\in C$ where $C\in\mathcal C$ is charged at most $\lg(|V|/|V(C)|)$ times. Hence the size of $S$ is at most $\sum_{C\in\mathcal C}4\omega|C|(\lg|V|)(\lg(|V|/|V(C)|))/d$, as desired.

It remains to bound the running time. Consider any recursive call \texttt{WPartition}$(H',d)$. Let $E_1$ be the set of edges of $E(H')$ with both endpoints in the same SCC of $H'$ and let $E_2 = E(H')\setminus E_1$. Excluding the time spent in recursive calls, \texttt{WPartition}$(H',d)$ takes $O(|V(H')|\log|V(H')| + |E_1| + |E_2|)$ time. The sum of $|E_2|$ over all recursive calls \texttt{WPartition}$(H',d)$ is $O(|E|)$. The sum of $O(|V(H')|\log|V(H')| + |E_1|)$ over all recursive calls \texttt{WPartition}$(H',d)$ in which \texttt{WSeparator} is not called is $O(|V|\log|V| + |E|)$. This is within the time bound of the lemma.

The time not yet accounted for is dominated by the total time spent in calls to \texttt{WSeparator}. By Corollary~\ref{Cor:WSeparator}, each call \texttt{WSeparator}$(C,d/2)$ takes $O(|V(C)|\log|V(C)| + |E(C)|)$ time. We use the same charging scheme as above but distribute a cost of $O(|V(C)|\log|V(C)| + |E(C)|)$ rather than $|S_C|$ among the vertices of $W$ where $W$ is defined as above. Since $C$ is strongly connected, we have $|S_C|\geq 1$ and thus $|W|\geq d/(4\omega\lg|V(C)|)$ so each vertex of $W$ is charged a cost of no more than $O((|V(C)|\log|V(C)| + |E(C)|)(\omega\log|V(C)|)/d) = O((|V|\log|V| + |E|)(\omega\log|V|)/d)$.

It follows that for each $C\in\mathcal C$ and each $v\in C$, $v$ is charged a total cost of $O((|V|\log|V| + |E|)(\omega\log|V|)(\log(|V|/|V(C)|))/d)$ over the entire execution of \texttt{Partition}$(H,d)$. The sum of this over all $C\in\mathcal C$ and all $v\in C$ is within the time bound of the lemma.

\end{proof}

We can now get the generalization of Theorem~\ref{Thm:LowDiamDecomp} to weighted graphs:
\begin{theorem}\label{Thm:WLowDiamDecomp}
Let $G = (V,E)$ be a graph with integer edge weights of at least $1$ and less than $\omega\in\mathbb N$ undergoing edge deletions, let $m = |E|$ and $n = |V|$, and let integers $0 < d_1 < d_2\leq n$ divisible by $\omega$ be given with $d_2 - d_1\geq 2\omega\lg n$. Then there is a Las Vegas data structure which maintains a pair $(S,\mathcal V)$ where $S\subseteq V$ is a growing set and where $\mathcal V$ is the family of vertex sets of the SCCs of $G\setminus E(S)$ such that at any point, all these SCCs have diameter at most $d_2$ and $|S| = \tilde{O}(n\omega/d_1)$.

After the initialization step, the data structure outputs the initial pair $(S,\mathcal V)$. After each update, the data structure outputs the set $S'$ of new vertices of $S$ where $S'\subseteq V'$ for some $V'\in\mathcal V$. Additionally, it updates $\mathcal V$ by replacing at most one $V'\in\mathcal V$ by the vertex sets $W_1,\ldots,W_p$ of the new SCCs of $G\setminus E(S)$ where $|W_i|\leq\frac 1 2 |V'|$ for $i = 1,\ldots,p-1$. Pointers to $W_1,\ldots,W_p$ are returned.

The total expected time is $\tilde O(m\sqrt n + mn\omega/d_1 + mn\omega d_1/(d_2 - d_1))$ and the data structure works against an adaptive adversary.
\end{theorem}
The data structure of Theorem~\ref{Thm:WLowDiamDecomp} is very similar to that of Theorem~\ref{Thm:LowDiamDecomp} so we only point out the differences here.  ES-structures are implemented using King's generalization to weighted graphs~\cite{King99}. Furthermore, the data structure checks if most $(d_2 - d_1)/(2\omega)$ vertices are missing from the in-/out-tree rather than $(d_2 - d_1)/2$.

The correctness proof follows using the same observations as in Section~\ref{subsec:CorrectnessLowDiamDecomp} together with the fact that a path containing at most $(d_2-d_1)/(2\omega)$ vertices has weight less than $d_2 - d_1$. The time bound follows using the proof of Section~\ref{subsec:TimeLowDiamDecomp} with $(d_2-d_1)/2$ replaced by $(d_2-d_1)/(2\omega)$.

We now get the following corollary whose proof is the same as that of Corollary~\ref{Cor:LowDiamDecomp} but applying Theorem~\ref{Thm:WLowDiamDecomp} instead of Theorem~\ref{Thm:LowDiamDecomp}:
\begin{Cor}\label{Cor:WLowDiamDecomp}
Let $G = (V,E)$ be a graph with integer weights of at least $1$ and less than $\omega\in\mathbb N$ undergoing edge deletions, let $m = |E|$ and $n = |V|$, and let integers $0 < d_1 < d_2\leq n$ divisible by $\omega$ be given with $d_2 - d_1\geq 2\omega\lg n$. Then there is a Las Vegas data structure which maintains pairwise disjoint growing subsets $S_0, S_1,\ldots, S_{\lceil\lg d_1\rceil}$ of $V$ and a family $\mathcal V$ of subsets of $V$ with the following properties. For $i = 0,\ldots,\lceil\lg d_1\rceil$, let $G_i = G\setminus(\cup_{j = 0}^i E(S_j))$. Then over all updates, $\mathcal V$ is the family of vertex sets of the SCCs of $G_{\lceil\lg d_1\rceil}$ and for $i = 0,\ldots,\lceil\lg d_1\rceil$,
\begin{enumerate}
\item each SCC of $G_i$ of vertex size at most $n/2^i$ has diameter at most $d_2/2^i$,
\item if $i > 0$, every vertex of $S_i$ belongs to an SCC of $G_{i-1}$ of vertex size at most $n/2^i$, and
\item $|S_i| = \tilde O(n\omega 2^i/d_1)$.
\end{enumerate}
After the initialization step, the data structure outputs the inital sets $S_0, S_1,\ldots, S_{\lceil\lg d_1\rceil}$ and pointers to the sets of $\mathcal V$. After each update, the data structure outputs the new vertices of $S_0,\ldots,S_{\lceil\lg d_1\rceil}$. Additionally, it updates $\mathcal V$ by replacing at most one $V'\in\mathcal V$ by the vertex sets $W_1,\ldots,W_p$ of the new SCCs of $G_{\lceil\lg d_1\rceil}$ where $|W_i|\leq\frac 1 2 |V'|$ for $i = 1,\ldots,p-1$. Pointers to both the old set $V'$ and to the new sets $W_1,\ldots,W_p$ are returned. 

The total expected time is $\tilde O(m\sqrt n + mn\omega/d_1 + mn\omega d_1/(d_2 - d_1))$ and the data structure works against an adaptive adversary.
\end{Cor}

\subsection{Showing Theorem~\ref{Thm:SSSPAdaptive} for weighted graphs}
Now, we are ready to show Theorem~\ref{Thm:SSSPAdaptive} for weighted graphs. We may assume w.l.o.g.~that $W$ and each edge weight is a power of $(1+\epsilon)$ and that the smallest edge weight of $G$ is $1$. Let $D'$ be a given power of $2$ between $1$ and $2W$. Using the same argument as in Section~\ref{subsec:ReductionD}, we only need to give a data structure with total expected update time $\tilde O(m^{2/3}n^{4/3}/\epsilon^{5/3})$ which can answer any intermediate query for $d_G(s,u)$ within an approximation factor of $(1+\epsilon)^{\Theta(1)}$, provided $D'\leq d_G(s,u) < 2D'$.

Let $w_{\min} = \epsilon D'/(n-1)$ and $w_{\max} = 2D'$. First, we observe that every edge of weight at least $w_{\max}$ can be removed from the initial graph $G$ since such an edge can never be part of a shortest path of weight less than $2D'$. Furthermore, every edge weight less than $w_{\min}$ can be rounded up to $w_{\min}$; this follows since any shortest path $P$ in any intermediate graph contains at most $n-1$ edges and hence the rounding of edge weights increases the weight of $P$ by at most $\epsilon D'\leq \epsilon w(P)$.

Now, every edge weight in $G$ is between $w_{\min}$ and $w_{\max}$. Dividing all these edge weights by $\epsilon w_{\min}$ does not introduce any approximation error in the distance estimates computed and we now have that all edge weights are between $1/\epsilon$ and $w_{\max}/(\epsilon w_{\min}) = \Theta(n/\epsilon^2)$. Finally, rounding all edge weights up to the nearest integer increases the weight of each edge by a factor of at most $(1/\epsilon + 1)/(1/\epsilon) = (1 + \epsilon)$. Hence, we may assume that every edge weight in $G$ is an integer between $1/\epsilon$ and $\Theta(n/\epsilon^2)$. Shortest path distances in the modified graph $G$ are between $1/\epsilon^2$ and $D_{\max} = \Theta(D'/w_{\min}) = \Theta(n/\epsilon)$.

Let $G$ be the modified graph and let $D$ be a power of $2$ of value at most $D_{\max}$. We will describe a data structure \DSSSP which gives an approximation factor of $(1+\epsilon)^{\Theta(1)}$ of any distance $d_G(s,u)$, provided $D\leq d_G(s,u) < 2D$.

We modify the initialization step from Section~\ref{subsec:Init} slightly. Let $E_\omega$ be the set of edges of $G$ of weight less than $\omega$, for some parameter $\omega\in\mathbb N$ to be specified later. First, \DSSSP initializes \DSCC as an instance of the data structure of Corollary~\ref{Cor:WLowDiamDecomp} for the subgraph of $G\cap E_\omega$ of $G$. Since distances can now be up to $D_{\max}$, \DSSSP sets $k = \log_{1+\epsilon}(D_{\max})$ and initializes sets $E_0,\ldots,E_k$. Redefine $\eta$ as $\eta(x) = \lfloor\log_{1+\epsilon}(x/\tau+1)\rfloor$ and $\Delta_i = (1+\epsilon)^{i+2}\tau$ for $i = 0,\ldots,k$.

Every edge $(u,v)\in E\setminus E_\omega$ is added to $E_d$ where $d = \eta(r(C(u)),r(C(v)))$. Every edge $(u,v)\in E_\omega$ is added to $E_d$ where $d = \lceil\log_{1+\epsilon} w(u,v)\rceil$ where $w(u,v)\geq\omega$ is the weight of $(u,v)$ in $G$.

\DSSSP then sets up $\mathcal M$ and sets up \DES with distance threshold $2D(1+\epsilon)$ and weight functions defined by $w_M(i) = (1+\epsilon)^i$ and $W_M(i) = (1+\epsilon)^{i+1}$ for $i = 0,\ldots,k$.

Updates and queries are handled by \DSSSP as in Section~\ref{subsec:UpdateQuery}.

Observe that for each edge of $E_{\omega}$, its weight is preserved up to a factor of $(1+\epsilon)$ in \DES at initialization. When this is no longer the case, an \texttt{Increase}-operation is applied to the edge as described in Section~\ref{subsec:UpdateQuery}. Correctness now follows from Corollary~\ref{Cor:WLowDiamDecomp} and from arguments similar to those in Section~\ref{subsec:CorrectnessAdaptive} provided that the following constraint is satisfied:
\[
  d_2 + n^2\omega/(d_1\tau) = \tilde O(\epsilon D).
\]

We now analyze the running time. By Corollary~\ref{Cor:WLowDiamDecomp}, the total expected time to maintain \DSCC is $\tilde O(m\sqrt n + mn\omega/d_1 + mn\omega d_1/(d_2 - d_1))$. The rest of \DSSSP takes a total of $\tilde O(m + Dn\tau)$ as before plus the total time spent by \DES on scanning edges of $M$ corresponding to edges of $E_\omega$. The latter takes $\tilde O(mD/\omega)$ time since such an edge has weight at least $\omega$ at any given time and the number of times it is scanned is therefore $O(D/\omega)$. Requiring that $d_2\geq 2d_1$, we get a total time bound of
\[
  \tilde O(m\sqrt n + mn\omega/d_1 + mn\omega d_1/d_2 + Dn\tau + m/(\omega\epsilon)).
\]

We maximize $d_2$ and $\tau$ to satisfy the constraint, getting $d_2 = \tilde\Theta(\epsilon D)$ and $\tau = \tilde\Theta(n^2\omega/(\epsilon d_1 D))$. The time bound becomes
\[
  \tilde O(m\sqrt n + mn\omega/d_1 + mn\omega d_1/(\epsilon D) + n^3\omega/(\epsilon d_1) + mD/\omega)
= \tilde O(m\sqrt n + mn\omega d_1/(\epsilon D) + n^3\omega/(\epsilon d_1) + mD/\omega).
\]
We set $d_1 = \tilde\Theta(\sqrt{Dn^2/m}) = \tilde\Theta(n\sqrt{D/m})$ and get a time bound of
\[
  \tilde O(m\sqrt n + \sqrt m n^2\omega/(\epsilon\sqrt D) + mD/\omega).
\]
Next, we set $\omega = \tilde\Theta(m^{1/4}D^{3/4}\sqrt\epsilon/n)$ and get the time bound
\[
  \tilde O(m\sqrt n + m^{3/4}nD^{1/4}/\sqrt\epsilon)
= \tilde O(m\sqrt n + m^{3/4}n^{5/4}/\epsilon^{3/4})
= \tilde O(m^{3/4}n^{5/4}/\epsilon^{3/4}).
\]
Ensuring the constraint $d_2\geq 2d_1$ is ensured as in Section~\ref{subsec:TimeAdaptive} at the cost of an additional time bound of $\tilde O(n^2/\epsilon^2)$.

Multiplying the running time by $O(\log W)$ gives the time bound for weighted graphs in Theorem~\ref{Thm:SSSPAdaptive}.

\subsection{Showing Theorem~\ref{Thm:SSSPDense} for weighted graphs}
To get the result of Theorem~\ref{Thm:SSSPDense} for weighted graphs, we can use the results of Section~\ref{subsec:LowDiamDecompW} to get a weighted version of Corollary~\ref{Cor:LowDiamDecompObl} in which each set $S_i$ has size $|S_i| = \tilde O(n2^i\omega/d)$ in expected time $\tilde O(md + mn\omega/d)$. The total time bound for \DSSSP in this case is
\[
  \tilde O(md + mn\omega/d + Dn\tau + mD/(\omega\epsilon))
\]
under the constraint $d + n^2\omega/(\tau d) = \tilde O(\epsilon D)$. Setting $\tau = \tilde\Theta(n^2\omega/(\epsilon dD))$ simplifies the constraint to $d = \tilde O(\epsilon D)$ and we get a time bound of
\[
  \tilde O(md + mn\omega/d + n^3\omega/(\epsilon d) + mD/(\omega\epsilon))
= \tilde O(md + n^3\omega/(\epsilon d) + mD/(\omega\epsilon)).
\]
Next, we set $\omega = \tilde\Theta(\sqrt{mDd}/n^{3/2})$ and get a time bound of
\[
  \tilde O(md + \sqrt mn^{3/2}\sqrt D/(\epsilon\sqrt d))).
\]
The optimal choice for $d$ is $d = \tilde\Theta(nD^{1/3}/(m^{1/3}\epsilon^{2/3}))$, provided $d = \tilde O(\epsilon D)$, i.e., provided that $D = \tilde\Omega(n^{3/2}/(\sqrt m\epsilon^{5/2}))$. In this case, we get a time bound of $\tilde O(m^{2/3}n^{4/3}/\epsilon^{5/3})$. Otherwise, the algorithm of King gives a time bound of $\tilde O(\sqrt mn^{3/2}/\epsilon^{5/2})$. Multiplying by $O(\log W)$ gives the bound for weighted graphs in Theorem~\ref{Thm:SSSPDense}.

\subsection{Showing Theorem~\ref{Thm:SSSPSparse} for weighted graphs}
Finally, to get the result of Theorem~\ref{Thm:SSSPSparse} for weighted graphs, we can use arguments similar to those above in order to get a time bound of
\[
  \tilde O(md + mn\omega/d + Dm/\omega + \delta m/\epsilon + Dn/(\epsilon D') + \delta Dn/(\epsilon(D')^3) + Dm/\rho)
\]
under the constraint $d + n^2\rho D'/(d\delta) = \tilde O(\epsilon D)$. After simplifying as in Section~\ref{subsec:OptParSSSPSparse}, we get a time bound of
\[
  \tilde O\left(dm + \frac{mn\omega}{d} + \frac{Dm}\omega + \frac{\rho m^{2/3}n^{7/3}}{\epsilon^2dD^{2/3}} + \frac{m^{1/3}}{\epsilon} + \frac{Dm}{\rho}\right).
\]
Picking $d = \tilde\Theta((\sqrt\rho n^{7/6})/(\epsilon D^{1/3}m^{1/6}))$ gives a time bound of
\[
  \tilde O\left(\frac{\sqrt\rho m^{5/6}n^{7/6}}{\epsilon D^{1/3}} + \frac{\epsilon D^{1/3}m^{7/6}\omega}{\sqrt\rho n^{1/6}} + \frac{Dm}\omega + \frac{m^{1/3}}{\epsilon} + \frac{Dm}{\rho}\right).
\]
We pick $\rho = \tilde\Theta((\epsilon^{2/3}D^{8/9}m^{1/9})/n^{7/9})$ to get a time bound of
\[
  \tilde O\left(\frac{D^{1/9}m^{8/9}n^{7/9}}{\epsilon^{2/3}} + \frac{\epsilon^{2/3}m^{10/9}n^{2/9}\omega}{D^{1/9}} + \frac{Dm}\omega + \frac{m^{1/3}}{\epsilon}\right).
\]
We set $\omega = \tilde\Theta(D^{5/9}/(m^{1/18}n^{1/9}\epsilon^{1/3}))$ and the time bound becomes
\[
  \tilde O\left(\frac{D^{1/9}m^{8/9}n^{7/9}}{\epsilon^{2/3}} + D^{4/9}m^{19/18}n^{1/9}\epsilon^{1/3} + \frac{m^{1/3}}{\epsilon}\right).
\]
It is easy to see that the first term dominates the second for any choice of $D$ and we get the same bound as we did in the unweighted case in Section~\ref{subsec:OptParSSSPSparse}. This shows Theorem~\ref{Thm:SSSPSparse} also for weighted graphs.
\end{document}